\pdfoutput=1
\documentclass[11pt]{article}

\usepackage[T1]{fontenc}
\usepackage[utf8]{inputenc}
\usepackage{lmodern}
\usepackage[margin=1in]{geometry}

\usepackage{amsmath}
\usepackage{amssymb}
\usepackage{amsthm}
\usepackage{bbm}
\usepackage{graphicx}
\usepackage{xcolor}
\usepackage{float}
\usepackage{longtable}
\usepackage{multirow}
\usepackage{hhline}
\usepackage{adjustbox}
\usepackage{placeins}
\usepackage{etoolbox}
\usepackage{setspace}
\usepackage[authoryear,round]{natbib}
\usepackage{url}
\usepackage{hyperref}

\hypersetup{colorlinks=true,citecolor=blue,urlcolor=blue,linkcolor=blue}

\providecommand{\doi}[1]{doi:\,\url{https://doi.org/#1}}

\graphicspath{{./}}

\newcommand{\ind}{\mathbbm{1}}
\newcommand{\R}{\mathbb{R}}

\newcommand{\E}{\mathbb{E}}
\newcommand{\F}{\mathbb{F}}

\providecommand{\new}[1]{#1}

\theoremstyle{plain}
\newtheorem{theorem}{Theorem}
\newtheorem{lemma}{Lemma}
\newtheorem{corollary}{Corollary}
\newtheorem{proposition}{Proposition}

\theoremstyle{definition}

\newtheorem{property}{Property}

\newcommand{\PropertyLabel}{}

\newenvironment{PropertyNamed}[1]
{%
  \renewcommand{\PropertyLabel}{#1}%
  \begin{property}\itshape
}
{%
  \end{property}%
}

\title{\bfseries Nonparametric Inference Based on\\ Expected Order Statistics}

\author{
  Tommaso Lando\thanks{Corresponding author. Department of Economics,
  University of Bergamo, Via dei Caniana 2, 24127 Bergamo, Italy.
  \texttt{tommaso.lando@unibg.it}. ORCID: 0000-0003-4288-0264.}
  \and
  Lorenzo Tedesco\thanks{Department of Economics, University of Bergamo,
  Via dei Caniana 2, 24127 Bergamo, Italy. ORCID: 0000-0002-2273-3577.}
}

\date{\today}

\begin{document}

\maketitle

\begin{abstract}
\noindent In many statistical procedures, such as smoothing, resampling, or privacy-preserving release, it is convenient to transform the empirical distribution according to the task. We introduce an averaging transformation that assigns equal mass to $m$ estimated expected order statistics, where $m$ acts as a resolution parameter. This yields finite-sample properties and a tractable route to asymptotic analysis. The error of the transformed empirical distribution from its deterministic counterpart is bounded by the estimation error of the classical empirical distribution, in both Wasserstein and integral distance. Moreover, every $L$-functional of the transformed distribution corresponds to an $L$-functional of the classical one, with updated weights. We establish almost sure and weak convergence, asymptotic distributions for distance-based functionals, and bootstrap validity. Finally, we obtain a differentially private release with explicit accuracy bounds, and compare it with private histograms in simulations and data.

\end{abstract}

\section{Introduction}
Consider the classical setting in which $X$ is a random variable with cumulative distribution function (CDF) $F$, and we observe an \textit{iid} sample of size $n$ from $F$, namely $X_1,\dots,X_n$. A simple yet central tool in statistics is the empirical CDF (ECDF), defined by
$
\mathbb{F}_n(x)=\frac{1}{n}\sum_{i=1}^n \ind\{X_i\le x\},
$
where $\ind\{\cdot\}$ denotes the indicator function. This estimator enjoys many fundamental properties and is a cornerstone of nonparametric inference. However, as a distributional representation, the ECDF has a fixed resolution: the number and the location of its support points are determined by the sample. In some tasks, it may be more convenient to use a representation in which the resolution is chosen separately from the sample size, for example, using a different number of support points, or a continuous representation. Kernel and Bernstein-polynomials methods transform the ECDF into a continuous version, with resolution controlled by smoothing parameters \citep{Azzalini1981,BabuCantyChaubey2002}. This introduces additional issues, such as bandwidth selection, kernel specification, support corrections, and an asymptotic theory that depends substantially on these choices (see, for example, \cite{mason2011,mason2013} for kernel CDF estimators, and \cite{einmahl2005uniform,stupfler} for kernel density estimators). Smoothed bootstrap methods resample from a smoothed version of the ECDF than from the ECDF itself \citep{SilvermanYoung1987}. Histogram-based methods replace the \(n\) empirical support points of the ECDF with a fixed number of bins, allowing a resolution different from the sample size. This represents a standard method for differentially private distribution release, where the bin frequencies are randomly perturbed in order to guarantee privacy \citep{WassermanZhou2010,DworkRoth2014}. For alternative approaches to differentially private distributional summaries and CDF estimation, see, for example, \citet{GillenwaterJosephKulesza2021}, \citet{TaoSarwate2026}, and \citet{BarczewskiMawassRamon2025}.

Following a different approach, this paper introduces a transformation of the ECDF that: (i) preserves, or modifies in a transparent way, some key finite-sample features of the ECDF; (ii) is closer to its deterministic counterpart than the ECDF is to \(F\); (iii) admits a rich and tractable asymptotic theory; and (iv) leads to a differentially private CDF release with explicit accuracy bounds. These properties are established in the paper. Before presenting the results, we describe the heuristic that led to the construction.

Denoting by $X_{1:n}\le...\le X_{n:n}$ the order statistics of the sample, one can equivalently write the ECDF as $ \mathbb{F}_n(x) = \frac{1}{n}\sum_{i=1}^n {\ind}\{X_{i:n}\le x\}.$ However, each order statistic $X_{i:n}$ is itself a random variable, with finite expectation $\mu_{i:n} = \mathbb{E}[X_{i:n}]$ whenever $\mathbb{E}|X|<\infty$. This suggests a natural question: what happens to the ECDF if we replace each random order statistic $X_{i:n}$ by an estimator $\hat{\mu}_{i:n}$ of its expectation? In other words, instead of putting mass $1/n$ at the observed order statistics $X_{i:n}$, we may consider an averaged version of the ECDF, that puts mass $1/n$ at the estimated mean order statistics $\hat{\mu}_{1:n},\dots,\hat{\mu}_{n:n}$. More generally, given a positive integer $m$, possibly different from the actual sample size $n$, we can construct a CDF that puts mass $1/m$ at estimators of the expectations of the order statistics that would occupy the ranks $j=1,\dots,m$ in a hypothetical sample of size $m$. This leads to a family of random distribution functions, indexed by $m$, of the form
\[
\mathbb{F}_{n,m}(x) = \frac{1}{m}\sum_{j=1}^m {\ind}\{\hat{\mu}_{j:m}\le x\}, \qquad x\in\mathbb{R},
\]
where $\hat{\mu}_{j:m}$ is a suitable estimator, based on the sample $X_1,\dots,X_n$, of the expected $j$-th order statistic in a sample of size $m$. We propose $\mathbb{F}_{n,m}$ as a family of transformed ECDFs.

To the best of our knowledge, $\F_{n,m}$ is a new object. The only related construction we are aware of is due to \citet{hoeffding1953}, who considered a deterministic version of $\mathbb{F}_{n,m}$, namely the family of CDFs that assigns mass $1/m$ to the true expected order statistics $\mu_{j:m}$. For this reason, we call this deterministic counterpart the Hoeffding CDF (HCDF), and we call $\F_{n,m}$ the empirical Hoeffding CDF (EHCDF). \cite{hoeffding1953} proved that the HCDF converges to the baseline CDF $F$, and that any moment of that distribution converges to the corresponding moment of $F$, whenever it exists. The work of \cite{hoeffding1953} gave rise to some related literature, mainly focused on mathematical characterisations via expected order statistics \citep{arnold1975,okolewski2024}, or on extensions of the convergence result, but still in a deterministic framework \citep{borisov2014,borisov2016}.

The objective of this paper is to introduce and study the EHCDF family, and to show that it provides a useful framework for nonparametric inference, with a particularly natural application to differential privacy. In Section~\ref{sec2}, we develop the mathematical structure underlying the method by introducing the Hoeffding CDF operator and its quantile counterpart, and deriving their main properties. In Section~\ref{sec3}, we define the proposed estimators by applying these operators to the ECDF and establish several finite-sample properties. In particular, the Hoeffding CDF operator is non-expansive in $L^1$, so that the distance between the EHCDF and its deterministic counterpart (HCDF) is controlled by the $L^1$ error of the ECDF; moreover, any $L$-functional of the EHCDF can be expressed as an $L$-functional of the ECDF with modified weights. As a consequence, the EHCDF preserves location, while measures of dispersion and skewness are deterministically reduced according to the choice of $m$. In Section~\ref{sec:asymptotic}, we establish a.s. consistency in $L^p$ norm and in Wasserstein distance of order $p$, both for fixed $m$ and in the joint regime $n,m\to\infty$. In Section~\ref{sec5}, we study the associated quantile process, which provides a convenient route to asymptotic analysis. We establish weak convergence in $L^p(0,1)$, for every $p\in[1,\infty)$, when $m$ is fixed, and in $L^1(0,1)$ and $L^2(0,1)$, when $n,m\to\infty$, relying on the distributional conditions developed by \citet{beare2025necessary} and \citet{delbarrio}, respectively. Bootstrap validity is also established in $L^p(0,1)$, for fixed $m$, and in $L^1(0,1)$ in the joint regime. As a byproduct, asymptotic distributions for functionals of the quantile process, including distance-based functionals, are obtained in Section~\ref{sec6}. Finally, in Section~\ref{sec7}, we develop the differentially private version of the EHCDF and obtain explicit Wasserstein accuracy bounds. These bounds decompose the error into an approximation term, depending on the choice of \(m\), a sampling term, and a privacy term.
Section~\ref{sec8} contains the numerical analysis. We first compare the private EHCDF with private histograms for different choices of the public clipping interval \([a,b]\). At the default resolutions, the private EHCDF has a smaller
\(W_1\) error in most of the settings considered. When both methods are optimised over prespecified resolution grids (an ex post comparison that does not provide
an implementable tuning rule), the histogram can be preferable when $b-a$ is smaller, whereas the EHCDF becomes more accurate as
\(b-a\) increases. Heavier tails and a point mass at zero further shift the
comparison in favour of the EHCDF. We finally compare the two methods on credit-card expenditure data. Proofs and additional simulation
results are reported in the Appendix. In particular, a numerical study of
(non-private) estimation error reported there shows that, in the settings considered, the EHCDF
often improves on the ECDF while remaining competitive with kernel-based
alternatives and exhibiting stable behaviour across different distributional
settings.

\section{The Hoeffding operators: definitions and main properties}\label{sec2}
\subsection{Notations}
We start with some basic notations. Given some CDF $G$, we define the quantile function as the left-continuous generalized inverse $G^{-1}(p)=\inf\{x:G(x)\geq p\}$. We will denote by $\|\cdot\|_p$ the $L^p$ norm, with $p \in[1,\infty]$. Given CDFs $F$ and $G$, the Wasserstein distance of order $p\geq1$ between $F$ and $G$ is defined as $W_p(F,G)=\big(\int_0^1|F^{-1}(t)-G^{-1}(t)|^pdt\big)^{1/p}$. The symbols $\to_{a.s.},$ $\to_p$ denote convergence almost surely and in probability, respectively. We use $\to_d$ to denote convergence in distribution for random variables and $\rightsquigarrow$ to denote weak convergence in function spaces.
\subsection{The Hoeffding CDF operator}
We shall be concerned with some random variable $X$ with CDF $F$ and quantile function $Q=F^{-1}$. If $X$ has a density, it will be denoted by $f$. We will assume throughout the paper that $\E|X|<\infty$. We denote with $X_{j:m}$ the $j$-th order statistic, that is, the $j$-th smallest value in an  \textit{iid} sample of size $m$ drawn from $X$. It is well known that the CDF of $X_{j:m}$ is \( F_{B_{j:m}} \circ F\), where $F_{B_{j:m}}$ and $f_{B_{j:m}}$ denote the CDF and density function, respectively, of a beta distribution with parameters $j$ and $m-j+1$, $j=1,...,m$. Equivalently, an order statistic $X_{j:m}$ from $F$ has the same law of $F^{-1}(U_{j:m}),$ where $U_{j:m}$ denotes an order statistic from the uniform distribution on $[0,1].$

Given that the CDF of $X_{j:m}$ is $F_{B_{j:m}}\circ F$, we can express the expected order statistic $\E X_{j:m}$ by the following functional
\[
\mu_{j:m}(F)=\int_\R xdF_{B_{j:m}}\circ F(x)=\int_0^1 F^{-1}(p)dF_{B_{j:m}}(p).
\]
From the latter representation it is readily seen that the expected order statistics $\mu_{j:m}(F)$ belong to the family of $L$-functionals \citep{serfling}.

Denote by $\mathcal D$ the family of CDFs with finite mean, and by $\mathcal D_m$ the family of discrete CDFs with $m$ jumps. Define the \emph{Hoeffding CDF operator} $\mathcal H_m : \mathcal D \to \mathcal D_m$ as
\begin{align}\label{def:hoeffding_operator}
	\mathcal H_m(G)
	=
	\frac{1}{m}\sum_{j=1}^m \ind\{\mu_{j:m}(G) \le \cdot\}
	=: G_m,
	\qquad G \in \mathcal D.
\end{align}
This operator enjoys a fundamental mathematical property that will be used in the sequel. We prove that $\mathcal H_m$ is an $L^1$-non-expansive operator, that is, the $L^1$ distance between a pair of $\mathcal H_m$-transformations is always bounded by the $L^1$ distance between the original CDFs. More generally, we establish the following bound, which holds uniformly for all $p\in[1,\infty)$.

\begin{theorem}\label{thm:Hm_L1_contraction}
	Let $F,G\in\mathcal D$. Then,
$\|\mathcal H_m (F)-\mathcal H_m (G)\|_p^p \le \|F-G\|_1$, for every $p\in[1,\infty)$.
\end{theorem}

\noindent Many results derived in the sequel are based on the quantile-Hoeffding operator, defined in the next subsection.
\subsection{The quantile-Hoeffding operator}
Given a CDF $G\in\mathcal D$, its quantile function $G^{-1}$ belongs to $L^1(0,1)$. To study a quantile version of the Hoeffding operator, it is convenient to introduce the following class of step functions:
\[
\mathcal S_m := \big\{ h:(0,1]\to\mathbb R: h(u)=\sum_{j=1}^m a_j\,\ind\{u\in(\tfrac{j-1}m,\tfrac jm]\},\ a_j\in\mathbb R \big\}.
\]
We define the \emph{quantile-Hoeffding operator} as the map $
\mathcal I_m:L^1(0,1)\to \mathcal S_m\subset L^1(0,1),$ given by
\begin{align}\label{def:Im_operator}
	\mathcal I_m(h)(u)
	&:=
	\sum_{j=1}^m
	\left(
	\int_0^1 h(t)\, dF_{B_{j:m}}(t)
	\right)
	\ind\{u\in\left(\tfrac{j-1}{m},\, \tfrac{j}{m}\right]\},
	\qquad u \in (0,1].
\end{align}
That is, $\mathcal I_m(h)$ is the step function which takes the constant value $
\int_0^1 h(t)\,dF_{B_{j:m}}(t)$ on each interval
$
\left(\frac{j-1}{m},\frac{j}{m}\right].$ If $G\in\mathcal D$ and $G_m=\mathcal H_m(G)$, then
$
\mathcal I_m(G^{-1})=(\mathcal H_m(G))^{-1}=G_m^{-1}.$ Thus, $\mathcal I_m$ is exactly the quantile counterpart of the CDF operator $\mathcal H_m$.

The quantile-Hoeffding operator enjoys some important properties that will be used to derive the main convergence results in the following sections. It is immediate that $\mathcal{I}_m(h)$ is linear in $h$, as this follows directly from the linearity of the integral. Moreover, it is a bounded operator, as shown in the following result, which in particular implies its continuity.

\begin{lemma}\label{lemma:Im_contraction}
	For every fixed \(m\), the operator \( \mathcal I_m : L^p(0,1)\to L^p(0,1) \) is linear and bounded. Moreover, for all $ h\in L^p(0,1),$ $\|\mathcal I_m(h)\|_{p} \le \|h\|_{p}$, uniformly in $m$.
\end{lemma}
As a byproduct, for every $p\in[1,\infty)$ and every $F,G\in\mathcal D$ such that $F^{-1},G^{-1}\in L^p(0,1)$, it holds
\[
W_p(\mathcal H_m(F),\mathcal H_m(G))\le W_p(F,G).
\]
\noindent Lastly, the following theorem plays a key role, as it allows one to transfer convergence in \(L^p\)-norm from an empirical process to its transformed version under \(\mathcal{I}_m\).
\begin{theorem}\label{teo:Im_convergence}
	Let $p\in[1,\infty)$ and let $h_n,h\in L^p(0,1)$ be such that
	$
	\|h_n-h\|_p\to 0.$ Then
	$
	\|\mathcal I_m(h_n)-h\|_p \to 0$, as $m,n\to\infty.$
\end{theorem}
\section{The estimator: definitions and finite sample properties}\label{sec3}
In this paper we deal with the classic inference problem, and we assume that $F$ is unknown. By the plugin approach, we replace $F$ with the ECDF. Applying the Hoeffding operator of order $m$ to $\F_n$, we obtain a random CDF with $m$ jump points (excluding the case of a degenerate sample), represented as
\begin{align}\label{def:EHCDF}
	\mathcal H_m(\F_n)=\frac1m\sum_{j=1}^m \ind(\mu_{j:m}(\F_n)\le \cdot)=:\F_{n,m}.
\end{align}
We refer to $\F_{n,m}$ as the Empirical Hoeffding CDF (EHCDF) of order $m$. Basically, $\F_{n,m}$ puts mass $1/m$ at the (random) points $\mu_{j:m}(\F_n)$, hereafter abbreviated as $\hat{\mu}_{j:m}$, where $j=1,...,m$. In particular, each $\hat{\mu}_{j:m}$ is an $L$-estimator, namely, a linear combination of order statistics, given by
\begin{align*}
	\hat{\mu}_{j:m}=\int_\R xdF_{B_{j:m}}\circ \F_n(x)=\sum_{i=1}^n X_{i:n}(F_{B_{j:m}}(\tfrac in)-F_{B_{j:m}}(\tfrac{i-1}n)).
\end{align*}
These estimators are consistent, as we will discuss in the sequel. \cite{lando2025new} recently developed a class of nonparametric tests based on these estimators. We can also express $\hat{\mu}_{j:m}$ as a weighted average, in which order statistics $X_{i:n}$ from $F$ are scaled by the probability that an order statistic $U_{j:m}$ from the uniform distribution on $[0,1]$ belongs to $(\tfrac{i-1}n,\tfrac in]$,
$$\hat{\mu}_{j:m} =\sum_{i=1}^n X_{i:n}P(U_{j:m}\in(\tfrac{i-1}n,\tfrac{i}n]).$$

Since $\F_{n,m}=\mathcal H_m(\F_n)$ and $F_m=\mathcal H_m(F)$, the following are immediate consequences of Theorem~\ref{thm:Hm_L1_contraction} and Lemma~\ref{lemma:Im_contraction}, respectively.
\begin{corollary} For every $p\in[1,\infty)$, it holds \( \|\F_{n,m}-F_m\|_p^p \le \|\F_n-F\|_1\) and \( W_p(\F_{n,m},F_m)\leq W_p(\F_n,F).\)
\end{corollary}

The latter part of the result above shows that the Hoeffding estimator is non-expansive in Wasserstein distance. Hence, moving from the $\mathbb F_n$ to $\mathbb F_{n,m}$ cannot amplify perturbations measured in the quantile domain. In this sense, the EHCDF may be viewed as a robust version of the ECDF, as shown in the following corollary. Denote by $\mathcal X=(X_1,\dots,X_n)$ a sample of size $n$, and write $\mathbb F_n^\mathcal X$ and $\mathbb F_{n,m}^\mathcal X$ for the corresponding ECDF and EHCDF of order $m$, respectively.
\begin{corollary}\label{cor:data_perturbation}
Given two samples $\mathcal X$ and $\mathcal Y$ of size $n_1$ and $n_2$, respectively, for every $p\in[1,\infty)$,
\[
W_p(\mathbb F_{n_1,m}^\mathcal X,\mathbb F_{n_2,m}^{\mathcal Y}) \le W_p(\mathbb F_{n_1}^\mathcal X,\mathbb F_{n_2}^{\mathcal Y}).
\]
\end{corollary}

To give an initial intuition about the roles of $n$ and $m$, we provide a graphical illustration in Figure~\ref{fig:matrix_nm}, for three samples of size $n\in\{10, 20, 50\}$ from a standard normal distribution with $m\in\{10,20,50\}$. We plot the EHCDF against the ECDF and the true CDF. As $m$ increases, the EHCDF gets deterministically closer to the ECDF (as samples are fixed), which, in turn, stochastically approaches the true CDF as $n$ increases. This behaviour will be addressed more precisely in the sequel.

\begin{figure}[ht]
    \centering
    \includegraphics[width=\textwidth]{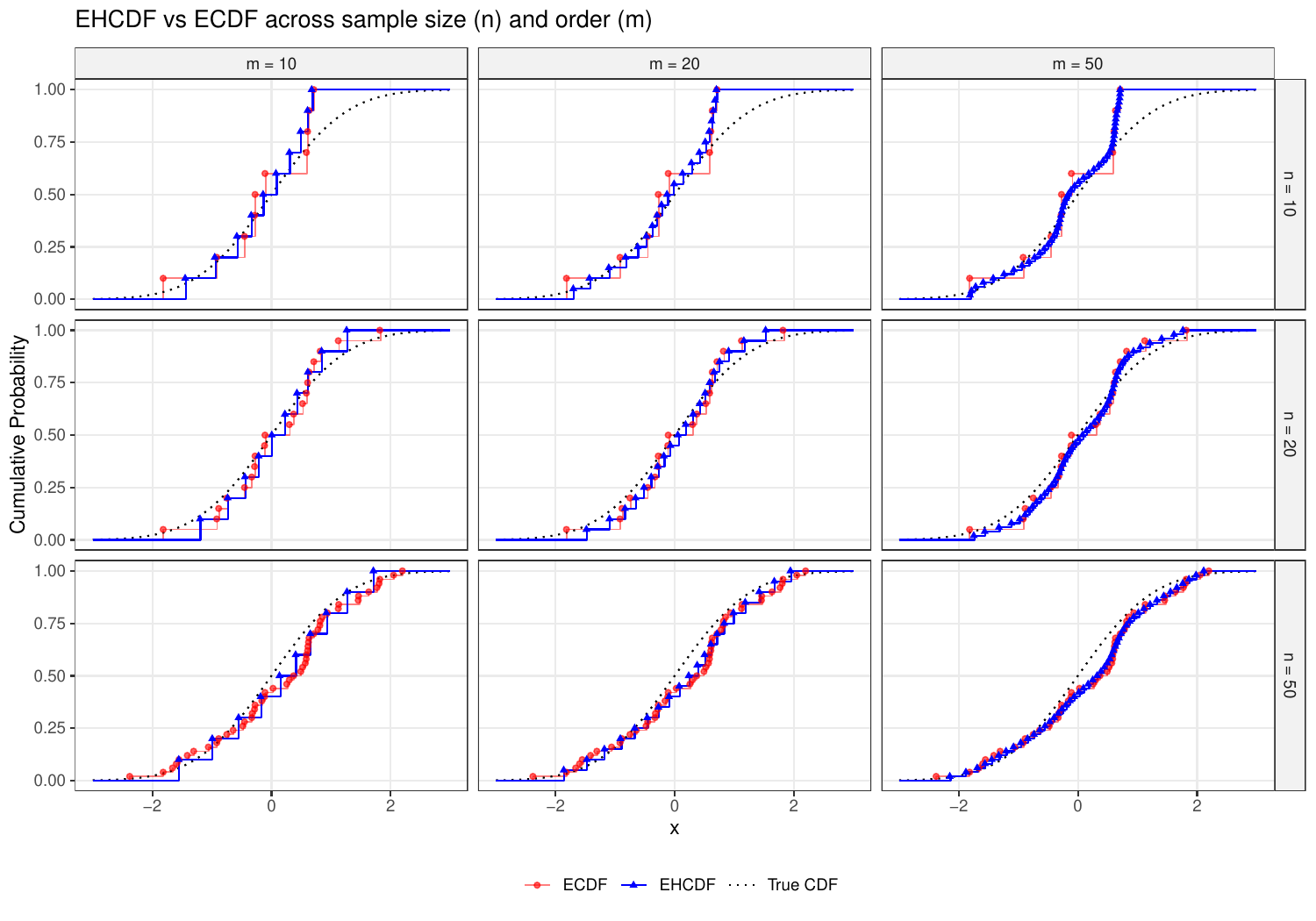}
    \caption{Graphical comparison of $\mathbb{F}_{n,m}$ (blue), $\mathbb{F}_n$ (red) and $F$ (dotted black).}
    \label{fig:matrix_nm}
\end{figure}
\subsection{Closure properties for \texorpdfstring{$L$}{L}-functionals}
In this section we investigate the behaviour of $L$-functionals under the Hoeffding CDF operator. In particular, we show that the Hoeffding operator enjoys a closure property on the class of $L$-functionals: composing an $L$-functional with $\mathcal H_m$ yields another $L$-functional with an updated weight, which depends explicitly on the original weight and on $m$.

Let $w$ be a weight function and define a general $L$-functional $T_w:\mathcal D\to\R$ with
\[
T_w(G):=\int_0^1 G^{-1}(u)\,w(u)\,du,\qquad G\in\mathcal D, \quad G^{-1}w\in L^1(0,1).
\]
Moreover, let \( c_j(w):=\int_{(j-1)/m}^{j/m} w(t)\,dt\) for \( j=1,\dots,m, \) and define the linear operator $P_m:L^1(0,1)\to L^1(0,1)$ by
\[
P_{m}[w](p):=\sum_{j=1}^m f_{B_{j:m}}(p)c_j(w)=\sum_{j=1}^m f_{B_{j:m}}(p)\int_{(j-1)/m}^{j/m} w(t)\,dt,\qquad p\in[0,1].
\]
An important sub-class of $L$-functionals are the so-called \textit{$L$-moments} \citep{Hosking1990}, defined as \[ \lambda_r(G):=\int_0^1 G^{-1}(p)\,w_{r-1}(p)\,dp,\qquad G\in\mathcal D,
\]
where $r$ is a positive integer and \[ w_r(u)=\sum_{k=0}^{r} (-1)^{\,r-k}\binom{r}{k}\binom{r+k}{k}\,u^{k}. \] The simplest example of $L$-moment is the mean functional, that is, $\lambda_1(G)=\mu(G),$ obtained when $w(p)=w_0(p)=1.$ Other important examples are the $L$-moments of orders $2,3,4$, corresponding to the weighting functions $w_1(p)=2p-1,$ $w_2(p)=6p^2-6p+1$ and $w_3(p)=20p^3-30p^2+12p-1$, respectively. In particular, $\lambda_2$ is a variability functional, generally referred to as $L$-scale \citep{Hosking1990}. It corresponds to half the \textit{mean absolute difference} (MAD), defined as $\operatorname{MAD}(G)=2\lambda_2(G)=\iint |x-y|dG(x)dG(y)=\E|Y_1-Y_2|$, where $Y_1$ and $Y_2$ are \textit{iid}  with CDF $G$. Differently, $\lambda_3,$ and $\lambda_4$ are related to skewness and kurtosis, respectively. In fact,
\[
\operatorname{Sk}(G)=\frac{\lambda_3(G)}{\lambda_2(G)} \quad\text{and}\quad \operatorname{Kurt}(G)=\frac{\lambda_4(G)}{\lambda_2(G)}
\]
are known as $L$-skewness and $L$-kurtosis \citep{Hosking1990}.

The following theorem establishes the closure of the family of $L$-functionals under the Hoeffding CDF operator $\mathcal H_m$.
\begin{theorem}\label{thm:Hoeffding_Lfunctionals}
	For every $G\in \mathcal D$, $ T_w\bigl(\mathcal H_m(G)\bigr)=T_{P_m[w]}(G).$
\end{theorem}

\noindent Since $\F_{n,m} = \mathcal{H}_m(\F_n)$, the following corollary holds.

\begin{corollary}\label{Lfunctionals}The following identity holds, $T_w\bigl(\F_{n,m}\bigr)=T_{P_m[w]}(\F_n).$ In particular:\begin{enumerate}
		\item $\mu(\F_{n,m})=\mu(\F_n)$;
		\item $\operatorname{MAD}(\F_{n,m})=\tfrac{m-1}m\operatorname{MAD}(\F_{n})$;
		\item $\operatorname{Sk}(\F_{n,m})=\tfrac{m-2}m\operatorname{Sk}(\F_{n})$;
		\item $\operatorname{Kurt}(\F_{n,m})=\tfrac{(m-2)(m-3)}{m^2}\operatorname{Kurt}(\F_{n})-\tfrac1{m^2}.$
	\end{enumerate}
\end{corollary}
Corollary~\ref{Lfunctionals} shows that $\F_{n,m}$ inherits a number of finite-sample identities from $\F_n$. In particular, $\F_{n,m}$ is mean-preserving. At the same time, measures of dispersion and shape are deterministically shrunk by factors depending only on $m$: the MAD is scaled by $(m-1)/m$, the $L$-skewness by $(m-2)/m$, while the $L$-kurtosis undergoes an affine transformation. Hence, for fixed $n$, smaller values of $m$ yield an estimator that is more concentrated (in MAD) and exhibits reduced skewness and tail heaviness (in $L$-moment ratios), whereas these effects vanish as $m$ grows.

\section{Convergence}\label{sec:asymptotic}
We now investigate the asymptotic properties of the estimator.  The structure of the section is related to the following simple observation. The error between the EHCDF and the true distribution may be decomposed into the stochastic error between $\F_{n,m}$ and $F_m$ and the deterministic error between $F_m$ and $F$, i.e.,
$$\F_{n,m}-F=(\F_{n,m}-F_m)+(F_m-F). $$
We separate the deterministic convergence $F_m\to F$ from the stochastic convergence $\F_{n,m}\to F_m$, and then combine the two to obtain consistency of $\F_{n,m}$.
\subsection{Deterministic convergence of \texorpdfstring{$F_m$}{Fm} to \texorpdfstring{$F$}{F}}
Recall the definition of the Hoeffding CDF operator in \eqref{def:hoeffding_operator} and the (deterministic) CDF $\mathcal H_m(F)=F_m$. The following fundamental result holds.

\begin{theorem}[\cite{hoeffding1953}]\label{thm:hoeffding1953}
	Let $g$ be a continuous real-valued function such that $|g|\le h$ where $h$ is convex and $\E h(X)<\infty$. Then
	$
	\lim_{m\to\infty}\frac1m\sum_{j=1}^m g(\mu_{j:m}) = \E g(X).
	$
\end{theorem}

As a consequence, taking $g(x)=\cos(tx)$ and $\sin(tx)$, the characteristic function of $F_m$ converges to that of the $F$, and hence
\begin{align}\label{eq:HCDF_pointwise}
	\lim_{m\to\infty} F_m = F,\quad \text{at every continuity point of $F$. }
\end{align}
By Polya's theorem, the convergence is uniform if $F$ is continuous. Differently, if $X$ is integrable, which is a necessary condition in our setting, $F_m$ converges not only pointwise to $F$ where the latter is continuous, as shown in \eqref{eq:HCDF_pointwise}, but also in $L^p$.
\begin{lemma}\label{lemma:deterministic_L1_convergence}
	If $\E|X|<\infty$, then \( \lim_{m\to\infty} \|F_m - F\|_p = 0. \)
\end{lemma}

With regard to the $L^1$ norm, we can establish the rate of convergence of $F_m$. This will be useful later on when addressing weak convergence. The result is obtained under a sufficient condition which will appear again in the next sections.
\begin{lemma}\label{lemma:HCDF_CDF_convergence_rate}
	If $\int_0^1 \sqrt{u(1-u)}dQ(u)<\infty$, then \( \| F_m - F \|_1 = O(m^{-1/2}) \).
\end{lemma}
\noindent We finally note that the results in this subsection can also be applied pathwise to $\F_{n,m}$. Fix some outcome $\omega$ in the sample space. Then $\F_n(\omega)$ is a deterministic CDF with bounded support, and
$\F_{n,m}(\omega)=\mathcal H_m(\F_n(\omega)).$
Applying Hoeffding's theorem to $\F_n(\omega)$, we obtain that $
\F_{n,m}(\omega)\to \F_n(\omega)$ at every continuity point of $F_n$, that is, for every $x\notin \{X_1(\omega),\ldots,X_n(\omega)\}$. Since $\F_n(\omega)$ has bounded support, this convergence also holds in $L^p$. We summarize this in the following lemma.

\begin{lemma}\label{lemma hoeffding empirical}
	For every fixed sample $x_1,...,x_n$, as $m\to\infty$, \( \F_{n,m}(x)\to \F_n(x)	\) for any $x\notin\{x_1,\ldots,x_n\}$. Moreover, for every $p\in[1,\infty)$, \(\lim_{m\to\infty}\|\F_{n,m}-\F_n\|_p=0.\)
\end{lemma}
\subsection{Stochastic convergence of \texorpdfstring{$\F_{n,m}$}{Fn,m} to \texorpdfstring{$F_m$}{Fm}}
We begin with a basic asymptotic property of the $L$-statistics underlying our construction, namely the strong consistency of the estimators $\hat{\mu}_{j:m}$. If $\E|X|<\infty$, Theorem 2.1 of \cite{zwet1980} immediately yields the following result.
\begin{proposition}
	If $\E|X|<\infty$, then $\hat{\mu}_{j:m}\to_{a.s.}\mu_{j:m}$ for every fixed $m$ and $j=1,\dots,m$.
\end{proposition}
\noindent Lemma~\ref{lemma hoeffding empirical} establishes the deterministic convergence of $\F_{n,m}$ to $\F_n$ in the $L^p$ norm as $m\to\infty$. We now study the stochastic convergence of the EHCDF to the HCDF as $n\to\infty$, with $m$ fixed. This yields convergence in $L^p$, but in general not in $L^\infty$. To see why, note that $\F_{n,m}$ and $F_m$ are step functions with the same jump size $1/m$, but with jump locations $\hat{\mu}_{j:m}$ and $\mu_{j:m}$, respectively. Although $\hat{\mu}_{j:m}\to_{a.s.}\mu_{j:m}$ for every $j$, as long as at least one jump location differs, the two functions exhibit a distance of size $1/m$ on a neighborhood of that jump. Hence the supremum norm of the difference does not, in general, converge to zero. On the other hand, for $n$ large enough, the jump locations become sufficiently close so that the distance cannot exceed a single jump height, and therefore $\|\F_{n,m}-F_m\|_\infty$ is eventually a.s. bounded by $1/m$.
\begin{lemma}\label{lemma:convergence_EHCDF_HCDF}
	Let $\E|X|<\infty$. For $n\to\infty,$ $||\F_{n,m}-F_m||_p\to_{a.s.}0$ for every $p\in[1,\infty)$. Moreover, for every $m$, there exists a.s. some $n_0$ such that, for $n>n_0$, $||\F_{n,m}-F_m||_\infty\leq1/m$, that is, $$P(\omega:\exists n_0(\omega)\text{ s.t. }\forall n>n_0,||\F_{n,m}-F_m||_\infty\leq1/m)=1.$$
\end{lemma}
\subsection{Consistency of \texorpdfstring{$\F_{n,m}$}{Fn,m}}
It is now relatively simple to combine the deterministic and stochastic convergence results to obtain the a.s. convergence of $\F_{n,m}$ to $F$ in the $L^p$ norm. We first obtain convergence as iterated limits. This holds in both ways when considering the $L^p$ norm, while it holds just for $n\to\infty$ and then $n\to\infty$ when considering the supremum norm.

\begin{lemma}\label{lemma:p_convergence}
	If $\E|X|<\infty$, then
	\begin{enumerate}
		\item $\lim_{n\to\infty}\lim_{m\to\infty}||\F_{n,m}-F||_p=_{a.s.}0$ for every $p\in[1,\infty)$;
		\item $\lim_{m\to\infty}\lim_{n\to\infty}||\F_{n,m}-F||_p=_{a.s.}0$, for every $p\in[1,\infty]$.
		\item If $F$ is continuous, $\lim_{m\to\infty}\lim_{n\to\infty}||\F_{n,m}-F||_\infty=_{a.s.}0$.
	\end{enumerate}
\end{lemma}

Finally, the convergence of $\F_{n,m}$ to $F$ as a double limit, that is, letting $n,m\to\infty$, follows from Theorem~\ref{thm:Hm_L1_contraction}.
\begin{theorem}\label{theo:double_limit}
	If $\mathbb E|X|<\infty$, then, for every $p\in[1,\infty)$, $
	\|\F_{n,m}-F\|_p \to_{a.s.} 0$ as $n,m\to\infty.$
\end{theorem}
\noindent We now focus on the Wasserstein distance between the EHCDF on the true distribution. Recall that for $p=1$, the Wasserstein distance coincides with the $L^1$ distance, that is, given CDFs $F$ and $G$, $\|F-G\|_1=\mathcal{W}_1(F,G).$ Therefore all the above results apply to $\mathcal W_1(\F_{n,m},F).$ With regard to $\mathcal W_p(\F_{n,m},F_m)$, establishing convergence is particularly simple because of the construction of $\F_{n,m}$ and $F_m$. Indeed, note that $$\mathcal W_p^p(\F_{n,m},F_m)=\frac1m\sum_{j=1}^m|\hat{\mu}_{j:m}-\mu_{j:m}|^p. $$ Hence, for fixed $m$, $\hat{\mu}_{j:m}\to_{a.s.}\mu_{j:m}$ immediately yields the following result.
\begin{proposition}\label{prop:Wp_fixed_m}
	Let $m\ge 1$ be fixed and let $p\ge 1$. If $\E|X|<\infty$, then \( W_p(\F_{n,m},F_m)\to_{a.s.}0\) as \(n\to\infty.\)
\end{proposition}

\noindent Note that, for $m\to\infty$, the identity $W_p(\F_{n,m},F)
=
\|\mathcal I_m(\F_n^{-1})-F^{-1}\|_p ,$ yields the following convergence using Theorem~\ref{teo:Im_convergence}, provided that $F^{-1}\in L^p(0,1).$
\begin{corollary}\label{cor:Wp_consistency}
	Assume that $\E|X|^p<\infty$ for some $p\ge 1$. Then
	\(
	W_p(\F_{n,m},F) \to_{a.s.}0 \) for  $n,m\to\infty$.
\end{corollary}
\section{Weak convergence of the Hoeffding empirical quantile process}\label{sec5}
We start by observing that it is not convenient to study empirical processes based on the proposed estimator directly in the CDF domain. Indeed, consider the process $\sqrt n(\F_{n,m}-F_m)$. For fixed $m$, the step functions $\F_{n,m}$ and $F_m$ have the same jump sizes $1/m$, and differ only through the locations of their jumps, namely $\hat{\mu}_{j:m}$ and $\mu_{j:m}$. Under suitable conditions, $\hat{\mu}_{j:m}-\mu_{j:m}=O_p(n^{-1/2})$ for every fixed $j$. Then, $\F_{n,m}-F_m$ can take non-zero values only on a union of intervals whose total length vanishes as $n\to\infty$. As a consequence the process $\sqrt n(\F_{n,m}-F_m)$ tends to zero a.s., almost everywhere. Centering at $F$ instead of $F_m$ does not solve this issue. If $m$ is fixed, the decomposition \( \sqrt n(\F_{n,m}-F) = \sqrt n(\F_{n,m}-F_m)+\sqrt n(F_m-F) \) shows that one introduces a deterministic error term of order $\sqrt n$, since $F_m\neq F$ in general. If $m\to\infty$, this term may become negligible under suitable conditions, but the  behaviour of stochastic term $\sqrt n(\F_{n,m}-F_m)$ becomes hard to address. This suggests that the CDF domain is not suitable to describe the asymptotic behaviour of the estimator. By contrast, such a behaviour can be studied more effectively in the quantile domain.

In this section we study the empirical quantile process associated with $\F_{n,m}$, denoted as $\mathbb Q_{n,m}=\sqrt n(\F_{n,m}^{-1}-F^{-1}_m).$ We will study weak convergence of the above process for $n\to\infty$ and $m$ fixed, and for $n,m\to\infty$ jointly. We will show that weak convergence of $\mathbb Q_{n,m}$ follows from that of the classic empirical quantile process $\mathbb Q_{n}=\sqrt n(\F_{n}^{-1}-F^{-1}),$ which, under suitable conditions, can be shown to converge to the well known limit $\mathbb Q:=\frac{\mathcal B}{f\circ F^{-1}},$ where $\mathcal B$ denotes a Brownian bridge.

In addition, our results can be extended to the bootstrap world. Define, the bootstrap ECDF as $$\F_n^*(x)=\frac1n\sum_{i=1}^n M_{n,i}\ind\{X_i\leq x\},$$ where $M_{n}=(M_{1,n},...,M_{n,n})$ is independent of the data and drawn from a multinomial distribution with uniform probabilities over $n$ trials \citep[p. 180]{vw}. Accordingly, define $\F_{n,m}^*=\mathcal H_m(\F_n^*)$. The bootstrap version of $\mathbb Q_{n,m}$ is defined as $\mathbb Q_{n,m}^*=\sqrt n((\F_{n,m}^*)^{-1}-\F_{n,m}^{-1}).$
\subsection{The empirical quantile process}
In general, weak convergence of empirical quantile processes is not easy to establish. In classic empirical process theory (see, for instance, 3.9.4.2 in \cite{vw}), convergence of $\mathbb Q_{n}$ is established with respect to the supremum norm using the functional delta method, but this works if $F$ has a compact support, or, similarly, in $\ell^\infty[p,q]$ where $0<p<q<1$. Differently, a key insight for our analysis comes from the recent work of \cite{beare2025necessary}, who established a necessary and sufficient condition for the weak convergence of $\mathbb Q_{n}$ in the space \(L^{1}(0,1)\). We will denote this condition as Property~\ref{prop:Q} as in the original paper.
\begin{PropertyNamed}{Q}\label{prop:Q}
	The quantile function $Q$ is locally absolutely continuous and satisfies $$\int_0^1 \sqrt{u(1 - u)} \, dQ(u) < \infty.$$
\end{PropertyNamed}
\noindent According to \cite{beare2025necessary}, if (and only if) Property~\ref{prop:Q} holds, we have $\mathbb{Q}_n \rightsquigarrow \mathbb Q$ in $L^1(0,1).$ Local absolute continuity of $Q$ means that the restriction of $Q$ to each compact subinterval of $(0,1)$ is absolutely continuous. Requiring the integral to be finite is slightly more restrictive than requiring $Q$ to be square integrable.

In a more general setting, we are interested in the weak convergence of $\mathbb Q_n$ to $\mathbb Q$ in some $L^p(0,1)$ space. In this regard, \cite{delbarrio} (Theorem 4.6) proved that the quantile process converges in $L^2(0,1)$ relying on the following assumption, which is stronger than Property~\ref{prop:Q}.
\begin{PropertyNamed}{Q2}
\label{prop:Q2}
The quantile function $Q$ is locally absolutely continuous with density $q=Q'$ satisfying $\int_0^1 u(1-u)q^2(u)\,du<\infty$. The CDF $F$ is twice differentiable on its open support $(a,b)$, with positive derivative $f$, and $\sup_{u\in(0,1)} |u(1-u)\,q(u)^2\,f'(Q(u))|<\infty$. Moreover, either $a>-\infty$ or $\liminf_{u\downarrow 0}|u\,q(u)^2\,f'(Q(u))|>0$, and either $b<\infty$ or $\liminf_{u\uparrow 1}|(1-u)\,q(u)^2\,f'(Q(u))|>0$.
\end{PropertyNamed}
\subsection{Fixed \texorpdfstring{$m$}{m}}
Using $\mathcal I_m$, we can express $\F_{n,m}^{-1}=\mathcal I_m(\F_n^{-1})$ and, by linearity,
$
\mathbb Q_{n,m}=\mathcal I_m(\mathbb Q_n).
$
For fixed $m$, the operator $\mathcal I_m$ is a continuous linear map from $L^1(0,1)$ into $L^p(0,1)$ for every $p\in[1,\infty)$; this follows from the definition of $\mathcal I_m$, together with Lemma~\ref{lemma:Im_contraction}. Hence, the following result is an immediate consequence of Property~\ref{prop:Q} and the continuous mapping theorem.

\begin{theorem}\label{theo:convergence_fixedm}
Under Property~\ref{prop:Q}, $\mathbb Q_{n,m}\rightsquigarrow \mathbb Q_m:=\mathcal I_m(\mathbb Q)$ in $L^p(0,1),$ for $p\in[1,\infty).$
	The limit process is centred and Gaussian with representation
	$\mathbb Q_m(u)=\sum_{j=1}^mZ_j\ind\{u\in\left(\tfrac{j-1}{m},\, \tfrac{j}{m}\right]\}$, where $Z_j:=\int_0^1 \mathbb QdF_{B_{j:m}}$.
\end{theorem}

\noindent It is interesting to observe that the limit process $\mathbb Q_m(u)$ is a step function which takes the random value $Z_j$ whenever $u\in \left(\tfrac{j-1}{m},\, \tfrac{j}{m}\right]$. In other words, the random function $\mathbb Q_m$ is obtained by assigning the random height $Z_j$ to the $j$-th interval of the regular partition of $(0,1)$ into $m$ subintervals. Indeed, a simpler way to express this process is $\mathbb Q_m(u)=Z_{\lceil mu \rceil}$, for $u\in[0,1].$

By Lemma~\ref{lemma:Im_contraction}, the Hoeffding quantile operator reduces the stochastic component of the empirical quantile fluctuation around its deterministic target, that is, the inequality$$ \|\mathbb Q_m\|_p\le \|\mathbb Q\|_p$$ holds pathwise. The total error with respect to $F^{-1}$, however, also includes the deterministic term $\sqrt n(F^{-1}_m-F^{-1})$.

The following bootstrap version of Theorem~\ref{theo:convergence_fixedm} follows immediately from the conditional weak convergence
$
\mathbb Q_n^* \rightsquigarrow \mathbb Q$ in $
L^1(0,1),$ established by \citet{beare2025necessary}, together with the continuity of $\mathcal I_m$ from
$L^1(0,1)$ into $L^p(0,1)$ for fixed $m$. Since
$
\mathbb Q_{n,m}^*=\mathcal I_m(\mathbb Q_n^*),$ bootstrap validity is inherited by the transformed process.

\begin{corollary}\label{cor_boot}
Under Property~\ref{prop:Q}, conditionally on the data, in probability,
$
\mathbb Q_{n,m}^* \rightsquigarrow \mathbb Q_m$ in $L^p(0,1),$ for every $p\in[1,\infty)$.
\end{corollary}
\subsection{Case \texorpdfstring{$m\to\infty$}{m to infinity}}
We now study the weak convergence of $\mathbb Q_{n,m}$ when $n$ and $m$ diverge to infinity together. We also consider the alternative version of this process, centred in $F^{-1}$ instead of $F_m^{-1},$ and defined as $\widetilde{\mathbb Q}_{n,m}=\sqrt n(\F_{n,m}^{-1}-F^{-1}).$ We will prove that, under suitable conditions, the limit process of both ${\mathbb Q}_{n,m}$ and $\widetilde{\mathbb Q}_{n,m}$ is the usual process $\mathbb Q.$ Therefore, quantile-wise, in the joint regime $n,m\to\infty,$ the EHCDF method is asymptotically equivalent to the classic approach based on the ECDF.

Under Property~\ref{prop:Q}, $\mathbb Q_n \rightsquigarrow \mathbb Q$ in $L^1(0,1)$. Theorem~\ref{teo:Im_convergence} allows us to extend this convergence to $\mathbb Q_{n,m} \rightsquigarrow \mathbb Q$ in $L^1(0,1)$ as $n,m\to\infty$. Moreover, $
\widetilde{\mathbb Q}_{n,m}
=
\mathbb Q_{n,m}+\sqrt n\,(F_m^{-1}-F^{-1}).$ By Lemma~\ref{lemma:HCDF_CDF_convergence_rate}, if $n=o(m)$, the second term vanishes in $L^1(0,1)$, and therefore $\widetilde{\mathbb Q}_{n,m}\rightsquigarrow \mathbb Q$ in $L^1(0,1)$. The same argument applies in $L^p(0,1)$ under suitable assumptions; in particular, the result holds in $L^2(0,1)$ under Property~\ref{prop:Q2}.

\begin{theorem}\label{theorem:weak} Assume that Property~\ref{prop:Q} holds. Then
	$
	\mathbb Q_{n,m}\rightsquigarrow \mathbb Q
	$ in $L^1(0,1),
	$
	as $m,n\to\infty$. The same results hold in $L^2(0,1)$ assuming that Property~\ref{prop:Q2} holds. In addition, if Property~\ref{prop:Q} holds and $n = o(m)$, then $\widetilde{\mathbb Q}_{n,m} \rightsquigarrow \mathbb{Q}$. More generally, if $\mathbb Q_n\rightsquigarrow \mathbb Q$ in $L^p(0,1)$ then we immediately have $\mathbb Q_{n,m}\rightsquigarrow \mathbb Q$ in $L^p(0,1)$.
\end{theorem}

\noindent Finally, the results of \cite{beare2025necessary} allow us to extend Theorem~\ref{theorem:weak} to the bootstrap version of $\mathbb Q_{n,m}$, i.e., $\mathbb Q_{n,m}^*$, and the corresponding bootstrap version of $\widetilde{ \mathbb Q}_{n,m}$, defined as $\widetilde {\mathbb Q}_{n,m}^*=\sqrt n((\F_{n,m}^*)^{-1}-\F_{n}^{-1})$. If $n=o(m)$, Lemma~\ref{lemma:HCDF_CDF_convergence_rate} still applies, and the difference between $\widetilde{\mathbb Q}_{n,m}^*$ and $\mathbb Q_{n,m}^*$ is asymptotically negligible in $L^1(0,1)$.

\begin{corollary}\label{cor_boot_inf}
Under Property~\ref{prop:Q}, conditionally on the data, in probability,
$
\mathbb Q_{n,m}^* \rightsquigarrow \mathbb Q$ in $L^1(0,1)$ as $n,m\to\infty$. If $n=o(m)$, the same result holds for $\widetilde{\mathbb Q}_{n,m}^*$.
\end{corollary}

\section{Asymptotic distribution of distance-based functionals}\label{sec6}
\subsection{Convergence in distribution of \texorpdfstring{$L^p$}{Lp} distances}
We start by studying convergence in distribution of $L^p$ norms of the difference between $\F_{n,m}$ and $F_m$, or between $\F_{n,m}$ and $F$. Such functionals can be used, for instance, in goodness-of-fit procedures. Recall that, for $p=1$, the two distances coincide, that is, given CDFs $F$ and $G$, $\|F-G\|_1=\mathcal{W}_1(F,G).$

We first focus on the case $p=1$. For a fixed value of $m$, it is relatively simple to establish that the limit of $\sqrt n\| \F_{n,m}-F_m\|_1$ is the sum of absolute values of zero-mean Gaussian variables. Moreover, we already know by Theorem~\ref{thm:Hm_L1_contraction} that, for the general $L^p$ norm, $\sqrt{n} \|\F_{n,m}-F_m\|_p^p\leq \sqrt{n} \|\F_{n}-F\|_1=\sqrt n \|\F_n^{-1}-F^{-1}\|_1$, where it is easily seen that the latter converges in distribution to $\int_0^1|\mathbb Q|$. For $p=1$ and for $n,m\to \infty,$ we show that the limit coincides with the limit of the upper bound. This follows directly from the identity $\|\F_{n,m}-F_m\|_1=\mathcal{W}_1(\F_{n,m},F_m)$ combined with Theorem~\ref{theorem:weak}.
\begin{theorem}\label{theo:asymptotic} Under Property~\ref{prop:Q}, as \(n \to \infty\) with \(m\) fixed,
	\[
	\sqrt{n}\,\|\mathbb{F}_{n,m}-F_m\|_1 \to_d \frac{1}{m}\sum_{j=1}^m\Bigg| \int_0^1 \mathbb QdF_{B_{j:m}}\Big|.
	\]
	Moreover, as $n,m\to\infty$,
	$
	\sqrt{n}\,\|\mathbb{F}_{n,m}-F_m\|_1
	\to_d
	\int_0^1 |\mathbb{Q}|.$
\end{theorem}
\noindent Now, by the reverse triangle inequality,
\[
\Bigl|\sqrt n\,\|\F_{n,m}-F\|_1-\sqrt n\,\|\F_{n,m}-F_m\|_1\Bigr| \le \sqrt n\,\|F-F_m\|_1.
\]
Since $\sqrt n\|F-F_m\|_1=O(n^{1/2}m^{-1/2})$ by Lemma~\ref{lemma:HCDF_CDF_convergence_rate}, we immediately derive the following result.
\begin{corollary}\label{cor:F}
	Under Property~\ref{prop:Q}, if $n,m\to\infty$ and $n = o(m)$, then $\sqrt n \|\F_{n,m}-F\|_1\to_d \int_0^1|\mathbb Q|$.
\end{corollary}

\noindent To establish the convergence in distribution of $\|\F_{n,m}-F_m\|_p,$ scaled by a suitable sequence of constants, we distinguish the cases when $m$ is fixed to the case when $m\to\infty$. In the first case, for $p \in [1,\infty)$, we establish the convergence of $n^{1/(2p)}\|\F_{n,m}-F_m\|_p$ to a limit that generalizes the one appearing in the first assertion of Theorem~\ref{theo:asymptotic}.

\begin{theorem}\label{theo:asymptotic_fixed_m}
	Under Property~\ref{prop:Q}, for $n\to\infty$, with $m$ fixed,
	$$\sqrt n\| \F_{n,m}-F_m\|_p^p\to_d\frac{1}{m^p}\sum_{j=1}^m \Bigg|\int_0^1 \mathbb Q dF_{B_{j:m}}\Bigg|.$$
\end{theorem}%
\noindent We proved that $\sqrt n\| \F_{n,m}-F_m\|_p^p=O_p(1)$, or, equivalently, $ n^{\frac1{2p}}\| \F_{n,m}-F_m\|_p=O_p(1)$. This means that $ \sqrt n\| \F_{n,m}-F_m\|_p=O_p(n^\frac{p-1}{2p})$, that is, if scaled by $\sqrt n,$ the $L^p$ norm of the difference diverges to infinity whenever $p>1.$

For $n,m\to\infty,$ the number of parameters to estimate, $\mu_{j:m}$, diverge, as well as the sample size. The following result establishes convergence under some rather restrictive conditions on $F$.
\begin{proposition}\label{theo:asymptotic_bounded_support}
	Assume that $F$ has a bounded support and density function $f$ such that $0<a<f<b<\infty.$ Then, if $m=o(\sqrt n)$, $m^{p-1}\sqrt n \| \F_{n,m}-F_m\|_p^p\to_d \int_0^1|\mathbb Q|$.
\end{proposition}
\noindent The assumption \(0<a<f<b<\infty\) guarantees the two ingredients used in the proof: a spacing of order \(1/m\) between consecutive values \(\mu_{j:m}\), and the bound \(\|Q_n-Q\|_\infty=O_p(n^{-1/2})\), since \(Q'=1/(f\circ Q)\) is bounded away from both \(0\) and \(\infty\). A simple particular case of the above result is obtained for the uniform distribution, which directly satisfies $\mu_{j+1:m}-\mu_{j:m}=\tfrac1{m+1}$ and $\| Q_n-Q\|_\infty=O_p(n^{-1/2})$. Other distributions which satisfy the assumptions are truncated distributions with well-behaved densities, for example, one can take normal or exponential distributions, restricted on supports of the form $[a,b]$.

Bootstrap analogues can be obtained by considering $\|\F_{n,m}^*-\F_{n,m}\|_p$, suitably rescaled. The fixed-$m$ limits corresponding to Theorems~\ref{theo:asymptotic} and~\ref{theo:asymptotic_fixed_m} follow from Corollary~\ref{cor_boot}, conditionally on the data, in probability. For $n,m\to \infty,$ the $L^1$ limit follows from Corollary~\ref{cor_boot_inf}.
\subsection{Convergence in distribution of Wasserstein distances}
For a fixed value of $m$, noticing that
$$
n^{p/2}W_p^p(F_{n,m},F_m) = \frac1m\sum_{j=1}^m \big|\sqrt n(\hat\mu_{j:m}-\mu_{j:m})\big|^p,$$ we easily obtain convergence in distribution.
\begin{proposition}
	Let $m\ge 1$ be fixed and let $p\ge 1$. Under Property~\ref{prop:Q},
	\[
	n^{p/2}W_p^p(\F_{n,m},F_m) \to_d \frac1m\sum_{j=1}^m \left| \int_0^1 \mathbb QdF_{B_{j:m}} \right|^p.
	\]
\end{proposition}
\noindent This result admits a bootstrap analogue. By Corollary~\ref{cor_boot},
$
n^{p/2}W_p^p(\F_{n,m}^*,\F_{n,m})
$
converges conditionally on the data, in probability, to the same limit.

If $m\to\infty$, Theorem~\ref{theo:asymptotic} already addresses convergence in distribution of $\mathcal{W}_1$, given the identity $\|\F_{n,m}-F_m\|_1=\mathcal{W}_1(\F_{n,m},F_m)$. To extend the result to $p>1$ we need weak convergence of the quantile process in $L^p(0,1)$, which could be obtained in general under stronger assumptions. Under Property~\ref{prop:Q2}, we are able to follow the same steps of the proof of Theorem~\ref{theo:asymptotic} and establish the following result.
\begin{theorem}Under Property~\ref{prop:Q2}, for $n,m\to\infty$,
	$$nW_2^2(\F_{n,m},F_m)\to_d \int_0^1\mathbb Q(t)^2dt.$$
\end{theorem}
\noindent At present, we are not aware of any bootstrap convergence result of the quantile process in $L^2(0,1)$, hence, the above result does not have a bootstrap analogue.

\section{A differentially private EHCDF}\label{sec7}
In differentially private inference, one aims to release useful information about the data while controlling the effect that any single observation can have on the released output. Thus, privacy and statistical accuracy have to be balanced. The problem of releasing private versions of a distribution, or corresponding functionals, has recently received increasing attention. The univariate setting is particularly relevant for economic size distributions, such as income and expenditure, where the underlying data are often
confidential and functionals such as quantiles and inequality
measures are of direct interest. In this context, the ECDF is not the most convenient tool: a direct private release of its support points implies a substantial
cost in terms of statistical accuracy. Private histograms, where bin absolute frequencies are perturbed by additive noise, are the standard method for the release a private distribution \citep{DworkRoth2014,WassermanZhou2010}. \citet{GillenwaterJosephKulesza2021} focus on the simultaneous release of a vector of quantiles. \citet{BarczewskiMawassRamon2025} consider private releases of the empirical CDF, based on privatized CDF values and suitable post-processing to recover a valid distribution function. Closest to our perspective, \citet{TaoSarwate2026} formulate private CDF estimation as a finite-dimensional function-approximation problem: the ECDF is approximated in a chosen function space, the corresponding coefficients are privatised, and the private CDF is reconstructed from them. In this section, instead of privatising selected quantiles, CDF values, or coefficients, we privatise the \(m\) jump points of the EHCDF. As discussed below, the structural properties of the EHCDF make this representation particularly well suited for a differentially private release of the CDF.

We begin by introducing some basic concepts. We use the standard model of differential privacy \citep{DworkRoth2014}, in which a curator observes the dataset $
\mathcal X=(X_1,\ldots,X_n)$ and releases a randomised output. As is customary in differential privacy, we assume that the private data, after possible clipping, belong to a bounded interval
\([a,b]\), where $a$ and $b$ are public, and \(B=b-a\).
A randomised mechanism with output space \(\mathcal Y\) is an algorithm which, given a dataset \(\mathcal X\in [a,b]^n\), returns a random vector
\(M(\mathcal X)\) of \(\mathcal Y\).  We write \(\mathcal X\sim\mathcal X'\) if two datasets $\mathcal X$ and $\mathcal X'$ of size $n$ are adjacent, that is, if they differ in at most one observation. A randomised mechanism \(M\), with output space \(\mathcal Y\), is
\(\varepsilon\)-differentially private if, for every pair of adjacent datasets
\(\mathcal X\sim\mathcal X'\) and every measurable set \(A\subseteq\mathcal Y\),
\[
 P(M(\mathcal X)\in A) \le e^\varepsilon P(M(\mathcal X')\in A).
\]
This means that the distribution of the output released changes by at most a multiplicative factor \(e^\varepsilon\) when an observation is modified. The parameter \(\varepsilon\) controls the privacy level: smaller values force closer probabilities, while larger values allow for more flexibility.

A deterministic function of the data is called a query. The \(\ell_1\)-sensitivity of a vector-valued query
\(q:\mathcal X\mapsto q(\mathcal X)\in\mathbb R^k\), is defined as
\[
\Delta_1 (q) = \sup_{\mathcal X\sim\mathcal X'} \|q(\mathcal X)-q(\mathcal X')\|_1.
\]
One of the main methods in differential privacy is the Laplace mechanism. Given some query \(q\), such a mechanism consists in perturbing each component of $q(\mathcal X)$ with random noise drawn from the Laplace distribution, obtaining the randomised release $
M(\mathcal X)
=
q(\mathcal X)+ (Z_1,...,Z_k),$ where $Z_1,...,Z_k$ are iid Laplace random variables with scale $\Delta_1(q)/{\varepsilon}$. The Laplace mechanism preserves \(\varepsilon\)-differential privacy \citep[Theorem 3.6]{DworkRoth2014}.

To define a mechanism based on the EHCDF, without loss of generality we can assume that the observations belong to \([a,b]\). Indeed, if $X$ is not naturally bounded, the mechanism is applied to the clipped variable
$
X^{[a,b]}=\Pi_{[a,b]}(X),$ where $
\Pi_{[a,b]}(x)=\min\{\max(x,a),b\}.$ Define the query
\[
\widehat{\boldsymbol\mu}_m(\mathcal X) = (\widehat\mu_{1:m},\ldots,\widehat\mu_{m:m})
\]
The next result computes the exact $\ell_1$-sensitivity of this query.
\begin{lemma}\label{sensitivity} The $\ell_1$-sensitivity of $\widehat{\boldsymbol\mu}_m$ is
\(
\Delta_1(\widehat{\boldsymbol\mu}_m) = (mB)/n. \)
\end{lemma}

We may now define a private version of the EHCDF. Generate independent Laplace random variables
\[
Z_1,\ldots,Z_m \stackrel{\mathrm{iid}}{\sim} \mathrm{Lap}\left(\frac{mB}{n\varepsilon}\right),
\]
where \(\mathrm{Lap}(s)\) denotes the centered Laplace distribution with scale
\(s\). Define
\[
\widetilde\mu_j = \Pi_{[a,b]}(\widehat\mu_{j:m}+Z_j), \qquad j=1,\ldots,m.
\]
The private EHCDF is
\[
\widetilde \F_{n,m}^{\,\varepsilon}(x) = \frac1m\sum_{j=1}^m \ind\{\widetilde\mu_j\le x\}.
\]

Now, the following result is an immediate consequence of Theorem~3.6 of \cite{DworkRoth2014} combined with Lemma~\ref{sensitivity} .

\begin{theorem}The mechanism \(\mathcal X\mapsto(\widetilde\mu_{1:m},\ldots,\widetilde\mu_{m:m})\) is \(\varepsilon\)-differentially private. Consequently, \(\mathcal X\mapsto \widetilde \F_{n,m}^{\,\varepsilon}\) is also \(\varepsilon\)-differentially private.
\end{theorem}

The next result gives an accuracy bound for the private EHCDF in \(W_1\) distance. The bound decomposes the total error into three terms: the deterministic error of the Hoeffding operator, the usual error of the ECDF, and an additional error due to privacy. This decomposition is a consequence of two EHCDF properties: the Wasserstein non-expansiveness, and the exact \(\ell_1\)-sensitivity
\(mB/n\).

\begin{theorem}\label{Theo: privacy}
If \(F\) has support \([a,b]\), then the private EHCDF satisfies the following bound.
$$
W_1(\widetilde \F_{n,m}^{\,\varepsilon},F) \le W_1(F_m,F) + W_1(\F_n,F) + \frac1m\sum_{j=1}^m |Z_j|,
$$
Moreover, \begin{align*}
\mathbb E W_1(\widetilde \F_{n,m}^{\,\varepsilon},F)
&\le
W_1(F_m,F)
+
\mathbb E W_1(\F_n,F)
+
\frac{mB}{n\varepsilon}
\\[0.5em]
&\le
\left(
\frac1{\sqrt m}
+
\frac1{\sqrt n}
\right)
\int_0^1 \sqrt{u(1-u)}\,dQ(u)
+
\frac{mB}{n\varepsilon}
\\[0.5em]
&\le
\frac{B}{2\sqrt m}
+
\frac{B}{2\sqrt n}
+
\frac{mB}{n\varepsilon}.
\end{align*}
\end{theorem}

The bounds decompose the error into a Hoeffding approximation term, a sampling term, and a privacy term. While the first bound holds deterministically, the others hold in expectation. The parameter \(m\) controls the accuracy--privacy tradeoff. Larger values of \(m\) reduce the deterministic approximation error \(W_1(F_m,F)\), but increase the privacy term \(mB/(n\varepsilon)\). This bound decomposition is similar to that given in Theorem~3 and Theorem~5 of \citet{TaoSarwate2026}. In their framework, the ECDF is first approximated in a chosen function space and the corresponding coefficients are then privatised. They obtain a bound expressed in the $L^2$ metric, as this is more natural for their projection-based function approximation. The present bound has the same structure, but is more direct in this setting: the finite-dimensional representation is intrinsic to the EHCDF, and the privacy is explicitly given by \(mB/(n\varepsilon)\). Moreover, our bound is based on \(W_1\), which is more natural for the EHCDF construction.

Since \(F\) has support \([a,b]\), we have
\(
W_1(F_m,F)=O(m^{-1/2})\) and \(
\mathbb E W_1(\F_n,F)=O(n^{-1/2}).\)
Thus the Wasserstein error bound contains three contributions: a
finite-resolution term of order \(m^{-1/2}\), a sampling term of order
\(n^{-1/2}\), and a privacy-noise term of order \(m/(n\varepsilon)\). Balancing
the approximation term with the privacy term suggests choosing
\(m\asymp(n\varepsilon)^{2/3}\), yielding
\[
\mathbb E W_1(\widetilde \F_{n,m}^{\,\varepsilon},F)
=
O\bigl((n\varepsilon)^{-1/3}+n^{-1/2}\bigr).
\]
Minimising the upper bound
$
B/{(2\sqrt m)}+mB/({n\varepsilon})
$ with respect to $m$ we obtain \(m=(n\varepsilon/4)^{2/3}\). We therefore use
\begin{align}\label{eq:dp_mstar}
m^*(\varepsilon)=\Big\lceil (n\varepsilon/4)^{2/3}\Big\rceil,
\end{align}
The factor \(B=b-a\) cancels in the minimisation, so \(m^*(\varepsilon)\)
depends only on the public quantities \(n\) and \(\varepsilon\). We do not claim optimality for \eqref{eq:dp_mstar}, instead we use $m^*(\varepsilon)$ as a practical choice for empirical cases.

A useful feature of a private CDF release is that, once an
\(\varepsilon\)-differentially private CDF has been released, any collection of
functionals computed from it remains \(\varepsilon\)-differentially private, by the post-processing property. This does not mean that the resulting estimator is optimal for each functional. If the objective is to estimate a single functional, a mechanism tailored to that functional may be more accurate.

Recalling the definition of an $L$-functional, given CDFs $G$ and $H$, we have
\[
|T_w(G)-T_w(H)| \le \|w\|_\infty W_1(G,H).
\]
Therefore, Theorem~\ref{Theo: privacy} implies the following bound for $L$-estimators based on $\widetilde \F_{n,m}^{\,\varepsilon}$.

\begin{corollary}
For every \(w\in L^\infty(0,1)\),
\[
\left| T_w(\widetilde \F_{n,m}^{\,\varepsilon})-T_w(F) \right| \le \|w\|_\infty \left\{ W_1(F_m,F) + W_1(\F_n,F) + \frac1m\sum_{j=1}^m |Z_j| \right\}.
\]
Moreover,
\begin{align*}
\mathbb E
\left|
T_w(\widetilde \F_{n,m}^{\,\varepsilon})-T_w(F)
\right|
&\le
\|w\|_\infty
\left\{
W_1(F_m,F)
+
\mathbb E W_1(\F_n,F)
+
\frac{mB}{n\varepsilon}
\right\}
\\
&\le
\|w\|_\infty
\left\{
\int_0^1\sqrt{u(1-u)}dQ(u)
\left(
\frac1{\sqrt m}
+
\frac1{\sqrt n}
\right)
+
\frac{mB}{n\varepsilon}
\right\}
\\
&\le
\|w\|_\infty
\left\{
\frac{B}{2\sqrt m}
+
\frac{B}{2\sqrt n}
+
\frac{mB}{n\varepsilon}
\right\}.
\end{align*}
\end{corollary}


\section{Numerical analysis}\label{sec8}\label{sec:dp_sim}
\subsection{Simulation study}

In this section we assess by simulation the differentially private EHCDF introduced in Section~\ref{sec7}. We report here some explanatory results, additional
simulation results are reported in the Appendix.
Code reproducing the simulation study is available at \texttt{github.com/tedescolor/EHCDF}.

We compare the private EHCDF $\widetilde{\mathbb F}_{n,m}^{\,\varepsilon}$ with the perturbed histogram of \citet{WassermanZhou2010}, the standard differentially private estimator of a distribution. The estimator is implemented by the function \texttt{histogramDP} in the \texttt{R} package \texttt{DPpack} \citep{giddens2023dppack} and is constructed as follows. The public interval $[a,b]$ is divided into $k$ bins of equal width. An \textit{iid} $\operatorname{Laplace}(0,2/\varepsilon)$ perturbation is added to the bin absolute frequencies, negative counts are set to zero, and the resulting vector is normalised. The released CDF is then obtained by linear interpolation within each bin.

For the EHCDF, we set $m=m^*(\varepsilon)$ as defined in \eqref{eq:dp_mstar}. For the histogram, we use the package default $k^*=\lceil \log_2 n+1\rceil$ \citep{sturges1926choice}.
Since these are only practical default choices, we also vary \(m\) for the EHCDF and \(k\) for the histogram, and
select the values that give the smallest \(W_1\) error with respect to the true
distribution.

The data are generated from the zero-inflated beta prime distribution: with probability $p_0$ the observation is zero, and otherwise it is drawn from the beta prime distribution $\mathcal{B}'(\alpha,\beta)$, obtained by transforming a beta random variable $Y\sim \mathcal B (\alpha,\beta)$ via the transformation $X=Y/(1-Y)$. The beta-prime distribution is a special case of the generalised beta of the
second kind, a family commonly used to model size distributions in economics,
including income and expenditures. The parameter $\beta$ represents the tail index of the distribution: the moment of order $k$ is finite when $\beta>k$. Since the simulation results are more sensitive to the tail parameter
\(\beta\), we fix \(\alpha=1\). Hence
\[
F(x)=p_0+(1-p_0)\bigl\{1-(1+x)^{-\beta}\bigr\},\qquad x\ge 0.
\]
We consider $\beta\in\{1.5,2,2.5,5\}$ and $p_0\in\{0,0.25\}.$ These choices are in line with the empirical features of the real data
analysed in Section~\ref{sec9}. In fact, apart from the scale, the values $p_0=0.25$, $\alpha=1$ and $\beta=5$ are close to the maximum likelihood estimates of a scaled zero-inflated beta prime distribution fitted to the credit card expenditure data of Section~\ref{sec9}, and the fitted model is not rejected by a bootstrap Kolmogorov--Smirnov test.

Because the support is unbounded, the public interval $[0,b]$ must be fixed without inspecting the data. The target is therefore the clipped distribution $F^c$, namely the CDF of $\Pi_{[0,b]}(X)$. To study the effect of the public bound, we set $b=cP_{99}$, where $P_{k}$ is the $k$-th percentile, with $c\in\{1.5,2,3,5,7,10,15,20\}$, that is, from relatively tight to highly conservative bounds. We consider the sample sizes $n=200$ and $n=1000$, the privacy levels $\varepsilon\in\{0.5,1,2\}$, and $500$ Monte Carlo replications. Errors are measured by $\E\,W_1(\cdot,F^c)$.

Figure~\ref{fig:dp_sim} shows the error of the two mechanisms as a function of their resolution, for the model with $p_0=0.25$ and $\beta=5$, and $n=1000$. Both curves are U-shaped, reflecting the trade-off between approximation and privacy error. For the EHCDF, the approximation error decreases rapidly and is largely unaffected by the width of the public interval. Consequently, only a small number of atoms is typically needed. In contrast, the approximation error of the histogram increases with $b$, so wider bounds require more bins and lead to larger privacy error: the minimum of the histogram curve both increases and moves to the right as $c$ grows. The default histogram resolution is close to the minimum of its error curve only for $c=1.5$. For larger values of $c$, $k^*$ is too small. The default EHCDF resolution is less affected by the values of $b$, but may become larger than necessary when $b$ is larger. Even in such cases, the EHCDF at $m^*$ is at least as accurate as the histogram at $k^*$.
The performance of the histogram improves for $p_0=0$, compared to the case $p_0=0.25$.

Figure~\ref{fig:dp_cross} compares the mechanisms at their default and best grid resolutions for all models, while Table~\ref{tab:dp_sim_hist} reports the corresponding numerical results for $p_0=0.25$ and $\beta=5$. The histogram at $k^*$ is the least accurate procedure in almost all settings. For the smallest values of $b$, the two default procedures are comparable when $p_0=0$, and for the lighter-tailed model $\beta=5$ the histogram at $k^*$ remains more accurate up to $c\approx 3$; in the other settings the EHCDF at $m^*$ is more accurate. When both procedures are optimised over their resolution grids, the histogram can be more accurate for small values of $b$. However, its error grows roughly proportionally to $c$, since both the binning and the privacy errors depend on the length of the public interval, whereas the error of the best EHCDF grows much more slowly. The crossing point lies at $c\approx 1.5$--$3$ in most settings, and larger samples and larger privacy levels move it to the right. Heavier tails and the mass at zero favour the EHCDF: for instance, when $\beta=1.5$, $n=200$ and $\varepsilon=0.5$, the best EHCDF is more accurate for every value of $c$. The best EHCDF usually requires a small number of jump points, especially for wide public intervals. Finally, Table~\ref{tab:dp_sim} shows that the total error depends mainly on the privacy component. The empirical privacy error is consistently smaller than its theoretical bound, while the contribution of the non-perturbed EHCDF gets larger (smaller) for smaller (larger) values of $b$.

    \begin{figure}[H]
        \centering
        \includegraphics[width=0.95\textwidth]{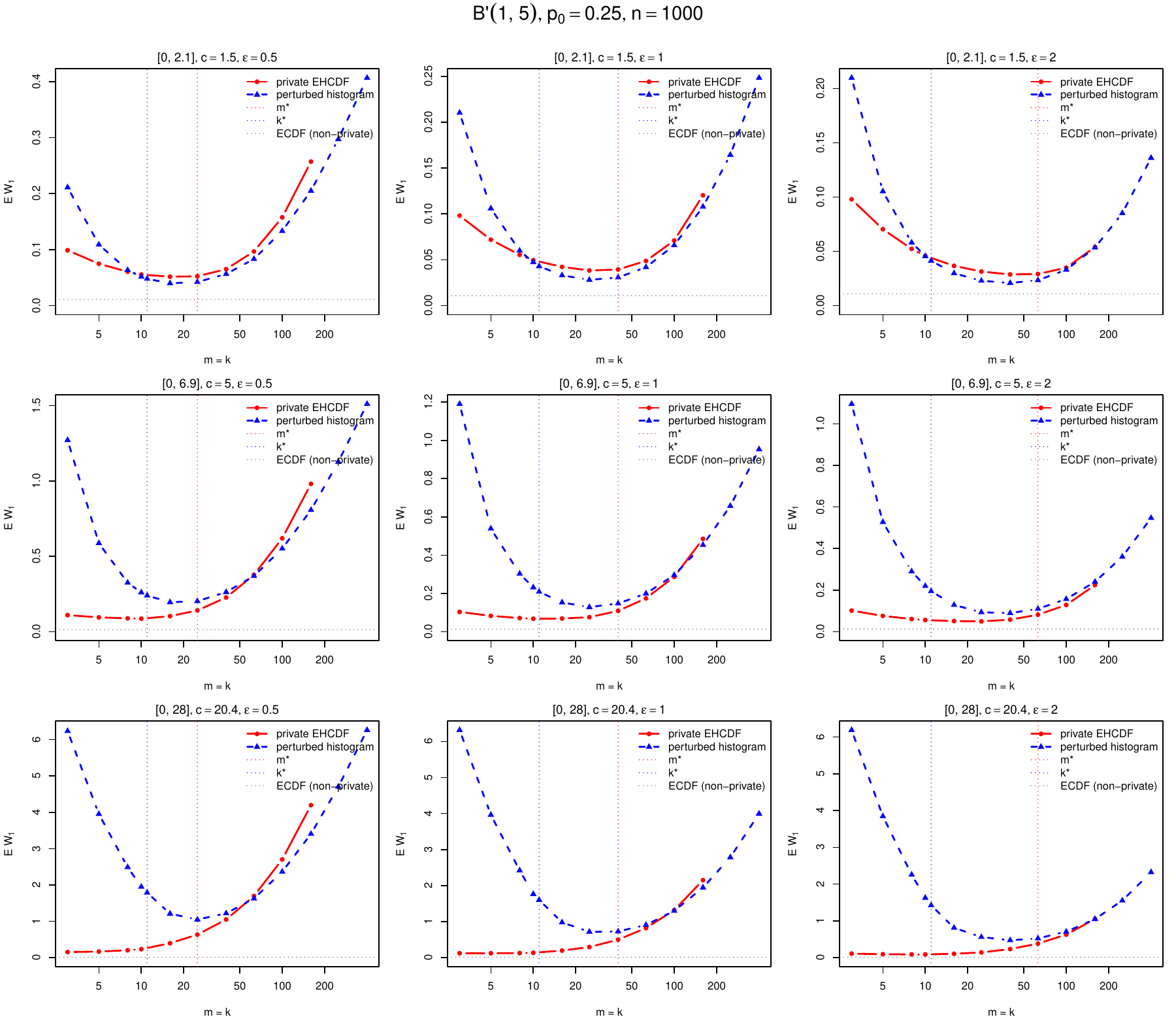}
        \caption{Mean error $\E\,W_1(\cdot,F^c)$ as a function of the common resolution $m=k$ on a logarithmic scale, for the zero-inflated $\mathcal{B}'(1,5)$ with $p_0=0.25$ and $n=1000$. Rows correspond to the values of $c$ (where $b=cP_{99}$), and columns to $\varepsilon$. Vertical lines indicate $m^*$ and $k^*$, while the horizontal line indicates the error of the non-private ECDF.}
        \label{fig:dp_sim}
    \end{figure}

    \begin{figure}[H]
        \centering
        \includegraphics[width=\textwidth,height=0.9\textheight,keepaspectratio]{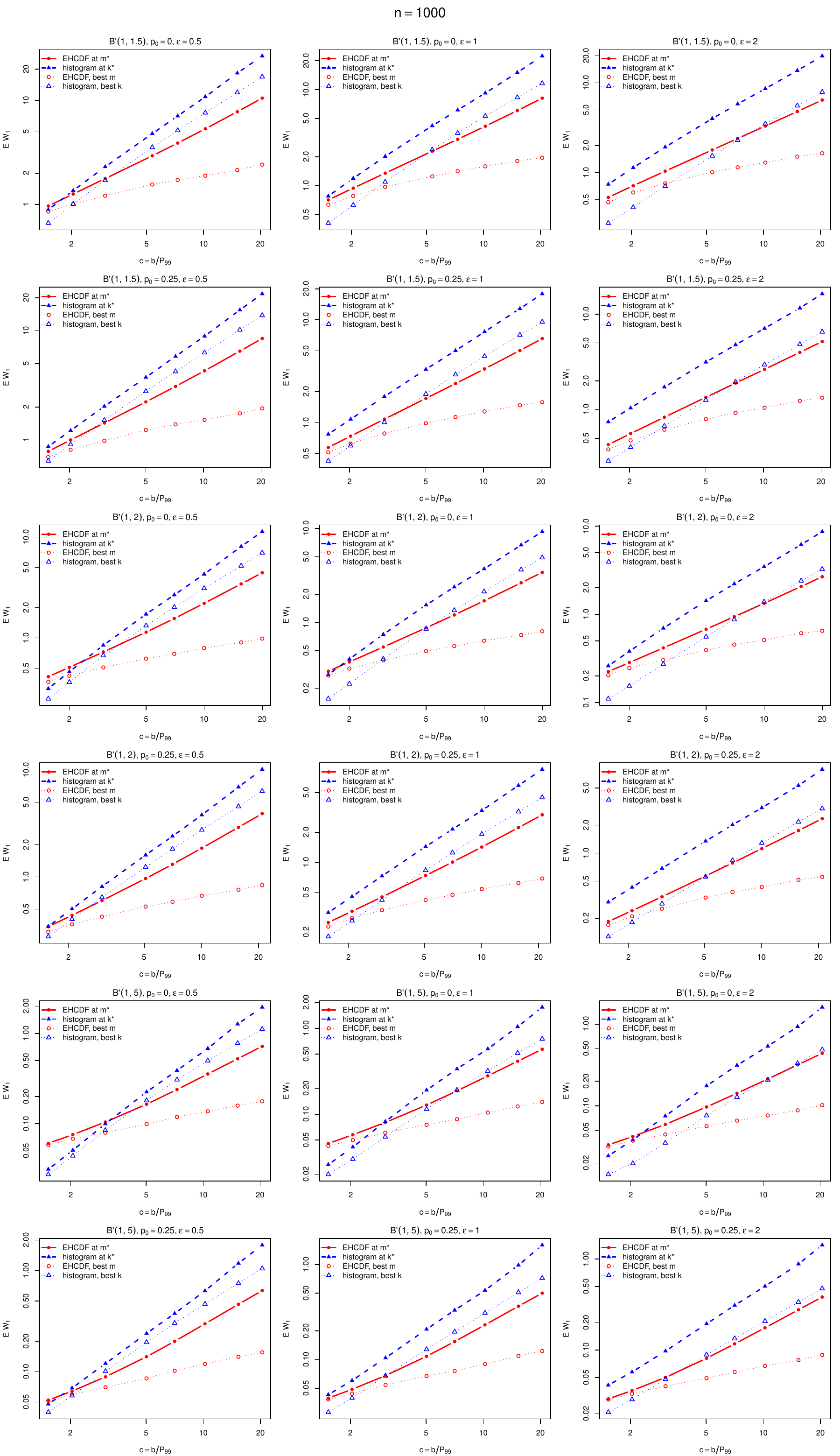}
        \caption{Mean error $\E W_1(\cdot,F^c)$ as a function of $c=b/P_{99}$ on a log-log scale, for the zero-inflated $\mathcal{B}'(1,\beta)$ with $\beta=1.5,2,5$ and $p_0=0,0.25$, $n=1000$, and $\varepsilon=0.5,1,2$. Rows correspond to the pairs $(\beta,p_0)$, and columns to $\varepsilon$.
        Solid and dashed curves show the private EHCDF at $m^*$ and the perturbed histogram at $k^*$, respectively.
        Dotted curves show the two mechanisms at their best resolutions over the grid.
        }
        \label{fig:dp_cross}
    \end{figure}

    \begin{table}[H]
        \centering\footnotesize
        \begin{adjustbox}{max width=0.95\textwidth}
            \begin{tabular}{c c c c c c cc cc}
\hline
 & & & & & & \multicolumn{2}{c}{default resolution} & \multicolumn{2}{c}{best resolution} \\
\cline{7-8}\cline{9-10}
$n$ & $[0,b]$ & $P(X>b)$ & $\varepsilon$ & $m^*$ & $k^*$ & EHCDF & Histogram & EHCDF ($m^\dagger$) & Histogram ($k^\dagger$) \\
\hline
200 & $[0,2.1]$ & 0.003 & 0.5 & 9 & 9 & 0.111 & 0.104 & 0.105 (5) & 0.103 (8) \\
200 & $[0,2.1]$ & 0.003 & 1.0 & 14 & 9 & 0.080 & 0.076 & 0.078 (9) & 0.072 (14) \\
200 & $[0,2.1]$ & 0.003 & 2.0 & 22 & 9 & 0.062 & 0.064 & 0.059 (16) & 0.051 (16) \\
200 & $[0,4.2]$ & 0.000 & 0.5 & 9 & 9 & 0.185 & 0.262 & 0.137 (3) & 0.256 (10) \\
200 & $[0,4.2]$ & 0.000 & 1.0 & 14 & 9 & 0.130 & 0.192 & 0.107 (8) & 0.176 (14) \\
200 & $[0,4.2]$ & 0.000 & 2.0 & 22 & 9 & 0.102 & 0.156 & 0.080 (10) & 0.117 (16) \\
200 & $[0,6.9]$ & 0.000 & 0.5 & 9 & 9 & 0.294 & 0.519 & 0.163 (3) & 0.495 (10) \\
200 & $[0,6.9]$ & 0.000 & 1.0 & 14 & 9 & 0.202 & 0.379 & 0.130 (5) & 0.328 (14) \\
200 & $[0,6.9]$ & 0.000 & 2.0 & 22 & 9 & 0.161 & 0.306 & 0.098 (8) & 0.221 (16) \\
200 & $[0,14]$ & 0.000 & 0.5 & 9 & 9 & 0.609 & 1.413 & 0.234 (3) & 1.203 (14) \\
200 & $[0,14]$ & 0.000 & 1.0 & 14 & 9 & 0.423 & 1.100 & 0.170 (3) & 0.783 (16) \\
200 & $[0,14]$ & 0.000 & 2.0 & 22 & 9 & 0.339 & 0.893 & 0.132 (5) & 0.531 (22) \\
200 & $[0,28]$ & 0.000 & 0.5 & 9 & 9 & 1.256 & 3.127 & 0.400 (3) & 2.787 (14) \\
200 & $[0,28]$ & 0.000 & 1.0 & 14 & 9 & 0.887 & 2.649 & 0.254 (3) & 1.898 (22) \\
200 & $[0,28]$ & 0.000 & 2.0 & 22 & 9 & 0.720 & 2.288 & 0.167 (3) & 1.209 (25) \\
\hline
1000 & $[0,2.1]$ & 0.003 & 0.5 & 25 & 11 & 0.052 & 0.048 & 0.052 (16) & 0.040 (16) \\
1000 & $[0,2.1]$ & 0.003 & 1.0 & 40 & 11 & 0.039 & 0.043 & 0.038 (25) & 0.028 (25) \\
1000 & $[0,2.1]$ & 0.003 & 2.0 & 63 & 11 & 0.029 & 0.041 & 0.029 (40) & 0.021 (40) \\
1000 & $[0,4.2]$ & 0.000 & 0.5 & 25 & 11 & 0.089 & 0.121 & 0.070 (10) & 0.101 (16) \\
1000 & $[0,4.2]$ & 0.000 & 1.0 & 40 & 11 & 0.068 & 0.104 & 0.054 (25) & 0.068 (25) \\
1000 & $[0,4.2]$ & 0.000 & 2.0 & 63 & 11 & 0.050 & 0.098 & 0.040 (25) & 0.048 (40) \\
1000 & $[0,6.9]$ & 0.000 & 0.5 & 25 & 11 & 0.141 & 0.239 & 0.086 (10) & 0.195 (16) \\
1000 & $[0,6.9]$ & 0.000 & 1.0 & 40 & 11 & 0.108 & 0.208 & 0.067 (10) & 0.128 (25) \\
1000 & $[0,6.9]$ & 0.000 & 2.0 & 63 & 11 & 0.081 & 0.194 & 0.049 (25) & 0.089 (40) \\
1000 & $[0,14]$ & 0.000 & 0.5 & 25 & 11 & 0.298 & 0.630 & 0.119 (5) & 0.468 (25) \\
1000 & $[0,14]$ & 0.000 & 1.0 & 40 & 11 & 0.232 & 0.535 & 0.090 (10) & 0.310 (25) \\
1000 & $[0,14]$ & 0.000 & 2.0 & 63 & 11 & 0.176 & 0.502 & 0.067 (10) & 0.208 (40) \\
1000 & $[0,28]$ & 0.000 & 0.5 & 25 & 11 & 0.635 & 1.789 & 0.155 (3) & 1.049 (25) \\
1000 & $[0,28]$ & 0.000 & 1.0 & 40 & 11 & 0.501 & 1.602 & 0.123 (5) & 0.718 (25) \\
1000 & $[0,28]$ & 0.000 & 2.0 & 63 & 11 & 0.383 & 1.424 & 0.089 (10) & 0.474 (40) \\
\hline
\end{tabular}

        \end{adjustbox}
        \caption{Mean error $\E\,W_1(\cdot,F^c)$ for the zero-inflated $\mathcal{B}'(1,5)$ with $p_0=0.25$ and $n=200,1000$, on the public interval $[0,b]$, where $b=cP_{99}$ and $c\in\{1.5,3,5,10,20\}$. The table reports the private EHCDF at $m^*$ and the perturbed histogram at $k^*$, as well as both mechanisms at their best grid resolutions, given in parentheses. The quantity $P(X>b)$ is the clipped tail mass.
        }
        \label{tab:dp_sim_hist}
    \end{table}

    \begin{table}[H]
        \centering\footnotesize
        \begin{adjustbox}{max width=0.95\textwidth}
            \begin{tabular}{c c c c r cc ccc}
\hline
$n$ & $[0,b]$ & $P(X>b)$ & $\varepsilon$ & $m^*$ & $\E W_1(\mathbb F_n,F^c)$ & $\E W_1(\mathbb F_{n,m},F^c)$ & $\E W_1(\widetilde{\mathbb F}_{n,m}^{\,\varepsilon},\mathbb F_{n,m})$ & $mB/(n\varepsilon)$ & $\E W_1(\widetilde{\mathbb F}_{n,m}^{\,\varepsilon},F^c)$ \\
\hline
200 & $[0,2.1]$ & 0.003 & 0.5 & 9 & 0.023 & 0.050 & 0.090 & 0.189 & 0.111 \\
200 & $[0,2.1]$ & 0.003 & 1.0 & 14 & 0.023 & 0.039 & 0.062 & 0.147 & 0.080 \\
200 & $[0,2.1]$ & 0.003 & 2.0 & 22 & 0.023 & 0.032 & 0.046 & 0.116 & 0.062 \\
200 & $[0,14]$ & 0.000 & 0.5 & 9 & 0.027 & 0.053 & 0.615 & 1.260 & 0.609 \\
200 & $[0,14]$ & 0.000 & 1.0 & 14 & 0.027 & 0.042 & 0.426 & 0.980 & 0.423 \\
200 & $[0,14]$ & 0.000 & 2.0 & 22 & 0.027 & 0.035 & 0.340 & 0.770 & 0.339 \\
\hline
1000 & $[0,2.1]$ & 0.003 & 0.5 & 25 & 0.011 & 0.024 & 0.040 & 0.105 & 0.052 \\
1000 & $[0,2.1]$ & 0.003 & 1.0 & 40 & 0.011 & 0.018 & 0.030 & 0.084 & 0.039 \\
1000 & $[0,2.1]$ & 0.003 & 2.0 & 63 & 0.011 & 0.015 & 0.022 & 0.066 & 0.029 \\
1000 & $[0,14]$ & 0.000 & 0.5 & 25 & 0.012 & 0.026 & 0.298 & 0.700 & 0.298 \\
1000 & $[0,14]$ & 0.000 & 1.0 & 40 & 0.012 & 0.020 & 0.230 & 0.560 & 0.232 \\
1000 & $[0,14]$ & 0.000 & 2.0 & 63 & 0.012 & 0.017 & 0.173 & 0.441 & 0.176 \\
\hline
\end{tabular}

        \end{adjustbox}
        \caption{Error decomposition for the private EHCDF with $m=m^*(\varepsilon)$ and {the zero-inflated $\mathcal{B}'(1,5)$ with $p_0=0.25$}, for $n=200,1000$ and $c=1.5,10$. The table reports the estimation errors of the ECDF and the noiseless EHCDF, the privacy error, its bound $mB/(n\varepsilon)$ from Theorem~\ref{Theo: privacy}, and the total error.}
        \label{tab:dp_sim}
    \end{table}

\subsection{Application of the private EHCDF to credit card expenditure data}\label{sec9}
We consider the \texttt{CreditCard} data set from the \texttt{R} package \texttt{AER} \citep{kleiber2020package}. The data contain the average monthly credit card expenditure, in U.S. dollars (USD), for $n=1319$ individuals. Approximately $24\%$ of the observations ($317$ of $1319$) are zero, and the maximum value is below $3200$.

Credit card expenditure is a typical example of data that may be released privately. In this setting, an upper bound on monthly expenditure may be available in advance, for example through a credit limit, without inspecting the data. We therefore consider the public intervals $[0,5000]$ and $[0,25000]$. Neither interval clips any observation, so the target is the ECDF $\mathbb F_n$ itself. We use $\varepsilon\in\{0.5,1,5\}$, which gives $m^*=31,48,140$ and $k^*=12$.
The expected Wasserstein distance between the private release and the ECDF, i.e., $\E W_1(\cdot,\mathbb F_n)$, is approximated by randomly generating $1000$ independent private releases.

\begin{table}[H]
	\centering\footnotesize
	\begin{adjustbox}{max width=0.95\textwidth}
		\begin{tabular}{c c c c cc cc}
\hline
 & & & & \multicolumn{2}{c}{public rules} & \multicolumn{2}{c}{best resolution} \\
\cline{5-6}\cline{7-8}
$[0,b]$ & $\varepsilon$ & $m^*$ & $k^*$ & EHCDF & Histogram & EHCDF ($m^\dagger$) & Histogram ($k^\dagger$) \\
\hline
$[0,5000]$ & 0.5 & 31 & 12 & 94.4 & 130.4 & 64.9 (12) & 104.2 (20) \\
$[0,5000]$ & 1 & 48 & 12 & 69.0 & 113.5 & 49.7 (25) & 68.9 (25) \\
$[0,5000]$ & 5 & 140 & 12 & 32.5 & 100.3 & 24.0 (75) & 28.9 (60) \\
$[0,25000]$ & 0.5 & 31 & 12 & 526.8 & 1366.5 & 127.3 (3) & 762.0 (25) \\
$[0,25000]$ & 1 & 48 & 12 & 398.4 & 1151.4 & 94.9 (5) & 524.2 (40) \\
$[0,25000]$ & 5 & 140 & 12 & 215.8 & 883.4 & 48.6 (20) & 203.4 (75) \\
\hline
\end{tabular}

	\end{adjustbox}
	\caption{Credit card expenditure: mean error $\E\,W_1(\cdot,\mathbb F_n)$, in USD, for the private EHCDF at $m^*$ and the perturbed histogram at $k^*$, as well as for both mechanisms at their best grid resolutions, given in parentheses. Results are reported for the public intervals $[0,5000]$ and $[0,25000]$, based on $1000$ independent noise replications.}
	\label{tab:dp_credit}
\end{table}

Table~\ref{tab:dp_credit} shows that the private EHCDF is closer to the ECDF than the private histogram at every privacy level, both at the default resolutions and at the best grid resolutions. For example, on $[0,5000]$ and at $\varepsilon=1$, the error of the EHCDF at $m^*=48$ is approximately \new{$69$} USD, compared with approximately $113$ USD for the histogram at $k^*=12$. At the best resolutions, the corresponding errors are approximately \new{$50$} and $69$ USD. On $[0,25000]$ and at $\varepsilon=1$, the histogram error is approximately \new{$1151$} USD at $k^*$ and $524$ USD at its best resolution. The corresponding EHCDF errors are $398$ USD at $m^*$ and $95$ USD at its best resolution, obtained with only $m=5$.

The reason is that the data occupy only a small part of the public interval. With equally spaced bins, the histogram at $k^*=12$ places $87\%$ of the observations in its first bin on $[0,5000]$, and $99.8\%$ on $[0,25000]$, where only $2$ of the $12$ bins are non-empty. The noise added to the empty bins provides no useful information and spreads mass over regions without data. By contrast, the jump points of the EHCDF adapt to the observed data, while only its privacy term, $mB/(n\varepsilon)$, depends on the width of the public interval. Unequal-width bins could reduce this difference, but their construction would require public information about the shape of the distribution. Such information may not be available to the curator.

Table~\ref{tab:dp_functionals_credit} reports the mean absolute errors of the mean, the percentiles $P_{10}$, $P_{50}$ (median) and $P_{90}$, the mean absolute difference (MAD), and the Gini index, all computed from the post-processed releases. For the EHCDF, the MAD and Gini estimates are multiplied by the exact correction factor $m/(m-1)$ described in Section~\ref{sec3}. On $[0,5000]$, the EHCDF is more accurate for the mean, the median, and the Gini index at every value of $\varepsilon$, with improvements ranging from approximately a factor of two to a factor of eight. The histogram is more accurate for the MAD and, when $\varepsilon\leq 1$, for $P_{90}$. On $[0,25000]$, the EHCDF is more accurate for every functional. The error of the EHCDF for $P_{10}$ is zero {or nearly zero} because the sample $P_{10}$ is zero and the noisy atoms are clipped at zero. In contrast, the histogram spreads the mass of the first bin uniformly over that bin. Figure~\ref{fig:dp_credit} shows one realisation of the two private releases at $\varepsilon=1$, together with the ECDF and the corresponding Lorenz curves, which are commonly used to summarise expenditure distributions; see, for example, \citet{bruhin2025consumer} for a recent application to transactional payment data. For comparison, the figure also shows a realisation of the histogram with $k=m^*=48$ bins, of width approximately $104$ USD. This release is closer to the ECDF over most of its range, but the noise added to the empty bins places part of the mass above the largest observation.

\begin{table}[H]
	\centering\footnotesize
	\begin{adjustbox}{max width=0.95\textwidth}
		\begin{tabular}{c c l cccccc}
\hline
$[0,b]$ & $\varepsilon$ & Mechanism & Mean & $P_{10}$ & Median & $P_{90}$ & MAD & Gini \\
\hline
$[0,5000]$ & 0.5 & EHCDF ($m^*=31$) & 73.2 & 0.0 & 62.5 & 190.4 & 89.3 & 0.044 \\
 &  & Histogram ($k^*=12$) & 131.3 & 48.2 & 139.9 & 100.7 & 68.0 & 0.159 \\
[2pt]
$[0,5000]$ & 1 & EHCDF ($m^*=48$) & 52.0 & 0.0 & 40.2 & 169.7 & 59.6 & 0.034 \\
 &  & Histogram ($k^*=12$) & 113.4 & 47.9 & 138.4 & 77.0 & 35.7 & 0.183 \\
[2pt]
$[0,5000]$ & 5 & EHCDF ($m^*=140$) & 25.3 & 0.0 & 24.0 & 48.9 & 22.6 & 0.025 \\
 &  & Histogram ($k^*=12$) & 99.6 & 47.7 & 137.4 & 61.6 & 10.3 & 0.205 \\
[2pt]
$[0,25000]$ & 0.5 & EHCDF ($m^*=31$) & 512.2 & 0.0 & 162.5 & 1502.5 & 757.0 & 0.086 \\
 &  & Histogram ($k^*=12$) & 1362.4 & 272.9 & 1263.4 & 1986.1 & 1034.5 & 0.220 \\
[2pt]
$[0,25000]$ & 1 & EHCDF ($m^*=48$) & 381.4 & 0.0 & 110.6 & 1283.7 & 556.8 & 0.072 \\
 &  & Histogram ($k^*=12$) & 1181.9 & 254.3 & 1170.4 & 1818.8 & 801.5 & 0.253 \\
[2pt]
$[0,25000]$ & 5 & EHCDF ($m^*=140$) & 204.6 & 0.0 & 64.1 & 634.1 & 288.9 & 0.041 \\
 &  & Histogram ($k^*=12$) & 883.7 & 209.6 & 946.6 & 1416.0 & 504.9 & 0.289 \\
[2pt]
\hline
\end{tabular}

	\end{adjustbox}
	\caption{Credit card expenditure: mean absolute errors of functionals computed from the private releases on $[0,5000]$ and $[0,25000]$. Errors are measured in USD, except for the dimensionless Gini index. Results are reported for the private EHCDF at $m^*$ and the perturbed histogram at $k^*$. For the EHCDF, the MAD and Gini estimates are corrected by the factor $m/(m-1)$. Based on $1000$ noise replications.}
	\label{tab:dp_functionals_credit}
\end{table}

\begin{figure}[H]
	\centering
	\includegraphics[width=\textwidth]{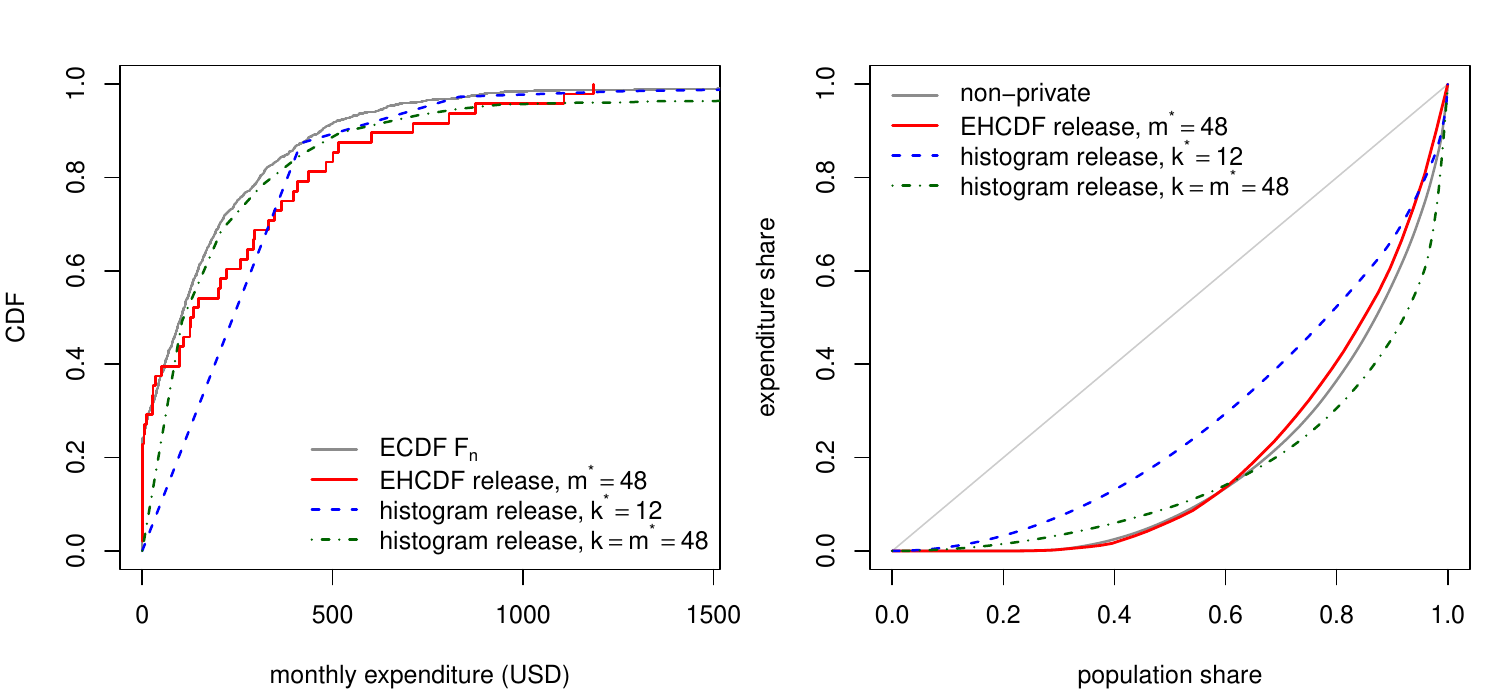}
	\caption{Credit card expenditure for the public interval $[0,5000]$ and $\varepsilon=1$. Left: the ECDF and one realisation of the private EHCDF with $m^*=48$, together with one realisation of the perturbed histogram with $k^*=12$ and one with $k=m^*=48$. Right: the corresponding Lorenz curves.}
	\label{fig:dp_credit}
\end{figure}

\section*{Acknowledgments}
L.T. acknowledges the support of the INdAM group ``Gruppo Nazionale per il Calcolo Scientifico'' (GNCS), of which he is a member. The authors used AI tools to assist with language editing and improvement of presentation.

\FloatBarrier
\section*{Conflicts of Interest}
The authors declare no conflicts of interest.

\section*{Data Availability Statement}
The credit card expenditure data analysed in Section~\ref{sec9} are publicly available.

\bibliographystyle{plainnat}
\bibliography{references}

\begin{thebibliography}{35}
\providecommand{\natexlab}[1]{#1}
\providecommand{\url}[1]{\texttt{#1}}
\expandafter\ifx\csname urlstyle\endcsname\relax
  \providecommand{\doi}[1]{doi: #1}\else
  \providecommand{\doi}{doi: \begingroup \urlstyle{rm}\Url}\fi

\bibitem[Arnold and Meeden(1975)]{arnold1975}
Barry~C Arnold and Glen Meeden.
\newblock Characterization of distributions by sets of moments of order
  statistics.
\newblock \emph{The Annals of Statistics}, pages 754--758, 1975.

\bibitem[Azzalini(1981)]{Azzalini1981}
Adelchi Azzalini.
\newblock A note on the estimation of a distribution function and quantiles by
  a kernel method.
\newblock \emph{Biometrika}, 68\penalty0 (1):\penalty0 326--328, 1981.
\newblock \doi{10.1093/biomet/68.1.326}.

\bibitem[Babu et~al.(2002)Babu, Canty, and Chaubey]{BabuCantyChaubey2002}
G.~Jogesh Babu, Angelo~J. Canty, and Yogendra~P. Chaubey.
\newblock Application of {Bernstein} polynomials for smooth estimation of a
  distribution and density function.
\newblock \emph{Journal of Statistical Planning and Inference}, 105\penalty0
  (2):\penalty0 377--392, 2002.
\newblock \doi{10.1016/S0378-3758(01)00265-8}.

\bibitem[Barczewski et~al.(2025)Barczewski, Mawass, and
  Ramon]{BarczewskiMawassRamon2025}
Antoine Barczewski, Amal Mawass, and Jan Ramon.
\newblock Differentially private empirical cumulative distribution functions,
  2025.
\newblock URL \url{https://arxiv.org/abs/2502.06651}.

\bibitem[Beare and Kaji(2026)]{beare2025necessary}
Brendan~K Beare and Tetsuya Kaji.
\newblock A necessary and sufficient condition for convergence in distribution
  of the quantile process in ${L}^{1} (0, 1) $.
\newblock \emph{Bernoulli (forthcoming)}, 2026.

\bibitem[Borisov(2014)]{borisov2014}
IS~Borisov.
\newblock A note on a result by {W}. {H}oeffding.
\newblock \emph{Statistics \& Probability Letters}, 87:\penalty0 7--11, 2014.

\bibitem[Borisov(2016)]{borisov2016}
IS~Borisov.
\newblock The rate of convergence in {H}oeffding's theorem and some
  applications.
\newblock \emph{Statistics \& Probability Letters}, 109:\penalty0 99--105,
  2016.

\bibitem[Bruhin et~al.(2025)Bruhin, Fengler, Koeniger, and
  Rohrkemper]{bruhin2025consumer}
Jonas~M Bruhin, Matthias~R Fengler, Winfried Koeniger, and Robert Rohrkemper.
\newblock Consumer spending in switzerland: insights from a novel transactional
  data index.
\newblock \emph{Swiss Journal of Economics and Statistics}, 161\penalty0
  (1):\penalty0 14, 2025.

\bibitem[Del~Barrio et~al.(2005)Del~Barrio, Gin{\'e}, and Utzet]{delbarrio}
Eustasio Del~Barrio, Evarist Gin{\'e}, and Frederic Utzet.
\newblock Asymptotics for ${L}^2$ functionals of the empirical quantile
  process, with applications to tests of fit based on weighted {W}asserstein
  distances.
\newblock \emph{Bernoulli}, 11\penalty0 (1):\penalty0 131--189, 2005.

\bibitem[Dwork and Roth(2014)]{DworkRoth2014}
Cynthia Dwork and Aaron Roth.
\newblock The algorithmic foundations of differential privacy.
\newblock \emph{Foundations and Trends in Theoretical Computer Science},
  9\penalty0 (3--4):\penalty0 211--407, 2014.
\newblock \doi{10.1561/0400000042}.

\bibitem[Einmahl and Mason(2005)]{einmahl2005uniform}
Uwe Einmahl and David~M Mason.
\newblock Uniform in bandwidth consistency of kernel-type function estimators.
\newblock \emph{The Annals of Statistics}, 33\penalty0 (3):\penalty0
  1380--1403, 2005.

\bibitem[Geenens(2014)]{Geenens2014}
G.~Geenens.
\newblock Probit transformation for kernel density estimation on the unit
  interval.
\newblock \emph{Journal of the American Statistical Association}, 109\penalty0
  (505):\penalty0 346--358, 2014.
\newblock \doi{10.1080/01621459.2013.872718}.

\bibitem[Geenens and Wang(2018)]{GeenensWang2018}
G.~Geenens and C.~Wang.
\newblock Local-likelihood transformation kernel density estimation for
  positive random variables.
\newblock \emph{Journal of Computational and Graphical Statistics}, 27\penalty0
  (3):\penalty0 620--633, 2018.
\newblock \doi{10.1080/10618600.2017.1390465}.

\bibitem[Giddens and Liu(2023)]{giddens2023dppack}
Spencer Giddens and Fang Liu.
\newblock {DP}pack: An {R} package for differentially private statistical
  analysis and machine learning.
\newblock \emph{arXiv preprint arXiv:2309.10965}, 2023.

\bibitem[Gillenwater et~al.(2021)Gillenwater, Joseph, and
  Kulesza]{GillenwaterJosephKulesza2021}
Jennifer Gillenwater, Matthew Joseph, and Alex Kulesza.
\newblock Differentially private quantiles.
\newblock In Marina Meila and Tong Zhang, editors, \emph{Proceedings of the
  38th International Conference on Machine Learning}, volume 139 of
  \emph{Proceedings of Machine Learning Research}, pages 3713--3722. PMLR,
  2021.
\newblock URL \url{https://proceedings.mlr.press/v139/gillenwater21a.html}.

\bibitem[Hoeffding(1953)]{hoeffding1953}
Wassily Hoeffding.
\newblock On the distribution of the expected values of the order statistics.
\newblock \emph{The Annals of Mathematical Statistics}, pages 93--100, 1953.

\bibitem[Hosking(1990)]{Hosking1990}
J.~R.~M. Hosking.
\newblock L-moments: Analysis and estimation of distributions using linear
  combinations of order statistics.
\newblock \emph{Journal of the Royal Statistical Society: Series B
  (Methodological)}, 52\penalty0 (1):\penalty0 105--124, 1990.

\bibitem[Kleiber et~al.(2020)Kleiber, Zeileis, and Zeileis]{kleiber2020package}
Christian Kleiber, Achim Zeileis, and Maintainer~Achim Zeileis.
\newblock Package `aer'.
\newblock \emph{R package version}, 1\penalty0 (4), 2020.

\bibitem[Lando and Benjrada(2025)]{lando2025new}
Tommaso Lando and Mohammed Es-Salih Benjrada.
\newblock A new class of tests for convex-ordered families based on expected
  order statistics.
\newblock \emph{Electronic Journal of Statistics}, 19\penalty0 (1):\penalty0
  2780--2802, 2025.

\bibitem[Mason and Swanepoel(2011)]{mason2011}
David~M Mason and Jan~WH Swanepoel.
\newblock A general result on the uniform in bandwidth consistency of
  kernel-type function estimators.
\newblock \emph{Test}, 20\penalty0 (1):\penalty0 72--94, 2011.

\bibitem[Mason and Swanepoel(2013)]{mason2013}
David~M Mason and Jan~WH Swanepoel.
\newblock Uniform in bandwidth limit laws for kernel distribution function
  estimators.
\newblock In \emph{From Probability to Statistics and Back: High-Dimensional
  Models and Processes--A Festschrift in Honor of Jon A. Wellner}, volume~9,
  pages 241--254. Institute of Mathematical Statistics, 2013.

\bibitem[Nagler(2018{\natexlab{a}})]{Nagler2018a}
T.~Nagler.
\newblock A generic approach to nonparametric function estimation with mixed
  data.
\newblock \emph{Statistics \& Probability Letters}, 137:\penalty0 326--330,
  2018{\natexlab{a}}.
\newblock \doi{10.1016/j.spl.2018.02.030}.

\bibitem[Nagler(2018{\natexlab{b}})]{Nagler2018b}
T.~Nagler.
\newblock Asymptotic analysis of the jittering kernel density estimator.
\newblock \emph{Mathematical Methods of Statistics}, 27\penalty0 (3):\penalty0
  177--196, 2018{\natexlab{b}}.
\newblock \doi{10.3103/S1066530718030040}.

\bibitem[Nagler et~al.(2025)Nagler, Vatter, and Nagler]{nagler2025package}
Thomas Nagler, Thibault Vatter, and Maintainer~Thomas Nagler.
\newblock \emph{Package `kde1d'}, 2025.

\bibitem[Okolewski and Papadatos(2024)]{okolewski2024}
A~Okolewski and N~Papadatos.
\newblock Finite sequences representing expected order statistics.
\newblock \emph{Sankhya A}, 86\penalty0 (2):\penalty0 755--774, 2024.

\bibitem[Serfling(2009)]{serfling}
Robert~J Serfling.
\newblock \emph{Approximation theorems of mathematical statistics}.
\newblock John Wiley \& Sons, 2009.

\bibitem[Sheather and Jones(1991)]{SheatherJones1991}
S.~J. Sheather and M.~C. Jones.
\newblock A reliable data-based bandwidth selection method for kernel density
  estimation.
\newblock \emph{Journal of the Royal Statistical Society: Series B
  (Methodological)}, 53\penalty0 (3):\penalty0 683--690, 1991.

\bibitem[Silverman and Young(1987)]{SilvermanYoung1987}
B.~W. Silverman and G.~A. Young.
\newblock The bootstrap: To smooth or not to smooth?
\newblock \emph{Biometrika}, 74\penalty0 (3):\penalty0 469--479, 1987.
\newblock \doi{10.1093/biomet/74.3.469}.

\bibitem[Stupfler(2016)]{stupfler}
Gilles Stupfler.
\newblock On the weak convergence of the kernel density estimator in the
  uniform topology.
\newblock \emph{Electronic Communications in Probability}, 21:\penalty0 1--13,
  2016.

\bibitem[Sturges(1926)]{sturges1926choice}
Herbert~A Sturges.
\newblock The choice of a class interval.
\newblock \emph{Journal of the american statistical association}, 21\penalty0
  (153):\penalty0 65--66, 1926.

\bibitem[Tao and Sarwate(2026)]{TaoSarwate2026}
Ye~Tao and Anand~D. Sarwate.
\newblock Functional approximation methods for differentially private
  distribution estimation.
\newblock \emph{IEEE Transactions on Information Forensics and Security}, 2026.
\newblock To appear.

\bibitem[Van~der Vaart and Wellner(1996)]{vw}
Aad Van~der Vaart and Jon Wellner.
\newblock \emph{Weak convergence and empirical processes: with applications to
  statistics}.
\newblock Springer Science \& Business Media, 1996.

\bibitem[Van~Zwet(1980)]{zwet1980}
Willem~R Van~Zwet.
\newblock A strong law for linear functions of order statistics.
\newblock \emph{The Annals of Probability}, 5\penalty0 (8):\penalty0 986--990,
  1980.

\bibitem[Wand et~al.(1991)Wand, Marron, and Ruppert]{wand1991transformations}
Matthew~P Wand, James~Stephen Marron, and David Ruppert.
\newblock Transformations in density estimation.
\newblock \emph{Journal of the American Statistical Association}, 86\penalty0
  (414):\penalty0 343--353, 1991.

\bibitem[Wasserman and Zhou(2010)]{WassermanZhou2010}
Larry Wasserman and Shuheng Zhou.
\newblock A statistical framework for differential privacy.
\newblock \emph{Journal of the American Statistical Association}, 105\penalty0
  (489):\penalty0 375--389, 2010.
\newblock \doi{10.1198/jasa.2009.tm08651}.

\end{thebibliography}

\newpage
\appendix
\setcounter{figure}{0}
\setcounter{table}{0}
\renewcommand{\thefigure}{A\arabic{figure}}
\renewcommand{\thetable}{A\arabic{table}}

\section{Proofs}\label{appendix:proofs}
\subsection{Proof of Lemma~\ref{lemma:Im_contraction}}
\label{proof:Im_contraction}
\begin{proof}
	Linearity of $\mathcal{I}_m$ is immediate from its definition. For the norm bound, we compute
	\[
	\|\mathcal{I}_m(h)\|_p^p = \frac{1}{m}\sum_{j=1}^m \left| \int_0^1 h(s)\, dF_{B_{j:m}}(s) \right|^p.
	\]
	Jensen's inequality gives
	\[
	\left| \int_0^1 h(s)\, dF_{B_{j:m}}(s) \right|^p \le \int_0^1 |h(s)|^p\, dF_{B_{j:m}}(s).
	\]
	Hence,
	\[
	\|\mathcal{I}_m(h)\|_p^p \le \frac{1}{m}\sum_{j=1}^m \int_0^1 |h(s)|^p\, f_{B_{j:m}}(s)\, ds=\int_0^1 |h(s)|^p \left( \frac{1}{m}\sum_{j=1}^m f_{B_{j:m}}(s) \right) ds.
	\]
	Using $
	\frac{1}{m}\sum_{j=1}^m f_{B_{j:m}}(s) = 1,$ for every $s \in (0,1),$ and taking the $p$-th root, we conclude
	\[
	\|\mathcal{I}_m(h)\|_p \le \bigg(\int_0^1 |h(s)|^p\, ds\bigg)^{1/p} = \|h\|_p.
	\]
\hfill\end{proof}
\subsection{Proof of Theorem~\ref{thm:Hm_L1_contraction}}
	\begin{proof}
    We start by recalling the following property of the beta density, which will be often used in the proofs.
\begin{equation}\label{eq:partition_of_unity}
  \frac{1}{m}\sum_{j=1}^{m} f_{B_{j:m}}(u) = 1, \qquad \forall u\in [0,1].
\end{equation}
In other words, the family $\{m^{-1}f_{B_{j:m}}\}_{j=1}^m$ forms a partition of unity.

	Let $\mathcal H_m(F)=F_m$ and $\mathcal H_m(G)=G_m$. For every $x\in\mathbb R$,
	\[
	F_m(x)-G_m(x)=\frac1m\sum_{j=1}^m\Big(\ind\{\mu_{j:m}(F)\le x\}-\ind\{\mu_{j:m}(G)\le x\}\Big).
	\]
	By convexity of $t\mapsto |t|^p$, Jensen's inequality implies that
	\begin{align*}
		\|F_m-G_m\|_p^p
		&=
		\int_{\mathbb R}\left|
		\frac1m\sum_{j=1}^m\Big(\ind\{\mu_{j:m}(F)\le x\}-\ind\{\mu_{j:m}(G)\le x\}\Big)
		\right|^p dx \\
		&\le
		\frac1m\sum_{j=1}^m
		\int_{\mathbb R}
		\left|
		\ind\{\mu_{j:m}(F)\le x\}-\ind\{\mu_{j:m}(G)\le x\}
		\right|^p dx \\
		&=
		\frac1m\sum_{j=1}^m
		\int_{\mathbb R}
		\left|
		\ind\{\mu_{j:m}(F)\le x\}-\ind\{\mu_{j:m}(G)\le x\}
		\right| dx \\
		&=
		\frac1m\sum_{j=1}^m
		|\mu_{j:m}(F)-\mu_{j:m}(G)|.
	\end{align*}
	Now,
	\[
	|\mu_{j:m}(F)-\mu_{j:m}(G)| = \left| \int_0^1 (F^{-1}-G^{-1})\,dF_{B_{j:m}} \right| \le \int_0^1 |F^{-1}-G^{-1}|\,dF_{B_{j:m}}.
	\]
	Therefore, using (\ref{eq:partition_of_unity}),
	\begin{align*}
		\frac1m\sum_{j=1}^m |\mu_{j:m}(F)-\mu_{j:m}(G)|
		&\le
		\frac1m\sum_{j=1}^m
		\int_0^1 |F^{-1}(u)-G^{-1}(u)|\,dF_{B_{j:m}}(u) \\
		&=
		\int_0^1 |F^{-1}(u)-G^{-1}(u)|
		\left(\frac1m\sum_{j=1}^m f_{B_{j:m}}(u)\right)du \\
		&=
		\int_0^1 |F^{-1}(u)-G^{-1}(u)|\,du =\|F-G\|_1.
	\end{align*}
\hfill\end{proof}
\subsection{Proof of Theorem~\ref{teo:Im_convergence}}
    \begin{proof}
	By the triangle inequality,
	$$
	\|\mathcal{I}_m(h_n)-h\|_p \le \|\mathcal{I}_m(h_n)-\mathcal{I}_m(h)\|_p + \|\mathcal{I}_m(h)-h\|_p.
	$$
	For the first term, we apply Lemma~1 and linearity of $\mathcal{I}_m$ to obtain:
	$$
	\|\mathcal{I}_m(h_n)-\mathcal{I}_m(h)\|_p = \|\mathcal{I}_m(h_n-h)\|_p \le \|h_n-h\|_p \to 0 \qquad\text{as }n\to\infty.
	$$
	We now prove that $ \|\mathcal{I}_m(h)-h\|_p\to 0$ as $m\to\infty.$ Write
	$$
	\|\mathcal{I}_m(h)-h\|_p^p = \sum_{j=1}^m \int_{\frac{j-1}{m}}^{\frac{j}{m}} \left| \int_0^1 h(s)\,f_{B_{j:m}}(s)\,ds - h(u) \right|^p du.
	$$
	Since $\int_0^1 f_{B_{j:m}}(s)\,ds=1$, we may write
	$$
	h(u)=h(u)\int_0^1 f_{B_{j:m}}(s)\,ds,
	$$
	and therefore
	$$
	\int_0^1 h(s)\,f_{B_{j:m}}(s)\,ds-h(u) = \int_0^1 (h(s)-h(u))\,f_{B_{j:m}}(s)\,ds.
	$$
	By Jensen's inequality,
	$$
	\left| \int_0^1 (h(s)-h(u))\,f_{B_{j:m}}(s)\,ds \right|^p \le \int_0^1 |h(s)-h(u)|^p\,f_{B_{j:m}}(s)\,ds.
	$$
	Hence
	$$
	\|\mathcal{I}_m(h)-h\|_p^p \le \sum_{j=1}^m \int_{\frac{j-1}{m}}^{\frac{j}{m}} \int_0^1 |h(s)-h(u)|^p\,f_{B_{j:m}}(s)\,ds\,du.
	$$
	Now define
	$$
	K_m(s,u) := \sum_{j=1}^m f_{B_{j:m}}(s)\, \ind\{u\in\left(\tfrac{j-1}{m},\tfrac{j}{m}\right]\},
	$$
    the integral kernel associated with the operator $\mathcal I_m$, i.e.,
$
\mathcal I_m(h)(u)=\int_0^1 h(s)\,K_m(s,u)\,ds.
$
	Then
	$$
	\|\mathcal{I}_m(h)-h\|_p^p \le \int_0^1\int_0^1 |h(s)-h(u)|^p\,K_m(s,u)\,ds\,du.
	$$
	The kernel $K_m$ satisfies the two marginal identities
	$$
	\int_0^1 K_m(s,u)\,du = 1, \quad\forall s\in(0,1), \quad\text{and}\quad \int_0^1 K_m(s,u)\,ds = 1, \quad\forall u\in(0,1).
	$$
	Indeed, for fixed $s$,
	$$
	\int_0^1 K_m(s,u)\,du = \sum_{j=1}^m f_{B_{j:m}}(s)\int_0^1 \ind\{u\in\left(\tfrac{j-1}{m},\tfrac{j}{m}\right]\}\,du = \frac{1}{m}\sum_{j=1}^m f_{B_{j:m}}(s) =1,
	$$
	while for fixed $u$ there is exactly one $j$ such that
	$
	u\in\left(\frac{j-1}{m},\frac{j}{m}\right],
	$
	so that
	$
	K_m(s,u)=f_{B_{j:m}}(s),
	$
	and therefore
	$$
	\int_0^1 K_m(s,u)\,ds=\int_0^1 f_{B_{j:m}}(s)\,ds=1.
	$$
	Fix $\varepsilon>0$. Since $C[0,1]$ is dense in $L^p(0,1)$, there exists a bounded continuous function $h_\varepsilon$ such that
	$
	\|h-h_\varepsilon\|_p\le \varepsilon.
	$
	Write
	$$
	h(s)-h(u) = \bigl(h(s)-h_\varepsilon(s)\bigr) + \bigl(h_\varepsilon(s)-h_\varepsilon(u)\bigr) + \bigl(h_\varepsilon(u)-h(u)\bigr).
	$$
	Using Jensen's inequality again, that is, $|\sum_{j=1}^k a_j|^p\leq k^{p-1}(|\sum_{j=1}^k|a_j|^p) $, for $k=3$ we obtain
	$$
	|h(s)-h(u)|^p \le 3^{p-1}\Bigl( |h(s)-h_\varepsilon(s)|^p + |h_\varepsilon(s)-h_\varepsilon(u)|^p + |h_\varepsilon(u)-h(u)|^p \Bigr).
	$$
	Integrating against $K_m(s,u)\,ds\,du$ yields
	$$
	\int_0^1\int_0^1 |h(s)-h(u)|^p\,K_m(s,u)\,ds\,du \le 3^{p-1}(A_m+B_m+C_m),
	$$
	where
	$$
	A_m = \int_0^1\int_0^1 |h(s)-h_\varepsilon(s)|^p\,K_m(s,u)\,ds\,du,
	$$
	$$
	B_m = \int_0^1\int_0^1 |h_\varepsilon(s)-h_\varepsilon(u)|^p\,K_m(s,u)\,ds\,du,
	$$
	$$
	C_m = \int_0^1\int_0^1 |h_\varepsilon(u)-h(u)|^p\,K_m(s,u)\,ds\,du.
	$$
We now use the fact that $K_m$ has both marginals equal to $1$ to show that the error arising from approximating $h$ by $h_\varepsilon$ is controlled by $\|h-h_\varepsilon\|_p^p$. In fact, since the integrand in $A_m$ depends only on $s$,
	$$
	A_m = \int_0^1 |h(s)-h_\varepsilon(s)|^p \left(\int_0^1 K_m(s,u)\,du\right) ds = \int_0^1 |h(s)-h_\varepsilon(s)|^p\,ds = \|h-h_\varepsilon\|_p^p.
	$$
	Similarly,
	$
	C_m
	=
	\|h-h_\varepsilon\|_p^p.
	$
	Hence,
	$
	A_m+C_m = 2\|h-h_\varepsilon\|_p^p \le 2\varepsilon^p.
	$

	It remains to deal with $B_m$. Since $h_\varepsilon$ is (uniformly) continuous on $[0,1]$, for every $\delta>0$, there exists $\eta>0$ such that
	$$
	|s-u|\le \eta \quad\Longrightarrow\quad |h_\varepsilon(s)-h_\varepsilon(u)|\le \delta.
	$$
	Therefore
	$$
	B_m \le \delta^p + (2\|h_\varepsilon\|_\infty)^p R_m(\eta),
	$$
	where
	$$
	R_m(\eta) := \int_0^1\int_0^1 \ind\{|s-u|>\eta\}\,K_m(s,u)\,ds\,du.
	$$
	We now show that $\lim_{m\to\infty}R_m(\eta)= 0$. For $u\in\left(\frac{j-1}{m},\frac{j}{m}\right]$, one has
	$
	K_m(s,u)=f_{B_{j:m}}(s).
	$
	Recall that $f_{B_{j:m}}$ is the density of $U_{j:m}$, where
	$
	\mathbb E[U_{j:m}] = \frac{j}{m+1}.
	$
	Moreover, if
	$
	u\in\left(\frac{j-1}{m},\frac{j}{m}\right],
	$
	then
	$
	\left|u-\frac{j}{m+1}\right|\le \frac{1}{m}.
	$
	Hence, if $|s-u|>\eta$, the reverse triangle inequality yields
	$$
	\left|s-\frac{j}{m+1}\right| \ge |s-u|-\left|u-\frac{j}{m+1}\right| > \eta-\frac{1}{m}.
	$$
	Therefore, for $m$ large enough so that $\frac1m<\eta/2$,
	$$
	|s-u|>\eta \quad\Longrightarrow\quad \left|s-\frac{j}{m+1}\right|>\frac{\eta}{2}.
	$$
	The above implication, in addition to Chebyshev's inequality, gives
	$$
	\int_0^1 \ind\{|s-u|>\eta\}\,f_{B_{j:m}}(s)\,ds \le P\left(\left|U_{j:m}-\frac{j}{m+1}\right|>\frac{\eta}{2}\right)\le \frac{\mathrm{Var}(U_{j:m})}{(\eta/2)^2}.
	$$
	Since
	$$
	\mathrm{Var}(U_{j:m}) = \frac{j(m-j+1)}{(m+1)^2(m+2)} \le \frac{1}{4(m+2)},
	$$
	we obtain
	$$
	P\left(\left|U_{j:m}-\frac{j}{m+1}\right|>\frac{\eta}{2}\right) \le \frac{1}{\eta^2(m+2)}.
	$$
	This bound is uniform in $j$. Therefore, for $m\to\infty,$
	$$
	R_m(\eta) \le \sum_{j=1}^m \int_{\frac{j-1}{m}}^{\frac{j}{m}} \frac{1}{\eta^2(m+2)}\,du = \frac{1}{\eta^2(m+2)} \sum_{j=1}^m \frac{1}{m} = \frac{1}{\eta^2(m+2)} \to 0.
	$$
	Hence
	$
	\limsup_{m\to\infty} B_m \le \delta^p.
	$
	Since $\delta>0$ is arbitrary, it follows that $B_m\to 0$. Thus
	$$
	\limsup_{m\to\infty}\|\mathcal{I}_m(h)-h\|_p^p \le 3^{p-1}\cdot 2\varepsilon^p.
	$$
	Letting $\varepsilon\to 0$, we conclude that
	$
	\|\mathcal{I}_m(h)-h\|_p\to 0.
	$
	Combining this with the first term proves the theorem.
\hfill\end{proof}
\subsection{Proof of Theorem~\ref{thm:Hoeffding_Lfunctionals}}
\begin{proof}
	The quantile function $G_m^{-1}$ is the step function
	\[
	G_m^{-1}(u)=\mu_{j:m}(G)\quad\text{for }u\in\Bigl(\frac{j-1}{m},\frac{j}{m}\Bigr],\qquad j=1,\dots,m.
	\]
	Therefore,
	\[
	T_w(G_m)=\int_0^1 G_m^{-1}(u)\,w(u)\,du =\sum_{j=1}^m \mu_{j:m}(G)\int_{(j-1)/m}^{j/m} w(u)\,du =\sum_{j=1}^m c_j(w)\,\mu_{j:m}(G).
	\]
	Using the quantile representation of expected order statistics, we obtain
	\begin{align*}
		T_w(G_m)
		=\sum_{j=1}^m c_j(w)\int_0^1 G^{-1}(u)\,f_{B_{j:m}}(u)\,du\\
		=\int_0^1 G^{-1}(u)\Bigl(\sum_{j=1}^m c_j(w)\,f_{B_{j:m}}(u)\Bigr)\,du
		=\int_0^1 G^{-1}(u)\,(P_m [w])(u)\,du=T_{P_m[w]}(G).
	\end{align*}
\hfill\end{proof}
\subsection{Proof of Lemma~\ref{lemma:HCDF_CDF_convergence_rate}}
\label{proof:HCDF_CDF_convergence_rate}
\begin{proof}
	Recall that uniform order statistics $ U_{j:m}$ have CDF $F_{B_{j:m}}$. Using $\|F_m - F\|_1 = \| Q_m- Q\|_1$, where $Q_m =\mathcal{I}_m(Q)$, we want to show that
	\[
	 \| Q_m- Q\|_1 = \sum_{j=1}^m \int_{(j-1)/m}^{j/m} \bigl| \mathbb{E}[Q(U_{j:m})] - Q(x) \bigr| \, dx = O\left( m^{-1/2} \right) \quad \text{as } m \to \infty.
	\]
	We start from the identity
	\[
	\mathbb{E}[Q(U_{j:m})] = \int_0^1 Q(u) dF_{B_{j:m}}(u).
	\]
	Thus, for any \( x \in (0,1) \),
	\[
	\mathbb{E}[Q(U_{j:m})] - Q(x) = \int_0^1 \bigl(Q(u)-Q(x)\bigr)\, dF_{B_{j:m}}(u).
	\]
	Note that
	\[
	|Q(u) - Q(x)| = \int_0^1 \ind\{\min\{u,x\}<t\leq \max\{u,x\}\} dQ(t).
	\]
	Substituting this into the previous integral and applying Fubini's theorem yields
	\[
	\bigl|\mathbb{E}[Q(U_{j:m})]-Q(x)\bigr|\leq \E | Q(U_{j:m})-Q(x)| = \int_0^1  H_{j,m}(t;x)\, dQ(t),
	\]
	where
	\[
	H_{j,m}(t;x) := \int_0^1 \ind\{\min\{u,x\}<t\leq \max\{u,x\}\}\, dF_{B_{j:m}}(u)\, = P\bigl((U_{j:m}-t)(x-t)\le 0\bigr).
	\]
	Intuitively, \( H_{j,m}(t;x) \) measures the probability that \( U_{j:m} \) lies on the opposite side of \( t \) than \( x \) does. Integrating the previous inequality over \( x \in [(j-1)/m, j/m] \) and summing over \( j \) yields
	\begin{align*}
		\sum_{j=1}^m \int_{\tfrac{j-1}m}^{\tfrac{j}{m}} \! |\mathbb{E}[Q(U_{j:m})]-Q(x)|\,dx
		&\le \int_0^1 W_m(t)dQ(t)
	\end{align*}
	where
	\[
	W_m(t):= \sum_{j=1}^m \int_{\tfrac{j-1}m}^{\tfrac{j}{m}} H_{j,m}(t;x)\, dx = \sum_{j=1}^m \int_{\tfrac{j-1}m}^{\tfrac{j}{m}} P\bigl((U_{j:m}-t)(x-t)\le 0\bigr)\, dx.
	\]
	Hence, the problem reduces to finding a bound for \( W_m(t) \).

Write $r=\lfloor mt\rfloor$ and \(\lambda=mt-r\in[0,1)\), so that
$
t=\frac{r+\lambda}{m}\in\left[\frac rm,\frac{r+1}{m}\right).
$
Hence, the intervals \(\bigl((j-1)/m,j/m\bigr]\) split into three groups: those entirely to the left of \(t\), for which the event \((U_{j:m}-t)(x-t)\le 0\) reduces to \(U_{j:m}\ge t\); the unique interval containing \(t\); and those entirely to the right of \(t\), for which \((U_{j:m}-t)(x-t)\le 0\) reduces to \(U_{j:m}\le t\). The middle interval must then be split at \(t\), according to whether \(x\le t\) or \(x\ge t\), which yields
\begin{align*}
   W_m(t)
=
\frac1m\sum_{j=1}^{r} P(U_{j:m}\ge t)
+
\frac{\lambda}{m}P(U_{r+1:m}\ge t)\\
+
\frac{1-\lambda}{m}P(U_{r+1:m}\le t)
+
\frac1m\sum_{j=r+2}^{m} P(U_{j:m}\le t).
\end{align*}
Recall that
	$$P(U_{j:m}\leq t)=\sum_{k=j}^m\binom{m}{k}t^k(1-t)^{m-k},$$
	so, given a binomial random variable $B$ with $m$ trials and success probability $t$,
	\[
	P(U_{j:m} \le t) = P(B \ge j), \qquad P(U_{j:m} > t) = P(B \le j-1).
	\]
Using this binomial representation
\begin{align*}
   W_m(t)
=
\frac1m\sum_{j=1}^{r} P(B\le j-1)
+
\frac{\lambda}{m}P(B\le r)\\
+
\frac{1-\lambda}{m}P(B\ge r+1)
+
\frac1m\sum_{j=r+2}^{m} P(B\ge j).
\end{align*}
We now claim that
\[
W_m(t)=\frac1m\,\mathbb E|B-mt|=\frac 1m\sum_{b=0}^m |b-mt|P(B=b) .
\]
Indeed, we may write
\[
W_m(t)=\frac1m\sum_{b=0}^m c_b\,P(B=b),
\]
where
\[
c_b = \sum_{j=1}^{r} \ind\{b\le j-1\} +\lambda \ind\{b\le r\} +(1-\lambda)\ind\{b\ge r+1\} +\sum_{j=r+2}^{m}\ind\{b\ge j\}.
\]
It is readily seen that
$c_b=r-b+\lambda$ if $b\leq r$, and $c_b=b-r-\lambda$ for $b\geq r+1,$
hence \(c_b=|b-(r+\lambda)|=|b-mt|\).

Now, by Jensen's inequality,
\[
W_m(t) \le \frac1m\sqrt{\mathbb E(B-mt)^2} = \frac1m\sqrt{\mathrm{Var}(B)} = \sqrt{\frac{t(1-t)}{m}}.
\]
Thus
\[
W_m(t)=O\!\left(m^{-1/2}\sqrt{t(1-t)}\right).
\]
Recalling the assumption
$
\int_0^1 \sqrt{u(1-u)}\,dQ(u)<\infty,
$
we conclude that
\begin{align*}
   \sum_{j=1}^m \int_{(j-1)/m}^{j/m}
\bigl|\mathbb E[Q(U_{j:m})]-Q(x)\bigr|\,dx
\le
\int_0^1 W_m(t)\,dQ(t)\\
\le
\frac{1}{\sqrt m}\int_0^1 \sqrt{t(1-t)}\,dQ(t)
=
O(m^{-1/2}).
\end{align*}

\hfill\end{proof}
\subsection{Proof of Lemma~\ref{lemma:convergence_EHCDF_HCDF}}
\label{proof:convergence_EHCDF_HCDF}
\begin{proof}
	Let $m$ be fixed. By the triangle inequality,
	\begin{align*}
		||\F_{n,m}-F_{m}||_p\\
		=\| \frac1m\sum_{j=1}^m (\ind\{\hat{\mu}_{j:m}\leq \cdot\}-\ind\{{\mu}_{j:m}\leq \cdot\}\|_p\\
		\leq\frac1{m} \sum_{j=1}^m\bigg(\int_\R|\ind\{\mu_{j:m}\leq x\}- \ind\{\hat{\mu}_{j:m}\leq x\}|^p dx\bigg)^{\frac1p}.
	\end{align*}
		However, for $n\to\infty$,
		\begin{align*}
			\int_\R|\ind\{\mu_{j:m}\leq x\}- \ind\{\hat{\mu}_{j:m}\leq x\}|^pdx\\
			=\int_\R\ind\{\min(\mu_{j:m},\hat{\mu}_{j:m})\leq x\leq \max(\mu_{j:m},\hat{\mu}_{j:m})\}dx=|\hat{\mu}_{j:m}-{\mu}_{j:m}|\to_{a.s.}0,	\end{align*} therefore $||\F_{n,m}-F_m||_p\to_{a.s.}0,$ for $n\to\infty$.

		About the second assertion, consider
		$$||\F_{n,m}-F_{m}||_\infty=\sup_x|\frac1{m} \sum_{j=1}^m(\ind\{\mu_{j:m}\leq x\}- \ind\{\hat{\mu}_{j:m}\leq x\})|.$$  Define the event $E_n=\{\omega:||\F_{n,m}(x,\omega)-F_{m}(x)||_\infty>1/m\}$. This event occurs if, for some $x_0$, $\sum_j|\ind\{\hat{\mu}_{j:m}\leq x_0\} -\ind\{{\mu}_{j:m}\leq x_0\}|\geq2,$ or equivalently, if we have (at least) two consecutive jump points of $F_{m}$, say $\mu_{k:m}$ and $\mu_{k+1:m}$, which are smaller than $\hat{\mu}_{k:m}$, so that $\mu_{k+1:m}<\hat{\mu}_{k:m}$, or vice-versa, if $\hat{\mu}_{k+1:m}<{\mu}_{k:m}$. Indeed, if $\mu_{k+1:m}<\hat{\mu}_{k:m}$ or $\hat{\mu}_{k+1:m}<{\mu}_{k:m}$, we have $F_m(\mu_{k+1:m})-\F_{n,m}(\mu_{k+1:m})\geq 2/m$. The two cases are similar so we can focus on the first one without loss of generality. We can re-write $E_n$ as
		$$E_n=\{\exists k:\hat{\mu}_{k:m}(\omega)> \mu_{k+1:m}\}=\{\exists k:\hat{\mu}_{k:m}(\omega)- \mu_{k+1:m}>0\}$$
		By construction, $\mu_{k+1:m}-\mu_{k:m}=\delta>0.$ For almost every $\omega,$ there exists $n_0$ such that, for every $n>n_0,$ $|\hat{\mu}_{k:m}(\omega)- \mu_{k:m}|<\delta/2$. So, for $n>n_0$ we have
		$$\hat{\mu}_{k:m}(\omega)<\mu_{k:m}+\delta/2=\mu_{k+1:m}-\delta/2 $$
		with probability 1. This implies that, for $n>n_0$, $E_n$ is not verified for almost every $\omega$, therefore
		$ ||\F_{n,m}(x)-F_{m}(x)||_\infty\leq1/m$ with probability 1.
	\hfill\end{proof}
	\subsection{Proof of Lemma~\ref{lemma:p_convergence}}
    \label{proof:p_convergence}
	\begin{proof}
		To prove the first assertion, Lemma~4 implies that $\lim_{m\to\infty}||\F_{n,m}(\omega)-\F_n(\omega)||_p=0$ for every outcome $\omega$ in the sample space. Now, we have
		\begin{align*}
			\lim_{m\to\infty}||\F_{n,m}(\omega)-F||_p\leq\\ \lim_{m\to\infty} \big(||\F_{n,m}(\omega)-\F_n(\omega)||_p+||\F_n(\omega)-F||_p\big)=||\F_n(\omega)-F||_p.
		\end{align*}
		Since the above inequality holds for every outcome $\omega,$ we can write $		\lim_{m\to\infty}||\F_{n,m}-F||_p\leq||\F_n-F||_p.$ The latter random variable tends to 0 with probability 1 because $X$ is integrable.

		We now focus on the second assertion. Let $p\in[1,\infty)$. By Lemma~2
         and since $\E|X|<\infty$, we obtain
		\begin{align*}
			\lim_{m\to\infty}\lim_{n\to\infty}||\F_{n,m}-F||_p\\
			\leq \lim_{m\to\infty}\lim_{n\to\infty} \big(||\F_{n,m}-F_m||_p+||F_m-F||_p\big)=_{a.s.}\lim_{m\to\infty}||F_m-F||_p=0
		\end{align*}

		Finally, if $p=\infty$, by Lemma~5
         and Polya's theorem, we obtain almost surely that
		\begin{align*}
			\lim_{m\to\infty}\lim_{n\to\infty}||\F_{n,m}-F||_\infty\\
			\leq \lim_{m\to\infty}\lim_{n\to\infty} \big(||\F_{n,m}-F_m||_\infty+||F_m-F||_\infty\big)\leq\lim_{m\to\infty}(1/m+||F_m-F||_\infty)=0
		\end{align*}
	\hfill\end{proof}
	\subsection{Proof of Theorem~\ref{theo:double_limit}}
    \label{proof:double_limit}
	\begin{proof} Consider the inequality
		$\|\F_{n,m}-F\|_p \le \|\F_{n,m}-F_m\|_p+\|F_m-F\|_p$.
		For the first term, Theorem~1 implies $\|\F_{n,m}-F_m\|_p^p
		\leq \|\F_n-F\|_1.$ Hence $\|\F_{n,m}-F_m\|_p \le \|\F_n-F\|_1^{1/p}$. For the second term, note that $||F_m-F||_p^p\leq ||F_m-F||_1$. Finally, we have
		$\|\F_{n,m}-F\|_p \leq \|\F_n-F\|_1^{1/p}+\|F_m-F\|_1^{1/p}$. The first term converges a.s. to 0 since $\E|X|<\infty$. The second term converges to 0 by Lemma~2.

	\hfill\end{proof}

\subsection{Proof of Theorem~\ref{theorem:weak}}
\begin{proof}
Under Property~Q, we have $\mathbb Q_n \rightsquigarrow \mathbb Q$ in $L^1(0,1)$. Since $L^1(0,1)$ is a separable Banach space, Skorokhod's representation theorem yields random elements $\mathbb Q_n'$ and $\mathbb Q'$ such that $\mathbb Q_n' \stackrel d= \mathbb Q_n$, $\mathbb Q' \stackrel d= \mathbb Q$, and $\mathbb Q_n' \to_{a.s.} \mathbb Q'$ \ in $L^1(0,1)$. Defining $\mathbb Q_{n,m}' := \mathcal I_m(\mathbb Q_n')$, Theorem~2
 implies that $\mathbb Q_{n,m}' \to_{a.s.} \mathbb Q'$ \ in $L^1(0,1)$ as $n,m\to\infty$. Since $\mathbb Q_{n,m}' \stackrel d= \mathbb Q_{n,m}$ and $\mathbb Q' \stackrel d= \mathbb Q$, it follows that $\mathbb Q_{n,m} \rightsquigarrow \mathbb Q$ in $L^1(0,1)$ as $n,m\to\infty$.

The same argument applies in $L^2(0,1)$ under Property~Q2 and in $L^p(0,1)$ whenever $\mathbb Q_n\rightsquigarrow\mathbb Q$ in $L^p(0,1)$.
\hfill\end{proof}
\subsection{Proof of Theorem~\ref{theo:asymptotic_fixed_m}}
\label{proof:asymptotic_fixed_m}
\begin{proof}
	Recall that
	$$\F_{n,m}(x)-F_m(x)
	=\frac1m\sum_{j=1}^m\Big(\ind\{\hat\mu_{j:m}\le x\}-\ind\{\mu_{j:m}\le x\}\Big).$$ Note that, a.e.,$$\Big|\ind\{\hat\mu_{j:m}\le x\}-\ind\{\mu_{j:m}\le x\}\Big| = \ind\Big\{\min(\mu_{j:m},\hat\mu_{j:m})<x\le \max(\mu_{j:m},\hat\mu_{j:m})\Big\}.$$ Hence, in general,
	\begin{align*}
		\big|\F_{n,m}(x)-F_m(x)\big|\\
		=
		\frac1m\left|\sum_{j=1}^m\big(\ind\{\hat\mu_{j:m}\le x\}-\ind\{\mu_{j:m}\le x\}\big)\right|\\
		\le
		\frac1m\sum_{j=1}^m
		\ind\Big\{\min(\mu_{j:m},\hat\mu_{j:m})<x\le \max(\mu_{j:m},\hat\mu_{j:m})\Big\}.
	\end{align*}
	Let $\delta=\min_j (\mu_{j+1:m}-\mu_{j:m})$, and assume that the following event holds:
	$$E_n=\{\omega:\max_j |\hat{\mu}_{j:m}-{\mu}_{j:m}|<\delta/2\}.$$
	Under $E_n,$ the two functions $\F_{n,m}$ and $F_m$ can differ just by $1/m$, that is $|\F_{n,m}(x)-F_m(x)|\in\{0,1/m\}$ for every $x$. In particular, for a fixed $x$, at most one of the indicators in the sum above is 1. In this case, the inequality above boils down to \begin{align*}
		\big|\F_{n,m}(x)-F_m(x)\big|
		=
		\frac1m\sum_{j=1}^m
		\ind\Big\{\min(\mu_{j:m},\hat\mu_{j:m})<x\le \max(\mu_{j:m},\hat\mu_{j:m})\Big\}.
	\end{align*}
	Consequently, the $L^p$ norm of the difference reduces to a sum of lengths of intervals, that is,
	\begin{align*}
		\|F_{n,m}-F_m\|_p^p
		=
		\int_{\mathbb R}\big|F_{n,m}(x)-F_m(x)\big|^p\,dx\\
		=\frac1{m^p}\sum_{j=1}^m\int \ind\Big\{\min(\mu_{j:m},\hat\mu_{j:m})<x\le \max(\mu_{j:m},\hat\mu_{j:m})\Big\}dx\\
		=\frac{1}{m^p}\sum_{j=1}^m |\hat\mu_{j:m}-\mu_{j:m}|.
	\end{align*}
	Hence, under $E_n$, we know that $\sqrt n\| \F_{n,m}-F_m\|_p^p\to_d \frac{1}{m^p}\sum_{j=1}^m |Z_j|$. It is easy to see that $P(E_n)\to1,$ therefore, the assertion is proved.
\hfill\end{proof}
\subsection{Proof of Proposition~\ref{theo:asymptotic_bounded_support}}
\label{proof:asymptotic_bounded_support}
\begin{proof}Note that, for $1>t>u>0$,
	$$Q(t)-Q(u)=\int_u^t Q'(y)dy>\frac1b(t-u).$$
	Letting $U_{j+1:m}=u$ and $U_{j:m}=t$ and taking expectations, we obtain
	$$\mu_{j+1:m}-\mu_{j:m}>\frac1b\E (U_{j+1:m}- U_{j:m})=\frac1{b(m+1)}.$$
	Letting $\delta_m=\min_j (\mu_{j+1:m}-\mu_{j:m})$ as in the proof of Theorem~9, we obtain $\delta_m>\tfrac1{b(m+1)}$. Consider the following bound
	$$|\hat{\mu}_{j:m}-\mu_{j:m}|\leq \|Q_n-Q\|_\infty\int_0^1 dF_{B_{j:m}}=\|Q_n-Q\|_\infty,$$where the latter term converges a.s. to zero because the support is bounded. Denote $D_n=\|\F_n-F\|_\infty$, so that $Q(p-D_n)\leq Q(p)\leq Q(p+D_n)$ for every $p$ such that $p-D_n\geq0$ and $p+D_n\leq1.$ By assumption, $1/b<Q'<1/a,$ where $Q'=1/f\circ Q.$ This means that $Q$ is Lipschitz continuous, which implies that $$Q(p+D_n)-Q(p-D_n)<\frac{2D_n}a,$$
	therefore, $\|Q_n-Q\|_\infty<\frac{2}{a}\|\F_n-F\|_\infty$, where the latter quantity is $O_p(n^{-1/2}).$ We conclude that $|\hat{\mu}_{j:m}-\mu_{j:m}|=O_p(n^{-1/2}).$ Consider the event $E_n=\{\omega:\max_j |\hat{\mu}_{j:m}-{\mu}_{j:m}|<\delta_m/2\}$ as in the proof of Theorem~9. If $m=o(\sqrt n),$ $P(E_n)\to1, $ therefore we have the same result as in Theorem~9. Moreover, in this case $m$ diverges to infinity, so that we can apply similar steps as in Theorem~9 and derive the limit.
\hfill\end{proof}

\subsection{Proof of Lemma~\ref{sensitivity}}
\label{proof:sensitivity}
\begin{proof}
Let \(\F_n\) and \(\F'_n\) be the empirical CDFs associated with
\(\mathcal X\) and \(\mathcal X'\). Since $\mathcal X\sim\mathcal X'$, Lemma~\ref{lemma:Im_contraction} gives
\[
W_1(\mathcal H_m(\F_n),\mathcal H_m(\F'_n)) \le W_1(\F_n,\F'_n) \le \frac{B}{n}.
\]
Since \(\mathcal H_m(\F_n)\) and \(\mathcal H_m(\F'_n)\) put mass \(1/m\) at the nondecreasing vectors
\(\widehat{\boldsymbol\mu}_m(\mathcal X)\) and
\(\widehat{\boldsymbol\mu}_m(\mathcal X')\), respectively,
\[
\frac1m \left\| \widehat{\boldsymbol\mu}_m(\mathcal X) - \widehat{\boldsymbol\mu}_m(\mathcal X') \right\|_1 = W_1(\mathcal H_m(\F_n),\mathcal H_m(\F'_n)).
\]
It follows that
\[
\Delta_1(\widehat{\boldsymbol\mu}_m)\le \frac{mB}{n}.
\]
To prove the reverse inequality, take
\[
\mathcal X=(a,\ldots,a), \qquad \mathcal X'=(a,\ldots,a,b).
\]
Then \(\mathcal X\sim\mathcal X'\) and
\(\widehat\mu_{j:m}(\mathcal X)=a\) for every \(j\). Moreover, the empirical quantile
function \((\F'_n)^{-1}\) is equal to \(a\) on \((0,1-1/n]\) and to \(b\) on
\((1-1/n,1]\). Hence
\[
\widehat\mu_{j:m}(\mathcal X') = a+B\,P(U_{j:m}>1-1/n).
\]
Therefore,
\[
\begin{aligned}
\left\|
\widehat{\boldsymbol\mu}_m(\mathcal X')
-
\widehat{\boldsymbol\mu}_m(\mathcal X)
\right\|_1
&=
B\sum_{j=1}^m  P(U_{j:m}>1-1/n) \\
&=
B\,\mathbb E\left[\sum_{j=1}^m \ind\{U_{j:m}>1-1/n\}\right] \\
&=
B\,\mathbb E\left[\sum_{j=1}^m \ind\{U_i>1-1/n\}\right] \\
&=
Bm\, P(U_1>1-1/n)
=
\frac{Bm}{n}.
\end{aligned}
\]
\end{proof}

\subsection{Proof of Theorem~\ref{Theo: privacy}}
\label{proof:privacy}
\begin{proof}
By the triangle inequality,
\[
W_1(\widetilde \F_{n,m}^{\,\varepsilon},F) \le W_1(\widetilde \F_{n,m}^{\,\varepsilon},\F_{n,m}) + W_1(\F_{n,m},F_m) + W_1(F_m,F).
\]
By Lemma~\ref{lemma:Im_contraction},
\(
W_1(\F_{n,m},F_m) \le W_1(\F_n,F), \) therefore
\[
W_1(\widetilde \F_{n,m}^{\,\varepsilon},F) \le W_1(\widetilde \F_{n,m}^{\,\varepsilon},\F_{n,m}) + W_1(\F_n,F) + W_1(F_m,F).
\]

Let
\(
\overline\mu_1\le \cdots \le \overline\mu_m \) be the increasing rearrangement of
\(
\widetilde\mu_1,\ldots,\widetilde\mu_m. \) Since both \(\widetilde \F_{n,m}^{\,\varepsilon}\) and \(\F_{n,m}\) assign mass
\(1/m\) to their jump points, the one-dimensional representation of \(W_1\) gives
\[
W_1(\widetilde \F_{n,m}^{\,\varepsilon},\F_{n,m}) = \frac1m \sum_{j=1}^m |\overline\mu_j-\widehat\mu_{j:m}|.
\]
On the real line, the increasing matching minimises the Wasserstein transport cost. Hence
\[
\sum_{j=1}^m |\overline\mu_j-\widehat\mu_{j:m}| \le \sum_{j=1}^m |\widetilde\mu_j-\widehat\mu_{j:m}|.
\]
Moreover, since \(\widehat\mu_{j:m}\in[a,b]\), projection onto \([a,b]\) cannot increase the distance from \(\widehat\mu_{j:m}\). Thus
\[
|\widetilde\mu_j-\widehat\mu_{j:m}| = |\Pi_{[a,b]}(\widehat\mu_{j:m}+Z_j)-\widehat\mu_{j:m}| \le |Z_j|.
\]
Consequently,
\[
W_1(\widetilde \F_{n,m}^{\,\varepsilon},\F_{n,m}) \le \frac1m \sum_{j=1}^m |Z_j|.
\]
Combining the previous inequalities yields the first bound. Taking expectations and using
\(
\mathbb E |Z_j| = \frac{mB}{n\varepsilon}, \) we obtain the first inequality of the second bound.

It remains to prove the last two inequalities bounds. By the argument used in the proof of Lemma~\ref{lemma:HCDF_CDF_convergence_rate},
\[
W_1(F_m,F) \le \frac1{\sqrt m} \int_0^1 \sqrt{u(1-u)}\,dQ(u).
\]
By Jensen's inequality,
\[
\mathbb E|\F_n(x)-F(x)| \le \sqrt{\mathbb E(\F_n(x)-F(x))^2} = \sqrt{\frac{F(x)(1-F(x))}{n}}.
\]
Therefore
\[
\mathbb E W_1(\F_n,F) \le \frac1{\sqrt n} \int_a^b \sqrt{F(x)(1-F(x))}\,dx.
\]
Combining these inequalities gives the second inequality. Finally, since \(F(x)(1-F(x))\le 1/4\) for every \(x\), we obtain the last inequality.
\end{proof}

\newpage
\section{Estimation error: a simulation study}\label{appendix:simulations}
We compare the performance of the EHCDF with two classic alternatives: the ECDF and the kernel estimator. For the construction of the EHCDF, we set $m = n^\gamma$, with $\gamma \in[1,1.5]$. Recall that the EHCDF is an averaged version of the ECDF where the parameter $m $ controls this effect. For larger values of $\gamma$ the two estimators get closer.

Kernel smoothing provides a natural benchmark in our context. However, there is no single kernel-based specification that is uniformly optimal across distributional settings, since performance depends on the choice of bandwidth and on how support restrictions and boundary effects are handled. The literature addresses these issues through bandwidth selection, transformations, and boundary-based kernel constructions \citep{SheatherJones1991,GeenensWang2018,Nagler2018a,Nagler2018b}. For this reason, we use the \texttt{R} package \texttt{kde1d} \citep{nagler2025package}, which computes kernel estimators for unconstrained as well as support-constrained data. In the unconstrained case (when the support is $\R$), this yields a Gaussian-kernel estimator with the plug-in bandwidth selector of \citet{SheatherJones1991}, which we denote by EKCDF. For constrained supports, the package automatically applies a log-transformation in the presence of a single boundary \citep{GeenensWang2018} or a probit transformation in the presence of two boundaries \citep{Geenens2014}, and we refer to the resulting modified estimator as EKCDFM.

 We consider the following distributions: uniform $\mathcal U[a,b]$, normal $\mathcal N(\mu,\sigma^2)$, log-normal $\mathrm{Log}\mathcal N(\mu,\sigma^2)$, Weibull $\mathcal W(a,\sigma)$, gamma $\Gamma(a,\sigma)$, Gompertz $\mathcal G(a,\sigma)$, beta $\mathcal B(a,b)$, Student's $t_\nu$, and binomial $\mathrm{Bin}(k,\pi)$. We also consider equally weighted mixtures of two distributions $D_1$ and $D_2$, denoted by $\mathcal M(D_1,D_2)$. The parametrisations are standard, except for the Gompertz distribution, whose CDF is $
1-\exp\{a(1-e^{\sigma x})\}.$

We generate Monte Carlo samples of size $n\in\{25,50,100\}$ using 1000 simulation runs. To assess global accuracy of a general estimator $\hat F$, we consider the $L^p$-errors $\|\hat F-F\|_p$ for $p=1,2,\infty$, and report their simulated average values relative to the ECDF, namely
$
L^p_{\%}(\hat F)
=
\frac{\mathbb E\|\hat F-F\|_p}{\mathbb E\|\mathbb F_n-F\|_p}\times 100\%.$
The corresponding results are reported in Table~\ref{tab:sim_results_real},~\ref{tab:sim_results_positive},~\ref{tab:sim_results_bounded}, and~\ref{tab:sim_results_mixture}, for distributions supported on $\mathbb R$, on $[0,\infty)$, on compact intervals, and for mixtures, respectively. Since such global criteria do not reveal where the estimation error is located, we also consider the standardised pointwise mean squared error (MSE ratio) at the quantiles $F^{-1}(p)$,
\[
\frac{\mathbb E\bigl[(\hat F(F^{-1}(p))-p)^2\bigr]}
     {\mathbb E\bigl[(\mathbb F_n(F^{-1}(p))-p)^2\bigr]}.
\]
The complete set of MSE-ratio plots is reported at the end of this section.

\subsection{\texorpdfstring{$L^p$}{Lp}-norm results}
 We start with distributions supported on $\mathbb R$ (Table~\ref{tab:sim_results_real}). For all the distributions where the classic EKCDF performs well, namely those for which regular smoothing improves estimation, the EHCDF also improves on the ECDF by 5--15\%, especially for lower values of $\gamma$. The EKCDF outperforms the EHCDF when the kernel specification is suited to the underlying model, as in the Gaussian case. The improvement in terms of sup-norm is more pronounced. This is a consequence of the continuity of the EKCDF, while the EHCDF is still a step function. When the tails become heavier, however, the ranking may reverse (see for instance the $t_2$ distribution for $n\geq 50$).

The results for distributions supported on $[0,\infty)$ are illustrated in Table~\ref{tab:sim_results_positive}
. Also in this case, the EHCDF typically improves on the ECDF. The performances of EKCDF and EKCDFM can be quite different. The EHCDF often lies between the two kernel variants in terms of error, or even outperform both of them, at least with respect to some metric (see for example the $\Gamma(0.5,1)$ case).
The EKCDFM does not necessarily improve on the EKCDF, and its performance depends on how suitable the support transformation is for the underlying distribution. As discussed by \citet{wand1991transformations}, transformations are helpful only insofar as they make a global-bandwidth kernel estimator more appropriate for the transformed density. In particular, if the transformation does not sufficiently regularise the density, especially near the boundary, the transformed estimator may perform worse. This is illustrated, for example, by the Weibull family. When the shape parameter is smaller than $1$, the density diverges at zero, so the EKCDFM tends to outperform the EKCDF. When the shape parameter exceeds $1$, the density is less irregular and becomes closer to the class of densities for which a global-bandwidth estimator is appropriate. In such cases, the classic EKCDF provides better results. Overall, the results suggest that the EHCDF may be a reliable alternative to kernel methods, as it is less sensitive to such support issues, improving over the ECDF across all metrics and all values of $m$ considered.

We now consider distributions supported on compact intervals, see Table~\ref{tab:sim_results_bounded}
. The results confirm that the EHCDF improves on the ECDF, while EKCDF and EKCDFM can behave differently (for example in the $\mathcal B(0.25,1)$ or $\mathcal B(4,1)$ cases). The EHCDF often outperforms one of the two alternative methods.

To discover the limitations of the EKCDF and the EHCDF, in Table~\ref{tab:sim_results_mixture}
we analyse ad-hoc mixtures, such as mixtures of two well-separated normal and uniform distributions. In such cases, both the EHCDF and kernel methods may assign mass to regions where the true distribution has little or no mass, which typically leads to a larger estimation error compared to the ECDF. In such cases, the EHCDF can outperform the EKCDF, especially for larger values of $m$ (see for instance the $\mathcal M(\mathcal N(-5,1),\mathcal N(5,1))$ and $\mathcal M(\mathcal N(-5,1),\mathcal W(0.5,1))$ cases).

{\footnotesize
\centering
\setlength{\tabcolsep}{2.5pt}
\renewcommand{\arraystretch}{1.05}
\begin{singlespace}
\begin{longtable}{|l|l|c|*{9}{c|}}
\caption{\label{tab:sim_results_real}
Simulation results comparing the relative performance of EKCDF and EHCDF with $m=n^\gamma$, for distributions supported on $\mathbb{R}$.}\\
\hline
\multirow{2}{*}{Distribution} & \multirow{2}{*}{Estimator} & \multirow{2}{*}{$\gamma$}
& \multicolumn{3}{c|}{$n=25$}
& \multicolumn{3}{c|}{$n=50$}
& \multicolumn{3}{c|}{$n=100$} \\
\hhline{|~|~|~|-|-|-|-|-|-|-|-|-|}
& & & $L^1_{\%}$ & $L^2_{\%}$ & $L^\infty_{\%}$ & $L^1_{\%}$ & $L^2_{\%}$ & $L^\infty_{\%}$ & $L^1_{\%}$ & $L^2_{\%}$ & $L^\infty_{\%}$ \\
\hline
\endfirsthead
\hline
\multirow{2}{*}{Distribution} & \multirow{2}{*}{Estimator} & \multirow{2}{*}{$\gamma$}
& \multicolumn{3}{c|}{$n=25$}
& \multicolumn{3}{c|}{$n=50$}
& \multicolumn{3}{c|}{$n=100$} \\
\hhline{|~|~|~|-|-|-|-|-|-|-|-|-|}
& & & $L^1_{\%}$ & $L^2_{\%}$ & $L^\infty_{\%}$ & $L^1_{\%}$ & $L^2_{\%}$ & $L^\infty_{\%}$ & $L^1_{\%}$ & $L^2_{\%}$ & $L^\infty_{\%}$ \\
\hline
\endhead
\hline
\endfoot
\hline
\endlastfoot
\multirow{6}{*}{$\mathcal{N}(0,1)$}
& EKCDF & --   & 83\% & 80\% & 57\% & 84\% & 81\% & 57\% & 83\% & 81\% & 56\% \\
\hhline{|~|-|-|-|-|-|-|-|-|-|-|-|}
& \multirow{5}{*}{EHCDF} & 1    & 93\% & 90\% & 79\% & 93\% & 91\% & 78\% & 94\% & 93\% & 80\% \\
\hhline{|~|~|-|-|-|-|-|-|-|-|-|-|}
&  & 1.05 & 92\% & 90\% & 78\% & 93\% & 92\% & 79\% & 94\% & 93\% & 80\% \\
\hhline{|~|~|-|-|-|-|-|-|-|-|-|-|}
&  & 1.1  & 93\% & 91\% & 78\% & 93\% & 92\% & 78\% & 94\% & 93\% & 81\% \\
\hhline{|~|~|-|-|-|-|-|-|-|-|-|-|}
&  & 1.2  & 93\% & 91\% & 77\% & 94\% & 93\% & 79\% & 95\% & 94\% & 82\% \\
\hhline{|~|~|-|-|-|-|-|-|-|-|-|-|}
&  & 1.5  & 94\% & 93\% & 79\% & 96\% & 95\% & 83\% & 97\% & 97\% & 87\% \\
\hline
\multirow{6}{*}{$t_2$}
& EKCDF & --   & 98\% & 95\% & 71\% & 104\% & 102\% & 77\% & 104\% & 106\% & 86\% \\
\hhline{|~|-|-|-|-|-|-|-|-|-|-|-|}
& \multirow{5}{*}{EHCDF} & 1    & 100\% & 98\% & 82\% & 100\% & 97\% & 82\% & 101\% & 97\% & 85\% \\
\hline
\end{longtable}
\end{singlespace}
}

\newpage
{\footnotesize
\centering
\setlength{\tabcolsep}{2pt}
\begin{singlespace}
\begin{longtable}{|l|l|c|c c c|c c c|c c c|}
\caption{\label{tab:sim_results_positive}\footnotesize Simulation results comparing the relative performance of EKCDF, EKCDFM, and EHCDF with $m=n^\gamma$, for distributions supported on $[0,\infty)$.}\\
\hline
\multirow{2}{*}{Distribution} & \multirow{2}{*}{Estimator} & \multirow{2}{*}{$\gamma$} & \multicolumn{3}{c|}{$n=25$} & \multicolumn{3}{c|}{$n=50$} & \multicolumn{3}{c|}{$n=100$} \\*\nopagebreak[4] \hhline{|~|~|~|---|---|---|}
 & & & $L^1_{\%}$ & $L^2_{\%}$ & $L^\infty_{\%}$ & $L^1_{\%}$ & $L^2_{\%}$ & $L^\infty_{\%}$ & $L^1_{\%}$ & $L^2_{\%}$ & $L^\infty_{\%}$ \\ \hline
\endfirsthead
\hline
\multirow{2}{*}{Distribution} & \multirow{2}{*}{Estimator} & \multirow{2}{*}{$\gamma$} & \multicolumn{3}{c|}{$n=25$} & \multicolumn{3}{c|}{$n=50$} & \multicolumn{3}{c|}{$n=100$} \\*\nopagebreak[4] \hhline{|~|~|~|---|---|---|}
 & & & $L^1_{\%}$ & $L^2_{\%}$ & $L^\infty_{\%}$ & $L^1_{\%}$ & $L^2_{\%}$ & $L^\infty_{\%}$ & $L^1_{\%}$ & $L^2_{\%}$ & $L^\infty_{\%}$ \\ \hline
\endhead
\hline
\endfoot
\hline
\endlastfoot
\multirow{7}{*}{$\Gamma(0.25,1)$} & EKCDFM & - & \multicolumn{1}{c|}{104\%} & \multicolumn{1}{c|}{107\%} & 84\% & \multicolumn{1}{c|}{106\%} & \multicolumn{1}{c|}{108\%} & 82\% & \multicolumn{1}{c|}{108\%} & \multicolumn{1}{c|}{109\%} & 84\% \\*\nopagebreak[4] \hhline{|~|-|-|---|---|---|}
 & EKCDF & - & \multicolumn{1}{c|}{118\%} & \multicolumn{1}{c|}{133\%} & 148\% & \multicolumn{1}{c|}{129\%} & \multicolumn{1}{c|}{160\%} & 201\% & \multicolumn{1}{c|}{148\%} & \multicolumn{1}{c|}{206\%} & 279\% \\*\nopagebreak[4] \hhline{|~|-|-|---|---|---|}
 & \multirow{5}{*}{EHCDF} & 1 & \multicolumn{1}{c|}{97\%} & \multicolumn{1}{c|}{96\%} & 83\% & \multicolumn{1}{c|}{96\%} & \multicolumn{1}{c|}{94\%} & 82\% & \multicolumn{1}{c|}{95\%} & \multicolumn{1}{c|}{94\%} & 83\% \\*\nopagebreak[4] \hhline{|~|~|-|---|---|---|}
 &  & 1.05 & \multicolumn{1}{c|}{95\%} & \multicolumn{1}{c|}{94\%} & 81\% & \multicolumn{1}{c|}{95\%} & \multicolumn{1}{c|}{93\%} & 81\% & \multicolumn{1}{c|}{95\%} & \multicolumn{1}{c|}{94\%} & 83\% \\*\nopagebreak[4] \hhline{|~|~|-|---|---|---|}
 &  & 1.1 & \multicolumn{1}{c|}{96\%} & \multicolumn{1}{c|}{94\%} & 82\% & \multicolumn{1}{c|}{95\%} & \multicolumn{1}{c|}{93\%} & 82\% & \multicolumn{1}{c|}{95\%} & \multicolumn{1}{c|}{94\%} & 83\% \\*\nopagebreak[4] \hhline{|~|~|-|---|---|---|}
 &  & 1.2 & \multicolumn{1}{c|}{94\%} & \multicolumn{1}{c|}{92\%} & 79\% & \multicolumn{1}{c|}{94\%} & \multicolumn{1}{c|}{93\%} & 80\% & \multicolumn{1}{c|}{95\%} & \multicolumn{1}{c|}{94\%} & 85\% \\*\nopagebreak[4] \hhline{|~|~|-|---|---|---|}
 &  & 1.5 & \multicolumn{1}{c|}{94\%} & \multicolumn{1}{c|}{93\%} & 81\% & \multicolumn{1}{c|}{96\%} & \multicolumn{1}{c|}{95\%} & 86\% & \multicolumn{1}{c|}{97\%} & \multicolumn{1}{c|}{96\%} & 90\% \\
\hline
\end{longtable}
\end{singlespace}

}
\newpage
{\footnotesize
\centering
\setlength{\tabcolsep}{2pt}
\begin{singlespace}
\begin{longtable}{|l|l|c|c c c|c c c|c c c|}
\caption{\label{tab:sim_results_bounded}\footnotesize Simulation results comparing the relative performance of EKCDF, EKCDFM, and EHCDF with $m=n^\gamma$, for distributions supported on compact intervals.}\\
\hline
\multirow{2}{*}{Distribution} & \multirow{2}{*}{Estimator} & \multirow{2}{*}{$\gamma$} & \multicolumn{3}{c|}{$n=25$} & \multicolumn{3}{c|}{$n=50$} & \multicolumn{3}{c|}{$n=100$} \\*\nopagebreak[4] \hhline{|~|~|~|---|---|---|}
 & & & $L^1_{\%}$ & $L^2_{\%}$ & $L^\infty_{\%}$ & $L^1_{\%}$ & $L^2_{\%}$ & $L^\infty_{\%}$ & $L^1_{\%}$ & $L^2_{\%}$ & $L^\infty_{\%}$ \\ \hline
\endfirsthead
\hline
\multirow{2}{*}{Distribution} & \multirow{2}{*}{Estimator} & \multirow{2}{*}{$\gamma$} & \multicolumn{3}{c|}{$n=25$} & \multicolumn{3}{c|}{$n=50$} & \multicolumn{3}{c|}{$n=100$} \\*\nopagebreak[4] \hhline{|~|~|~|---|---|---|}
 & & & $L^1_{\%}$ & $L^2_{\%}$ & $L^\infty_{\%}$ & $L^1_{\%}$ & $L^2_{\%}$ & $L^\infty_{\%}$ & $L^1_{\%}$ & $L^2_{\%}$ & $L^\infty_{\%}$ \\ \hline
\endhead
\hline
\endfoot
\hline
\endlastfoot
\multirow{7}{*}{$\mathcal{B}(0.25,1)$} & EKCDFM & - & \multicolumn{1}{c|}{90\%} & \multicolumn{1}{c|}{86\%} & 65\% & \multicolumn{1}{c|}{89\%} & \multicolumn{1}{c|}{85\%} & 68\% & \multicolumn{1}{c|}{91\%} & \multicolumn{1}{c|}{87\%} & 78\% \\*\nopagebreak[4] \hhline{|~|-|-|---|---|---|}
 & EKCDF & - & \multicolumn{1}{c|}{110\%} & \multicolumn{1}{c|}{117\%} & 121\% & \multicolumn{1}{c|}{117\%} & \multicolumn{1}{c|}{132\%} & 155\% & \multicolumn{1}{c|}{126\%} & \multicolumn{1}{c|}{152\%} & 204\% \\*\nopagebreak[4] \hhline{|~|-|-|---|---|---|}
 & \multirow{5}{*}{EHCDF} & 1 & \multicolumn{1}{c|}{96\%} & \multicolumn{1}{c|}{93\%} & 82\% & \multicolumn{1}{c|}{95\%} & \multicolumn{1}{c|}{92\%} & 80\% & \multicolumn{1}{c|}{95\%} & \multicolumn{1}{c|}{93\%} & 81\% \\*\nopagebreak[4] \hhline{|~|~|-|---|---|---|}
 &  & 1.05 & \multicolumn{1}{c|}{94\%} & \multicolumn{1}{c|}{92\%} & 80\% & \multicolumn{1}{c|}{95\%} & \multicolumn{1}{c|}{92\%} & 80\% & \multicolumn{1}{c|}{95\%} & \multicolumn{1}{c|}{93\%} & 81\% \\*\nopagebreak[4] \hhline{|~|~|-|---|---|---|}
 &  & 1.1 & \multicolumn{1}{c|}{94\%} & \multicolumn{1}{c|}{91\%} & 79\% & \multicolumn{1}{c|}{94\%} & \multicolumn{1}{c|}{92\%} & 79\% & \multicolumn{1}{c|}{95\%} & \multicolumn{1}{c|}{93\%} & 81\% \\*\nopagebreak[4] \hhline{|~|~|-|---|---|---|}
 &  & 1.2 & \multicolumn{1}{c|}{93\%} & \multicolumn{1}{c|}{91\%} & 78\% & \multicolumn{1}{c|}{94\%} & \multicolumn{1}{c|}{92\%} & 79\% & \multicolumn{1}{c|}{95\%} & \multicolumn{1}{c|}{94\%} & 82\% \\*\nopagebreak[4] \hhline{|~|~|-|---|---|---|}
 &  & 1.5 & \multicolumn{1}{c|}{94\%} & \multicolumn{1}{c|}{92\%} & 79\% & \multicolumn{1}{c|}{96\%} & \multicolumn{1}{c|}{95\%} & 83\% & \multicolumn{1}{c|}{97\%} & \multicolumn{1}{c|}{97\%} & 87\% \\
\hline
\end{longtable}
\end{singlespace}

}
\newpage

{\footnotesize
\centering
\setlength{\tabcolsep}{2pt}
\begin{singlespace}
\begin{longtable}{|l|l|c|c c c|c c c|c c c|}
\caption{\label{tab:sim_results_mixture}\footnotesize Simulation results comparing the relative performance of EKCDF, EKCDFM, and EHCDF with $m=n^\gamma$, for mixture distributions.}\\
\hline
\multirow{2}{*}{Distribution} & \multirow{2}{*}{Estimator} & \multirow{2}{*}{$\gamma$} & \multicolumn{3}{c|}{$n=25$} & \multicolumn{3}{c|}{$n=50$} & \multicolumn{3}{c|}{$n=100$} \\*\nopagebreak[4] \hhline{|~|~|~|---|---|---|}
 & & & $L^1_{\%}$ & $L^2_{\%}$ & $L^\infty_{\%}$ & $L^1_{\%}$ & $L^2_{\%}$ & $L^\infty_{\%}$ & $L^1_{\%}$ & $L^2_{\%}$ & $L^\infty_{\%}$ \\ \hline
\endfirsthead
\hline
\multirow{2}{*}{Distribution} & \multirow{2}{*}{Estimator} & \multirow{2}{*}{$\gamma$} & \multicolumn{3}{c|}{$n=25$} & \multicolumn{3}{c|}{$n=50$} & \multicolumn{3}{c|}{$n=100$} \\*\nopagebreak[4] \hhline{|~|~|~|---|---|---|}
 & & & $L^1_{\%}$ & $L^2_{\%}$ & $L^\infty_{\%}$ & $L^1_{\%}$ & $L^2_{\%}$ & $L^\infty_{\%}$ & $L^1_{\%}$ & $L^2_{\%}$ & $L^\infty_{\%}$ \\ \hline
\endhead
\hline
\endfoot
\hline
\endlastfoot
\multirow{7}{*}{$\mathcal{M}(\mathcal{U}(0,1),\mathcal{U}(3,4))$} & EKCDFM & - & \multicolumn{1}{c|}{122\%} & \multicolumn{1}{c|}{131\%} & 128\% & \multicolumn{1}{c|}{129\%} & \multicolumn{1}{c|}{139\%} & 142\% & \multicolumn{1}{c|}{130\%} & \multicolumn{1}{c|}{144\%} & 161\% \\*\nopagebreak[4] \hhline{|~|-|-|---|---|---|}
 & EKCDF & - & \multicolumn{1}{c|}{122\%} & \multicolumn{1}{c|}{132\%} & 129\% & \multicolumn{1}{c|}{133\%} & \multicolumn{1}{c|}{142\%} & 145\% & \multicolumn{1}{c|}{149\%} & \multicolumn{1}{c|}{160\%} & 173\% \\*\nopagebreak[4] \hhline{|~|-|-|---|---|---|}
 & \multirow{5}{*}{EHCDF} & 1 & \multicolumn{1}{c|}{112\%} & \multicolumn{1}{c|}{119\%} & 121\% & \multicolumn{1}{c|}{114\%} & \multicolumn{1}{c|}{121\%} & 125\% & \multicolumn{1}{c|}{115\%} & \multicolumn{1}{c|}{124\%} & 135\% \\*\nopagebreak[4] \hhline{|~|~|-|---|---|---|}
 &  & 1.05 & \multicolumn{1}{c|}{110\%} & \multicolumn{1}{c|}{116\%} & 115\% & \multicolumn{1}{c|}{111\%} & \multicolumn{1}{c|}{119\%} & 122\% & \multicolumn{1}{c|}{114\%} & \multicolumn{1}{c|}{121\%} & 126\% \\*\nopagebreak[4] \hhline{|~|~|-|---|---|---|}
 &  & 1.1 & \multicolumn{1}{c|}{107\%} & \multicolumn{1}{c|}{113\%} & 112\% & \multicolumn{1}{c|}{110\%} & \multicolumn{1}{c|}{115\%} & 115\% & \multicolumn{1}{c|}{110\%} & \multicolumn{1}{c|}{116\%} & 120\% \\*\nopagebreak[4] \hhline{|~|~|-|---|---|---|}
 &  & 1.2 & \multicolumn{1}{c|}{105\%} & \multicolumn{1}{c|}{109\%} & 105\% & \multicolumn{1}{c|}{106\%} & \multicolumn{1}{c|}{110\%} & 108\% & \multicolumn{1}{c|}{106\%} & \multicolumn{1}{c|}{110\%} & 111\% \\*\nopagebreak[4] \hhline{|~|~|-|---|---|---|}
 &  & 1.5 & \multicolumn{1}{c|}{100\%} & \multicolumn{1}{c|}{102\%} & 93\% & \multicolumn{1}{c|}{102\%} & \multicolumn{1}{c|}{103\%} & 93\% & \multicolumn{1}{c|}{101\%} & \multicolumn{1}{c|}{102\%} & 95\% \\
\hline
\end{longtable}
\end{singlespace}
}

 \subsection{MSE results}
 The pointwise MSE plots complement the global $L^p$ criteria by showing where the estimation error is concentrated.
 The MSE-ratio plots reported below show that the improvement of the EHCDF over the ECDF is mainly located in the body of the distribution, while estimation is generally less precise in the tails, especially when these are heavier or when probability mass concentrates near the boundaries. This is quite natural: the EHCDF replaces extreme observations with less extreme averages, making the estimator less sensitive to heavy tails and producing error when many observations lie close to the boundary. For example, with regard to tail heaviness, tail estimation for the $t_\nu$ and $\mathcal W(a,1)$ distributions improves as $\nu$ and $a$ increase, respectively. Differently, when the density diverges at the boundary, estimation deteriorates in that region; see, for instance, the left-tail estimation for $\mathcal B(a,1)$ or $\mathcal W(a,1)$ when $a$ approaches 0.
With regard to kernel estimators, we note that they can yield a uniformly smaller error in some cases (for example, for the normal distribution), but also a substantially larger error in some regions in other cases. The difference between EKCDF and EKCDFM is particularly noticeable in the Weibull case, where one method performs well and the other performs poorly depending on the shape parameter (see also the gamma case). Since it is not known \textit{a priori} which version should be chosen, this may be viewed as a limitation in practice. By contrast, the performance of the EHCDF is more stable and appears to be less sensitive to the underlying distribution.

\begin{figure}[p]\label{plot1}
\centering
\includegraphics[width=0.95\textwidth,height=0.28\textheight,keepaspectratio]{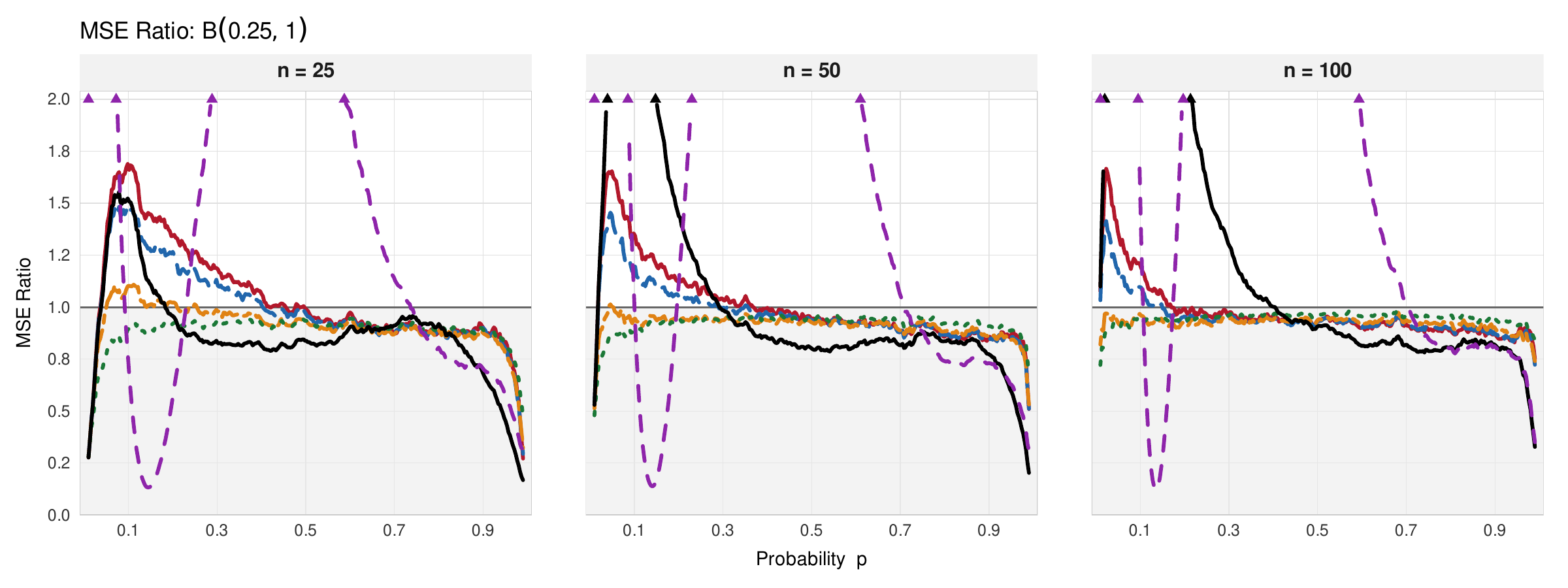}\par\medskip
\includegraphics[width=0.95\textwidth,height=0.28\textheight,keepaspectratio]{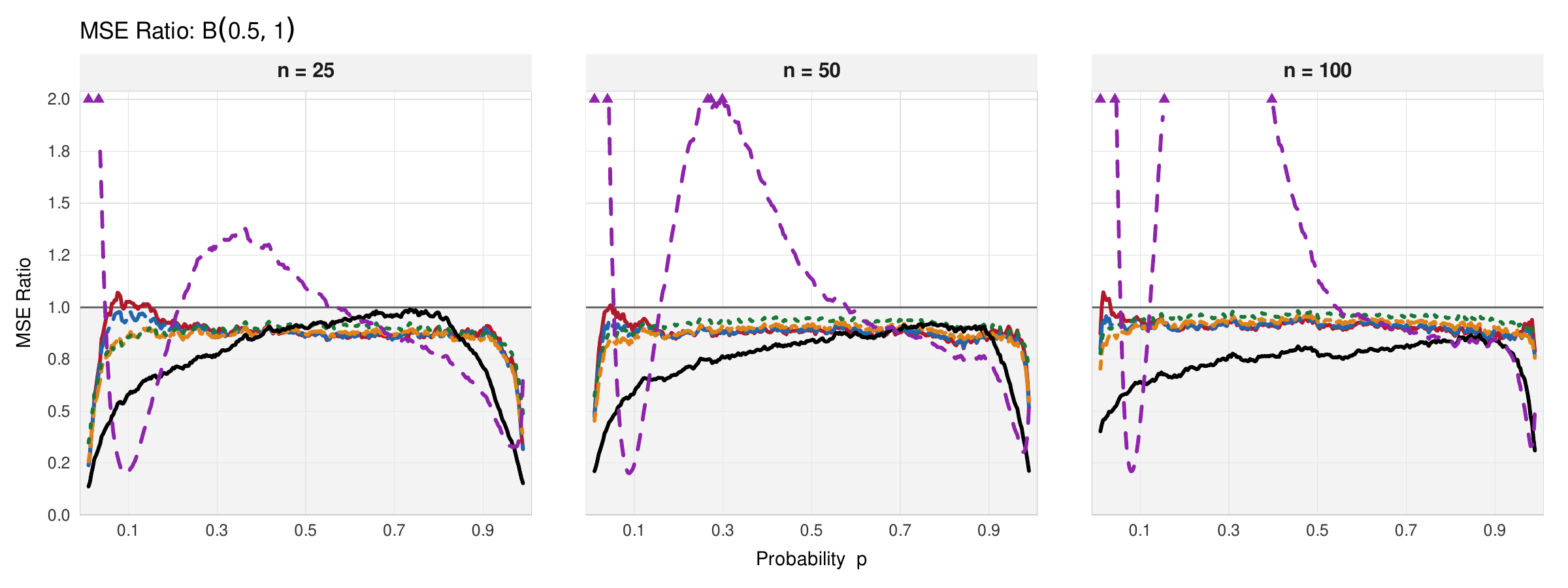}\par\medskip
\includegraphics[width=0.95\textwidth,height=0.28\textheight,keepaspectratio]{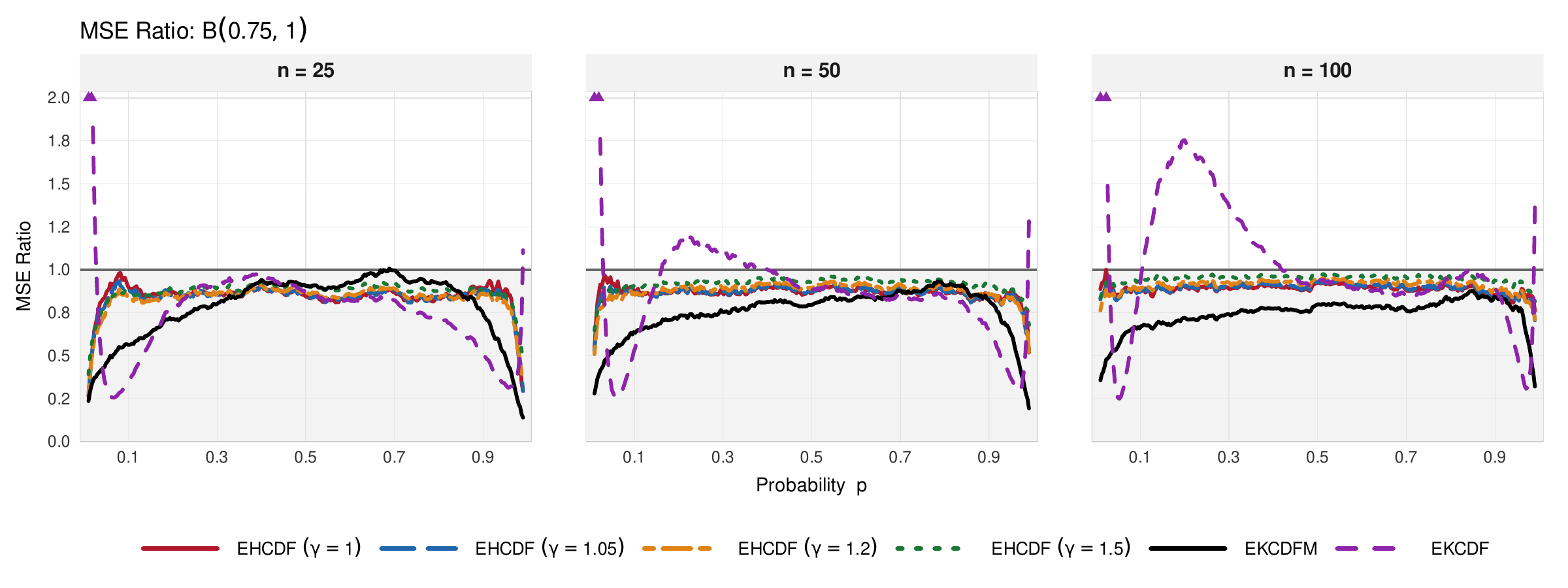}\par\medskip
\caption{
MSE ratio for $\mathcal B(0.25,1)$, $\mathcal B(0.5,1)$ and $\mathcal B(0.75,1)$ (top to bottom). Different lines correspond to different CDF estimators and the dotted horizontal line corresponds to 1.
}
\end{figure}

\begin{figure}[p]
\centering
\includegraphics[width=0.95\textwidth,height=0.28\textheight,keepaspectratio]{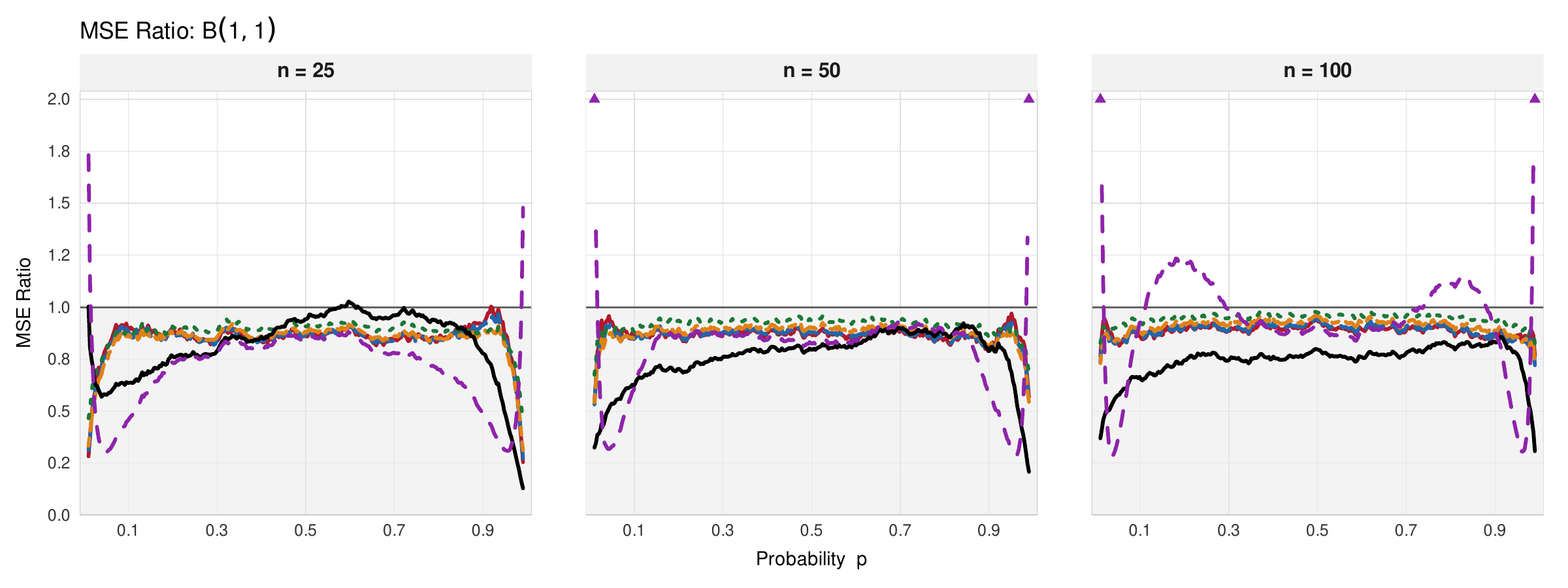}\par\medskip
\includegraphics[width=0.95\textwidth,height=0.28\textheight,keepaspectratio]{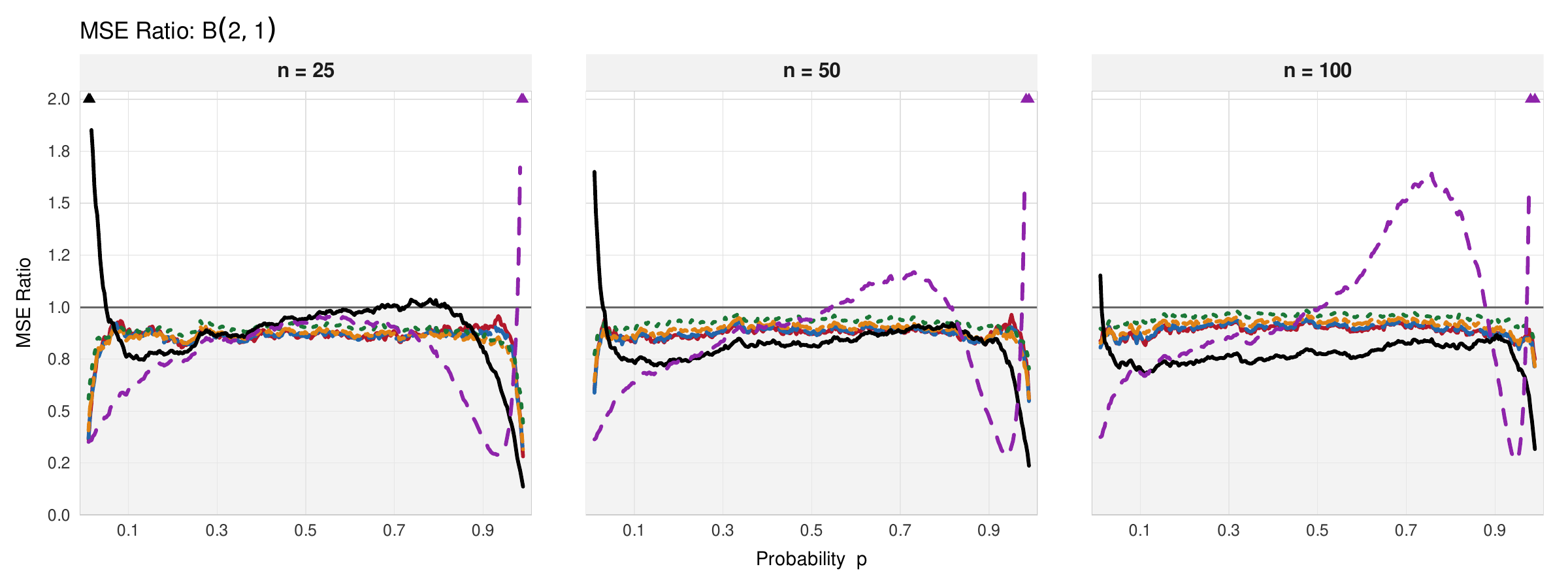}\par\medskip
\includegraphics[width=0.95\textwidth,height=0.28\textheight,keepaspectratio]{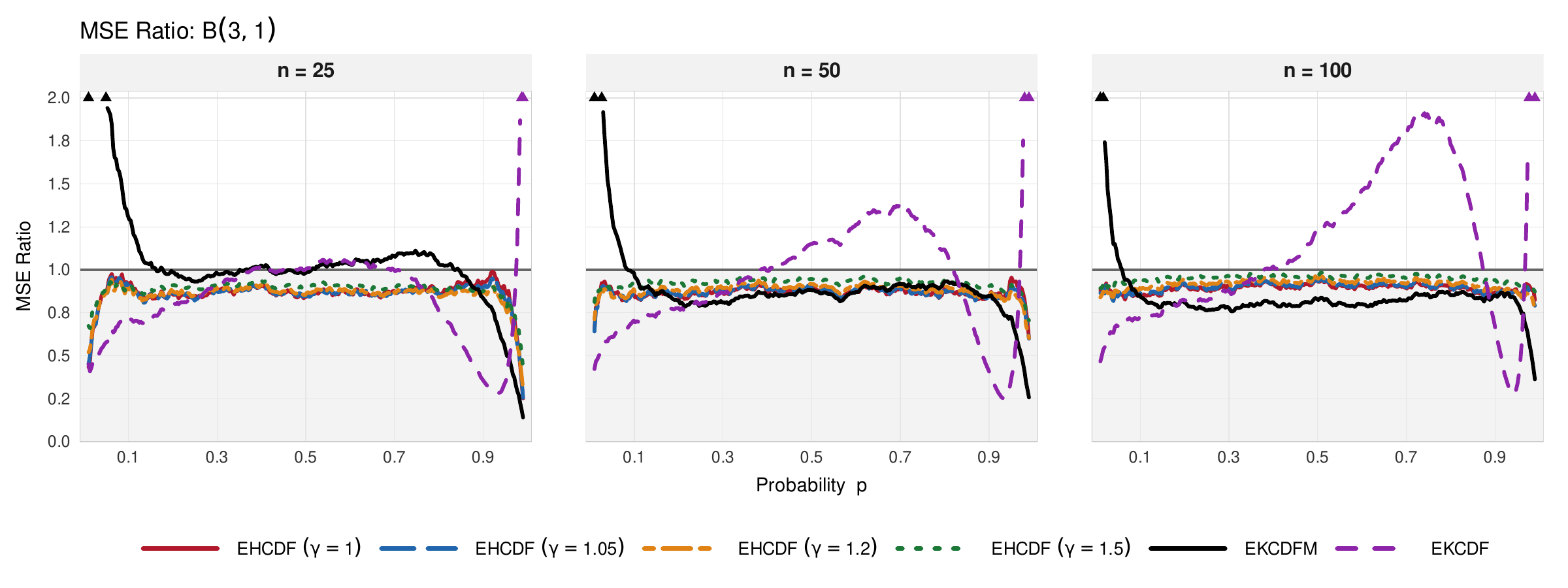}\par\medskip

\caption{
MSE ratio for $\mathcal B(1,1)$, $\mathcal B(2,1)$ and $\mathcal B(3,1)$ (top to bottom). Different lines correspond to different CDF estimators and the dotted horizontal line corresponds to 1.
}
\end{figure}

\begin{figure}[p]
\centering
\includegraphics[width=0.95\textwidth,height=0.28\textheight,keepaspectratio]{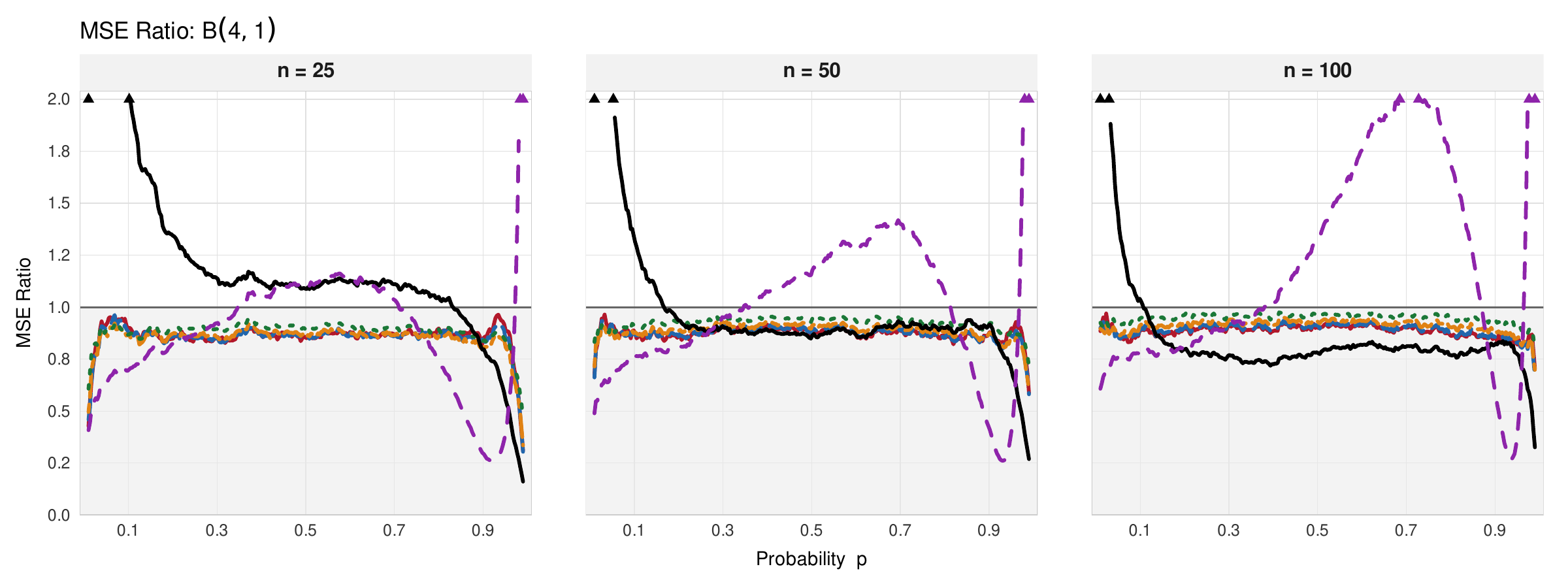}\par\medskip
\includegraphics[width=0.95\textwidth,height=0.28\textheight,keepaspectratio]{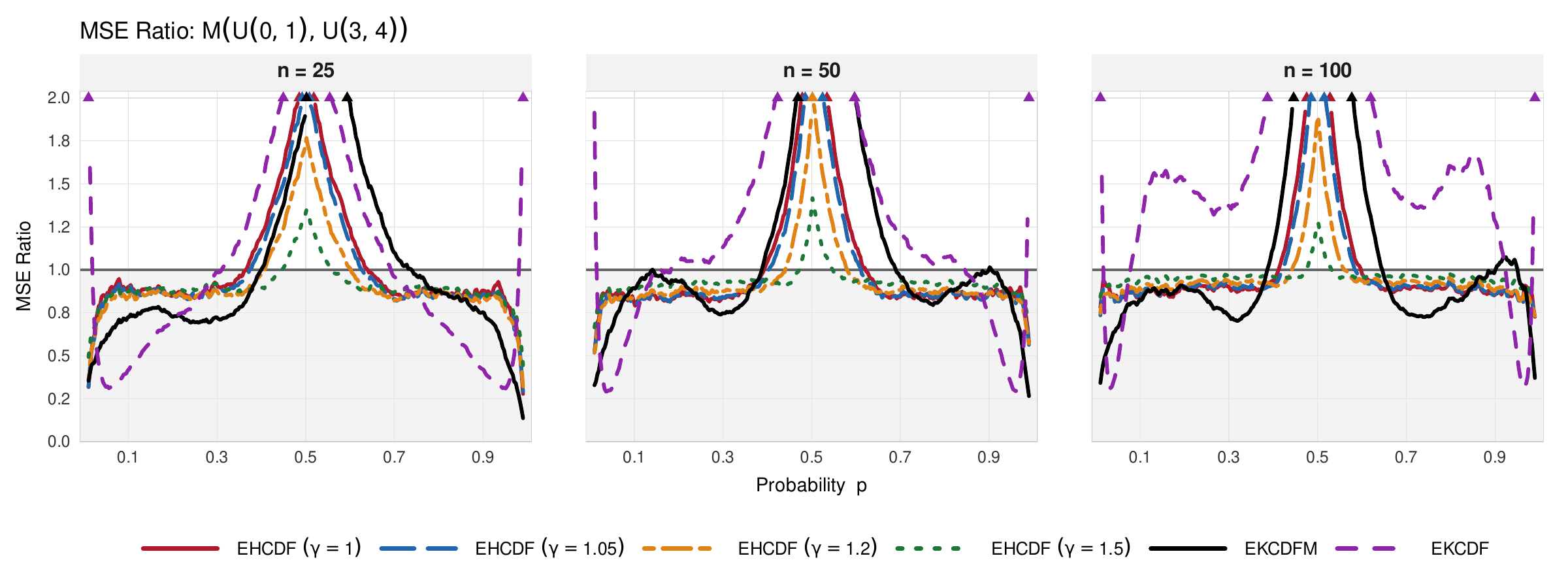}\par\medskip

\caption{
MSE ratio for $\mathcal B(4,1)$ and $\mathcal M(\mathcal U(0,1),\mathcal U(3,4))$. Different lines correspond to different CDF estimators and the dotted horizontal line corresponds to 1.
}
\end{figure}

\newpage

\begin{figure}[p]
\centering
\includegraphics[width=0.95\textwidth,height=0.28\textheight,keepaspectratio]{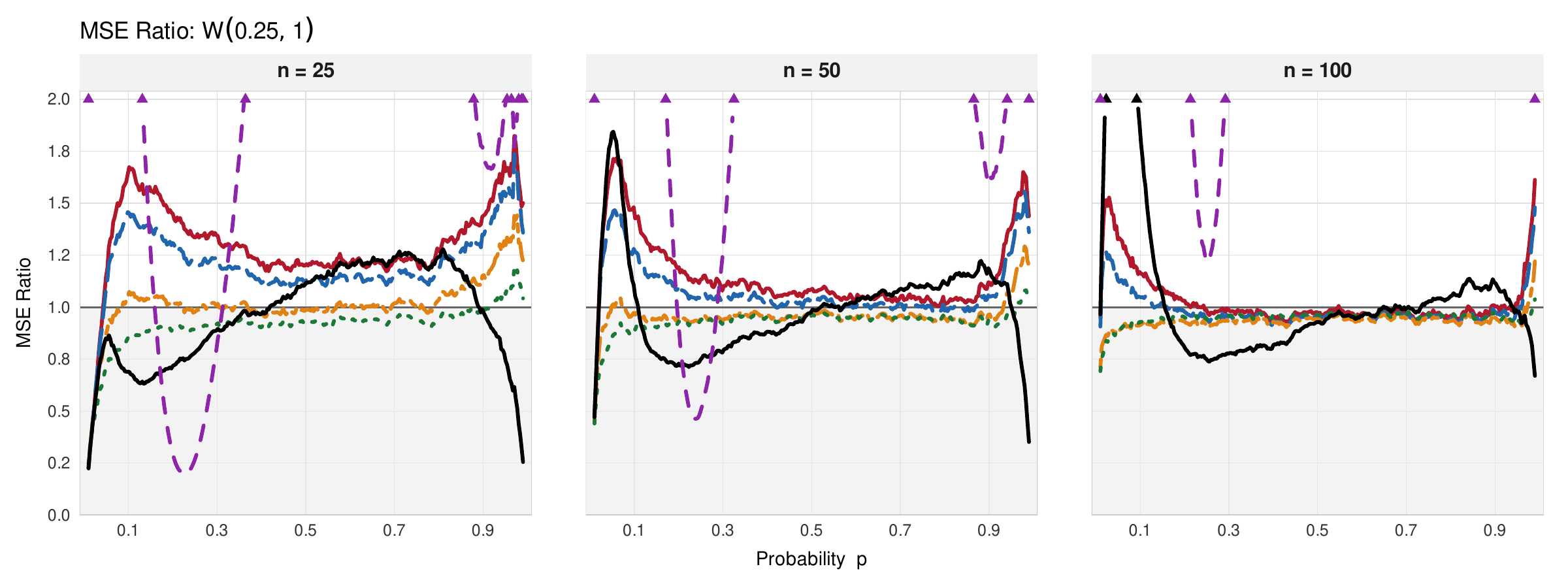}\par\medskip
\includegraphics[width=0.95\textwidth,height=0.28\textheight,keepaspectratio]{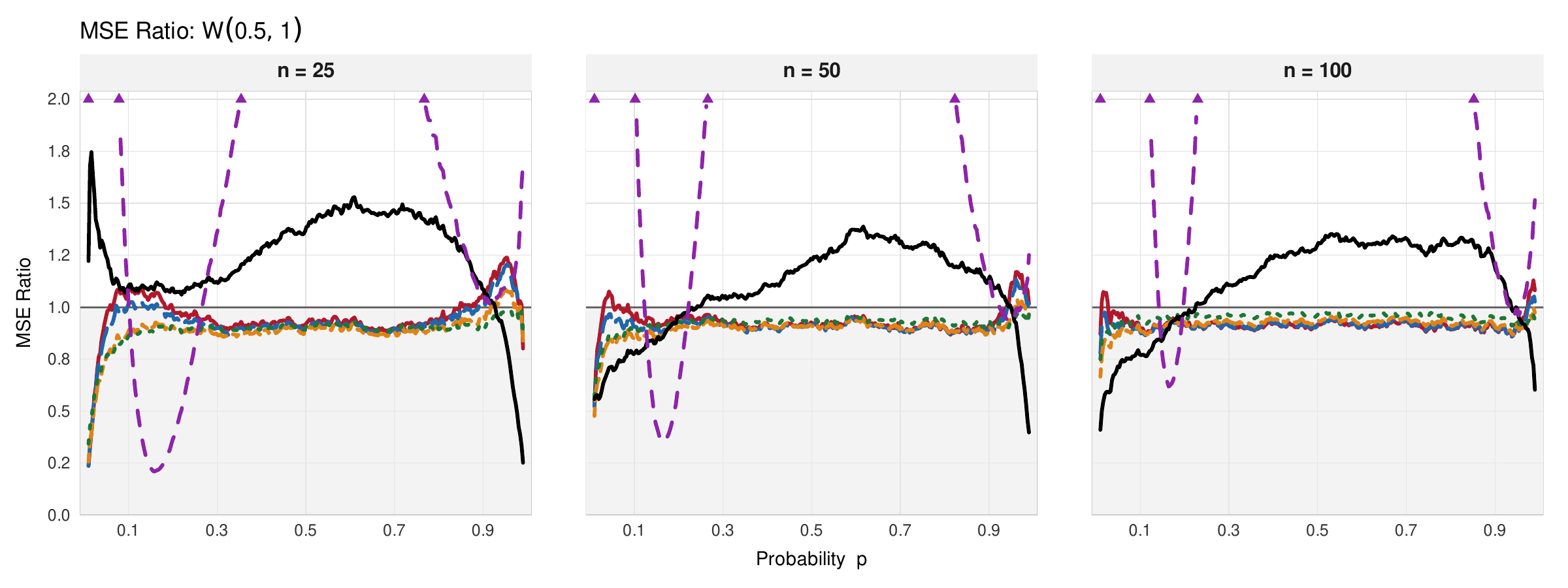}\par\medskip
\includegraphics[width=0.95\textwidth,height=0.28\textheight,keepaspectratio]{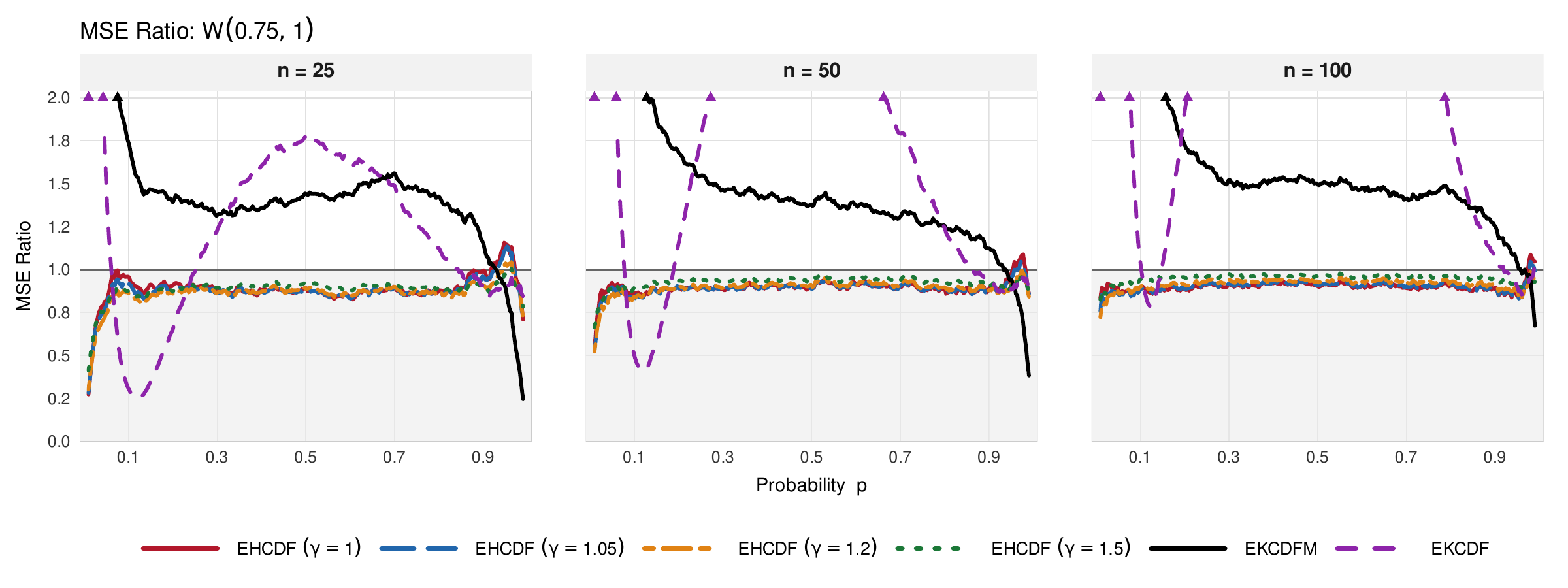}\par\medskip

\caption{
MSE ratio for $\mathcal W(1,0.25)$, $\mathcal W(1,0.5)$ and $\mathcal W(1,0.75)$ (top to bottom). Different lines correspond to different CDF estimators and the dotted horizontal line corresponds to 1.
}
\end{figure}

\begin{figure}[p]
\centering
\includegraphics[width=0.95\textwidth,height=0.28\textheight,keepaspectratio]{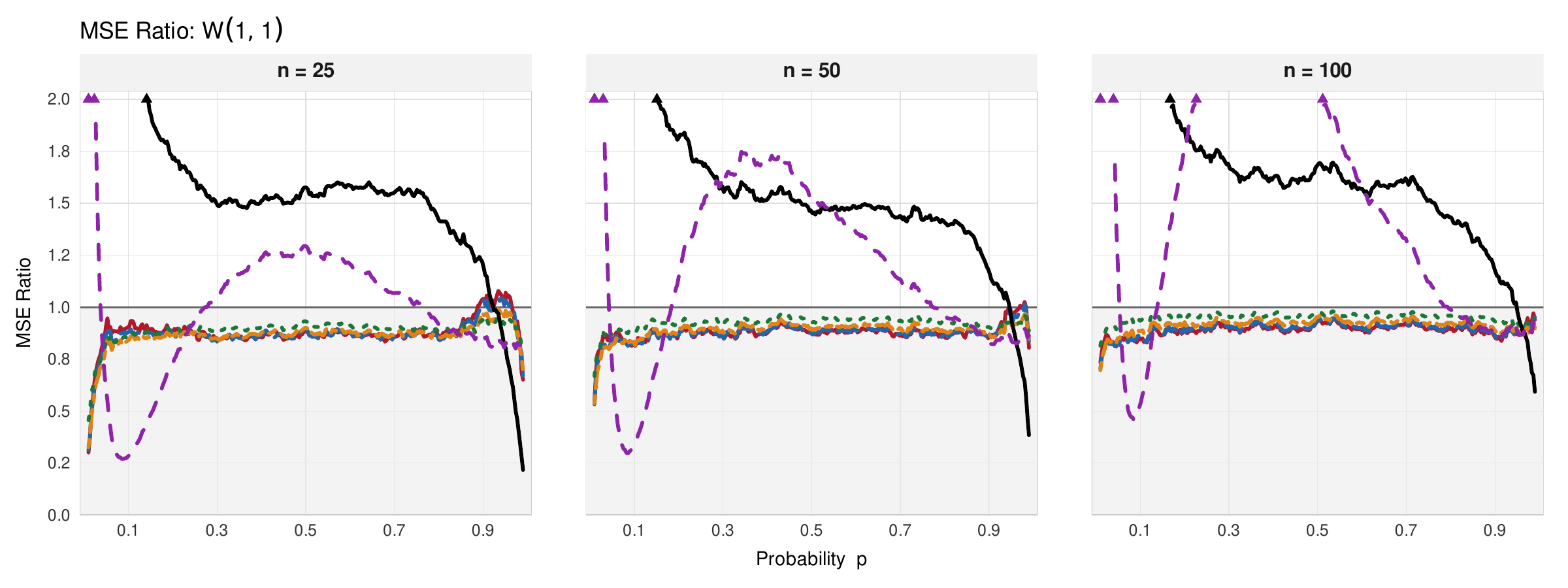}\par\medskip
\includegraphics[width=0.95\textwidth,height=0.28\textheight,keepaspectratio]{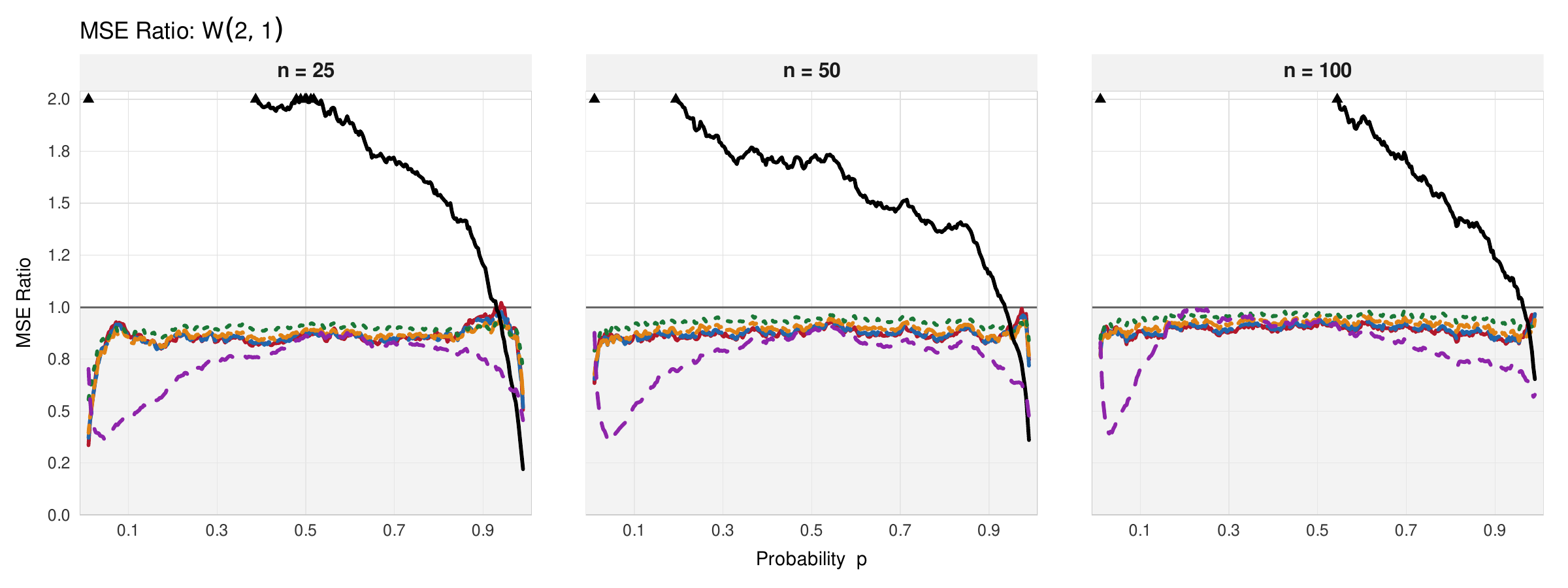}\par\medskip
\includegraphics[width=0.95\textwidth,height=0.28\textheight,keepaspectratio]{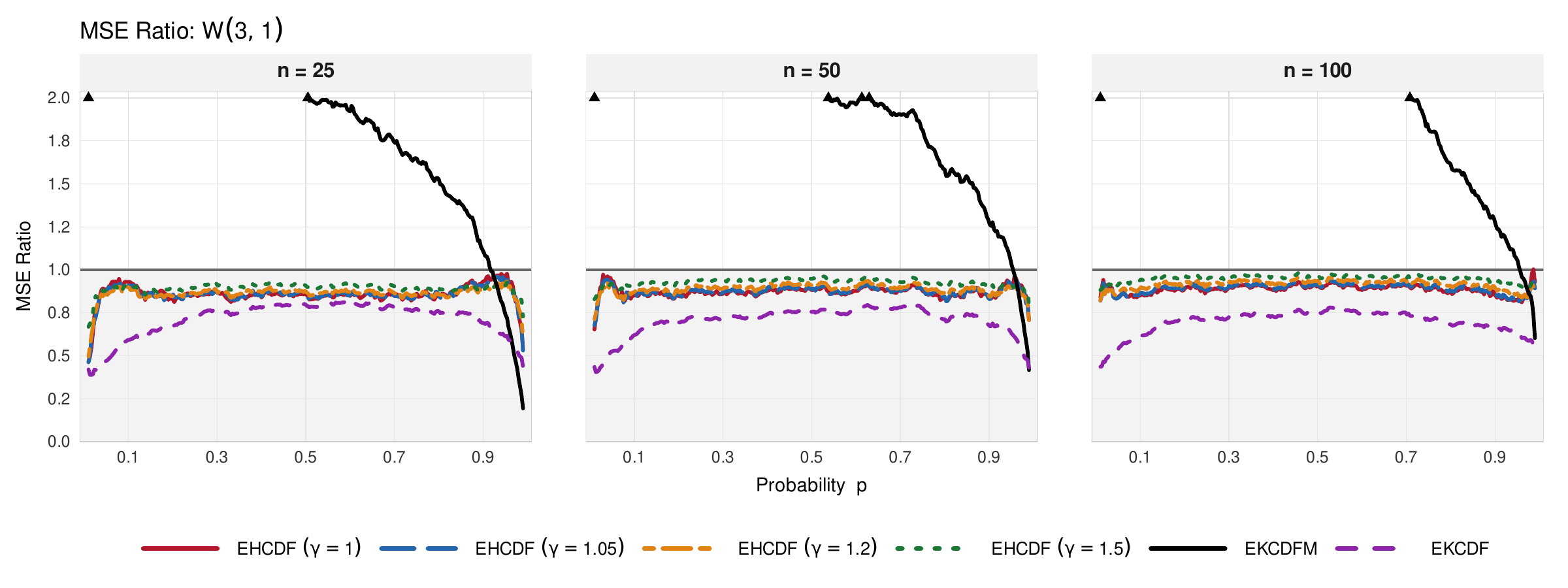}\par\medskip

\caption{
MSE ratio for $\mathcal W(1,1)$, $\mathcal W(1,2)$ and $\mathcal W(1,3)$ (top to bottom). Different lines correspond to different CDF estimators and the dotted horizontal line corresponds to 1.
}
\end{figure}

\begin{figure}[p]
\centering
\includegraphics[width=0.95\textwidth,height=0.28\textheight,keepaspectratio]{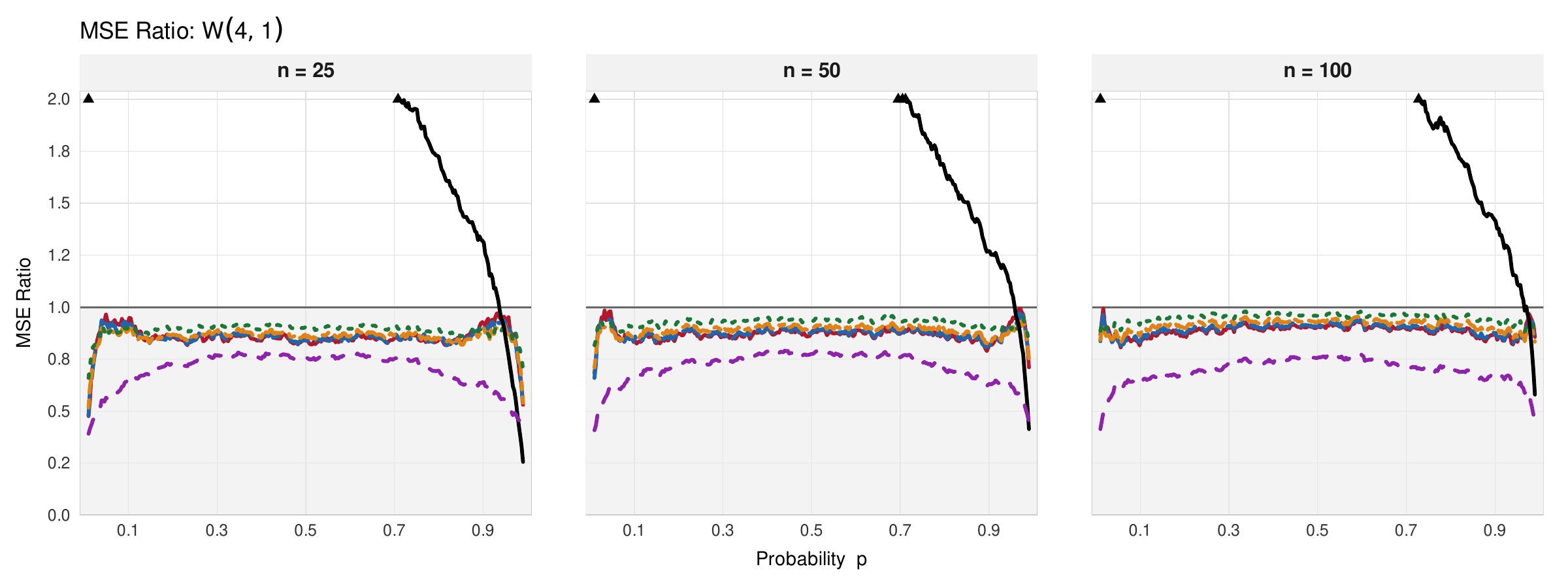}\par\medskip
\includegraphics[width=0.95\textwidth,height=0.28\textheight,keepaspectratio]{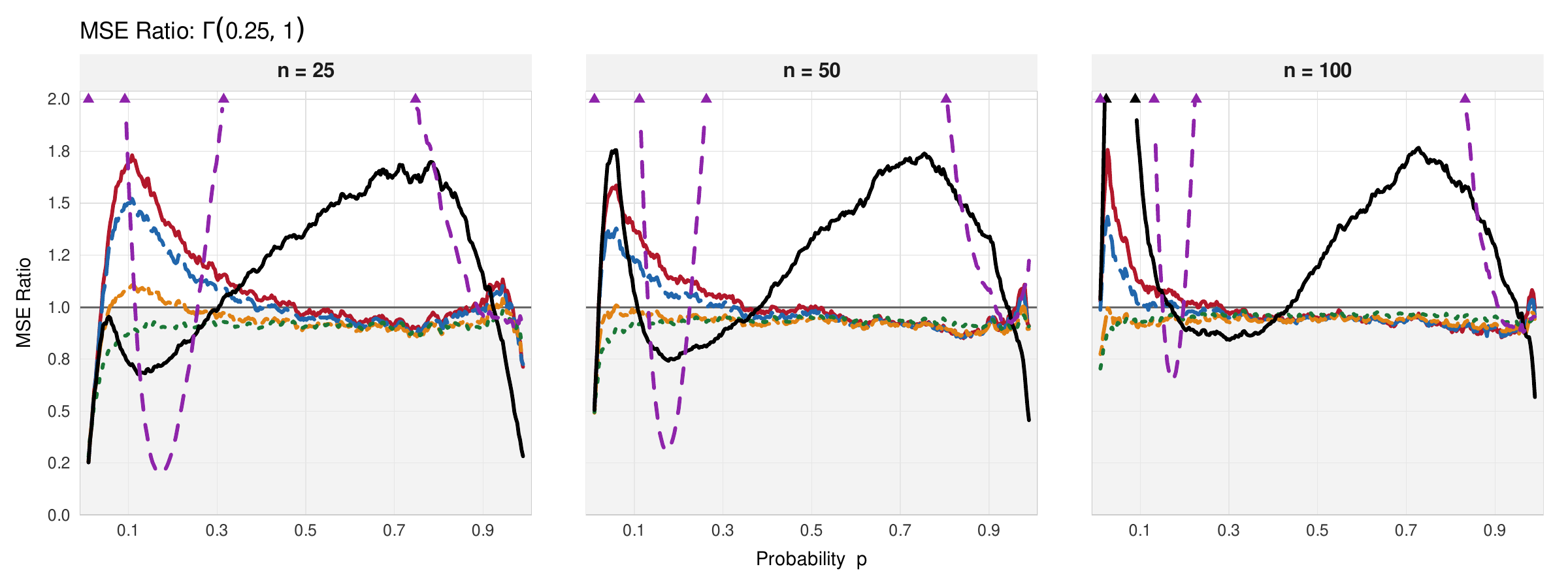}\par\medskip
\includegraphics[width=0.95\textwidth,height=0.28\textheight,keepaspectratio]{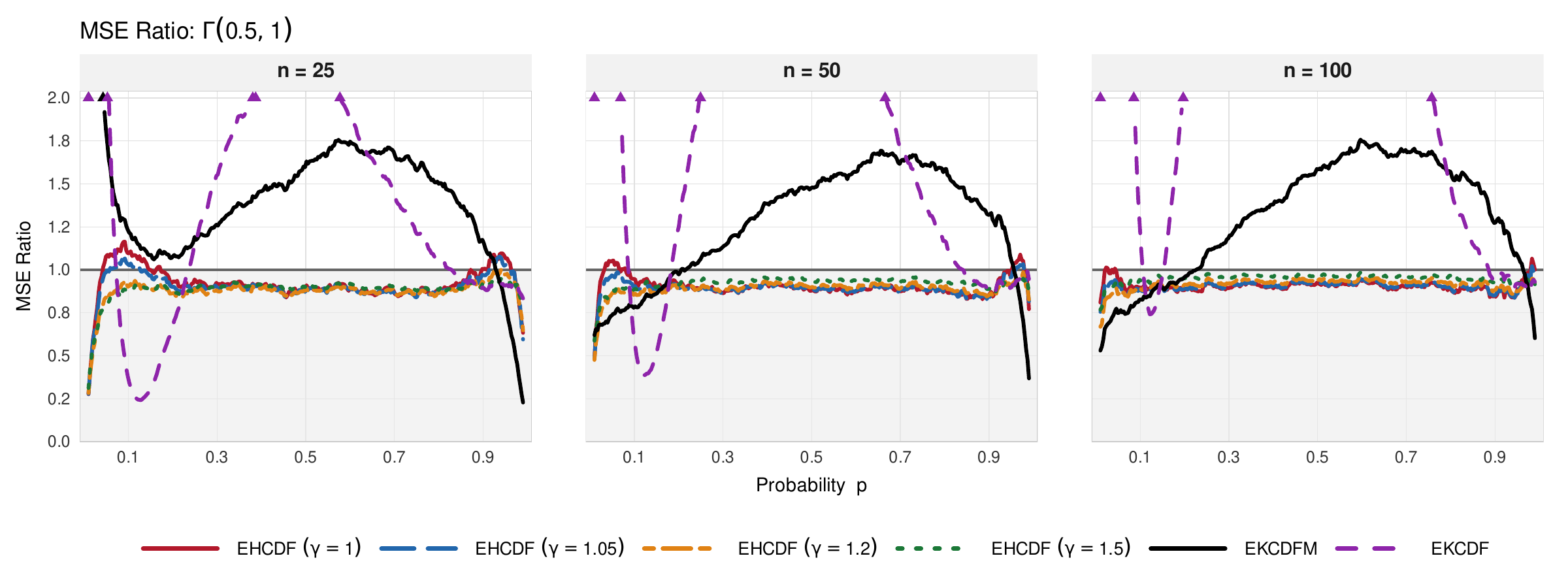}\par\medskip

\caption{
MSE ratio for $\mathcal W(1,4)$, $\Gamma(0.25,1)$ and $\Gamma(0.5,1)$ (top to bottom). Different lines correspond to different CDF estimators and the dotted horizontal line corresponds to 1.
}
\end{figure}

\begin{figure}[p]
\centering
\includegraphics[width=0.95\textwidth,height=0.28\textheight,keepaspectratio]{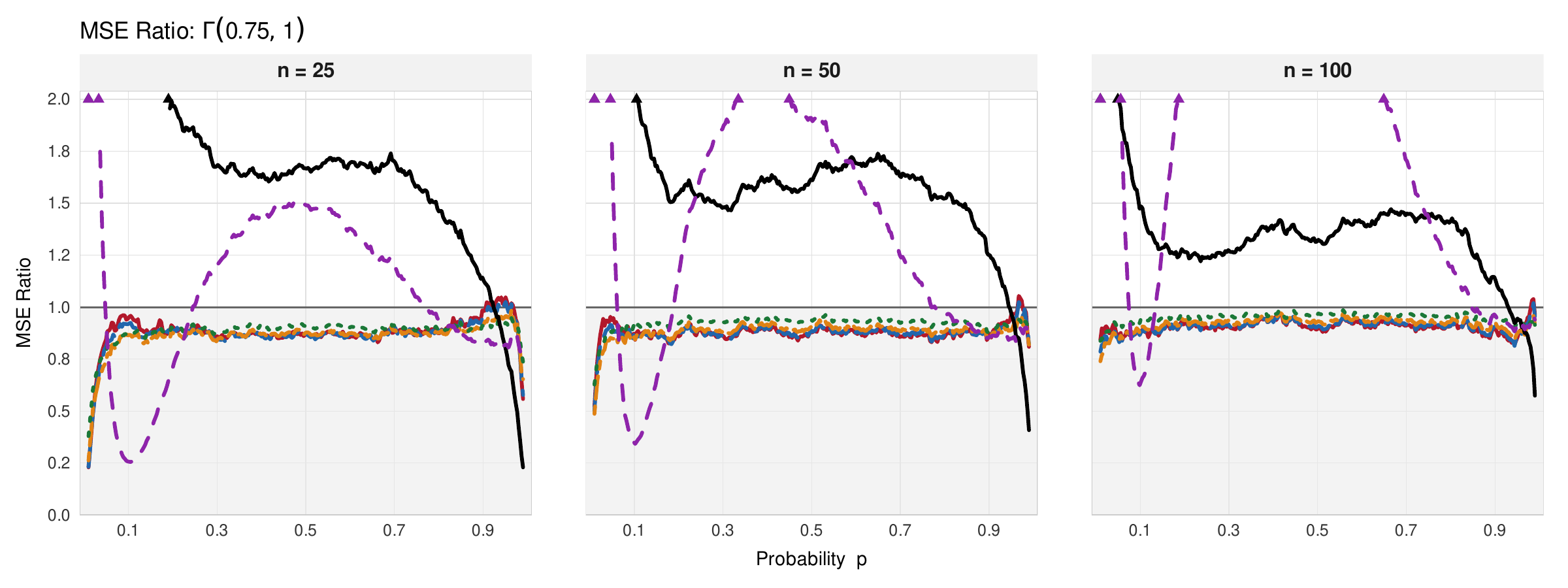}\par\medskip
\includegraphics[width=0.95\textwidth,height=0.28\textheight,keepaspectratio]{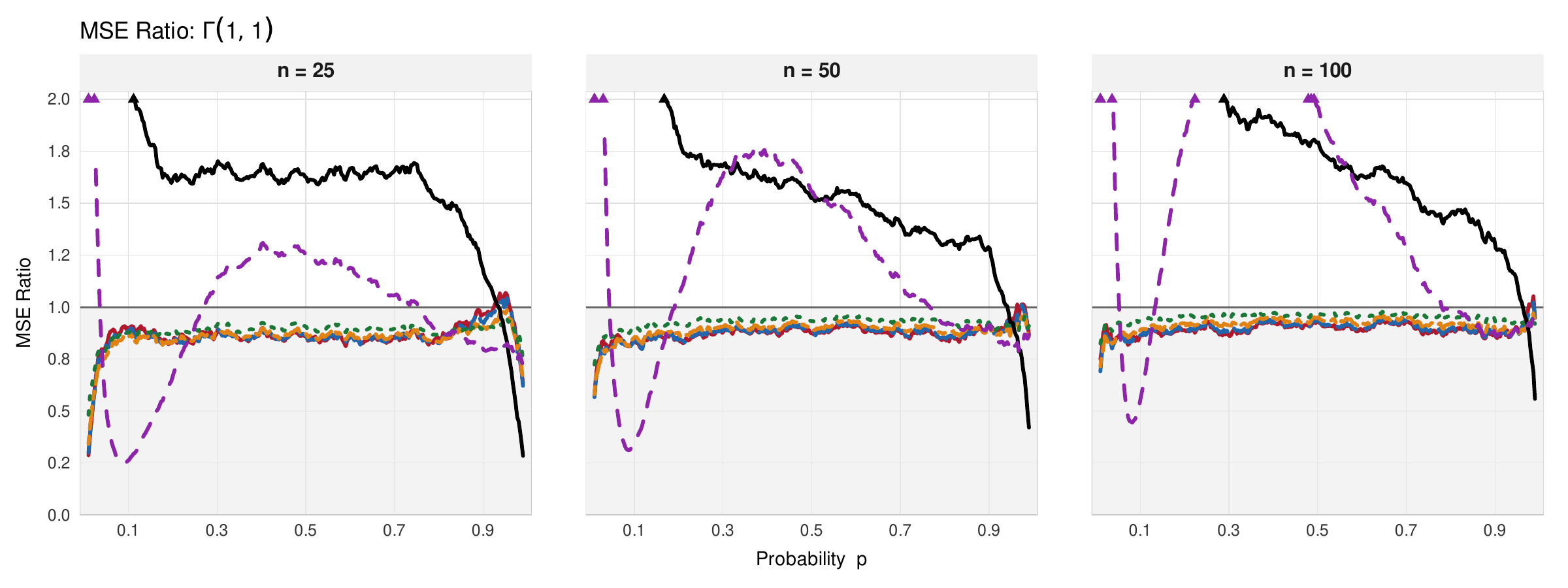}\par\medskip
\includegraphics[width=0.95\textwidth,height=0.28\textheight,keepaspectratio]{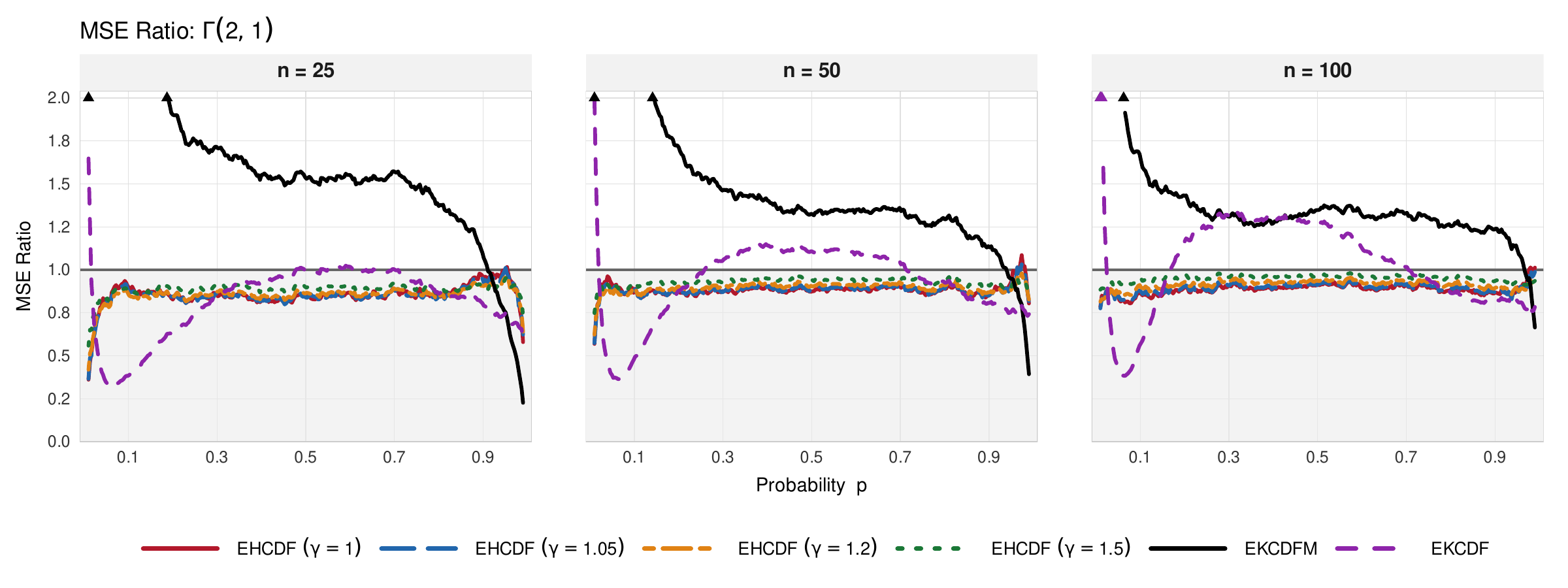}\par\medskip

\caption{
MSE ratio for $\Gamma(0.75,1)$, $\Gamma(1,1)$ and $\Gamma(2,1)$ (top to bottom). Different lines correspond to different CDF estimators and the dotted horizontal line corresponds to 1.
}
\end{figure}

\begin{figure}[p]
\centering
\includegraphics[width=0.95\textwidth,height=0.28\textheight,keepaspectratio]{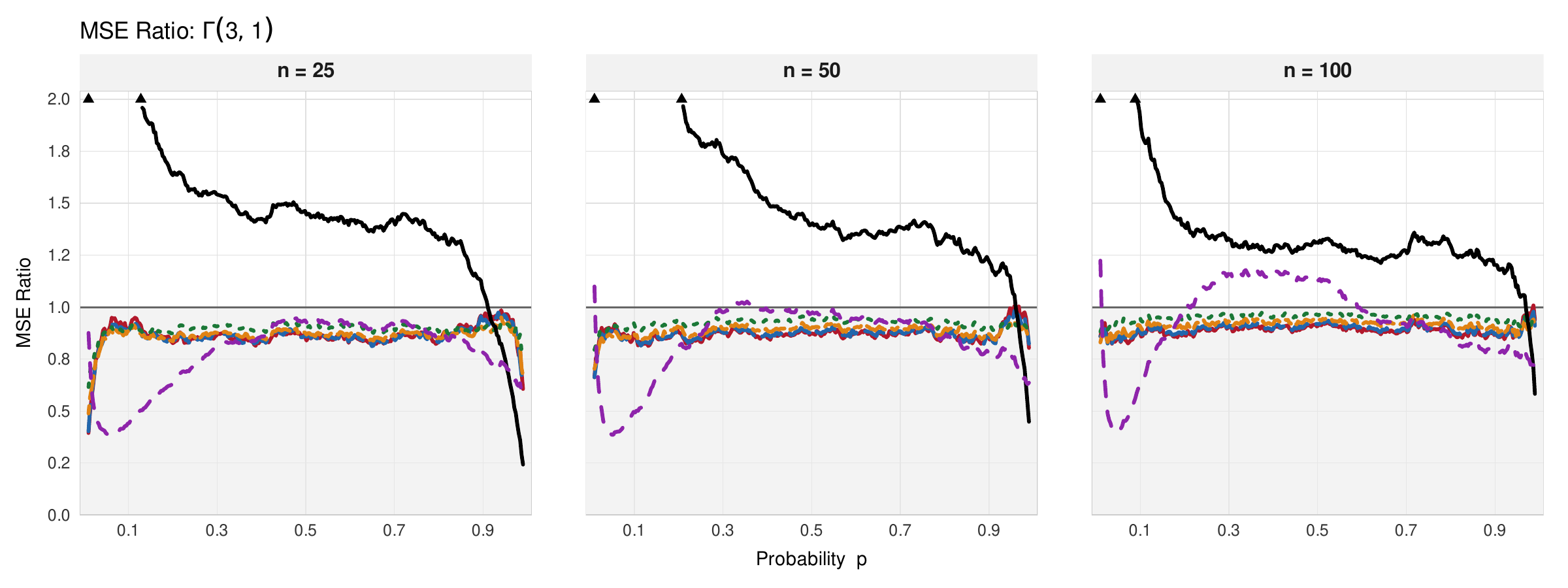}\par\medskip
\includegraphics[width=0.95\textwidth,height=0.28\textheight,keepaspectratio]{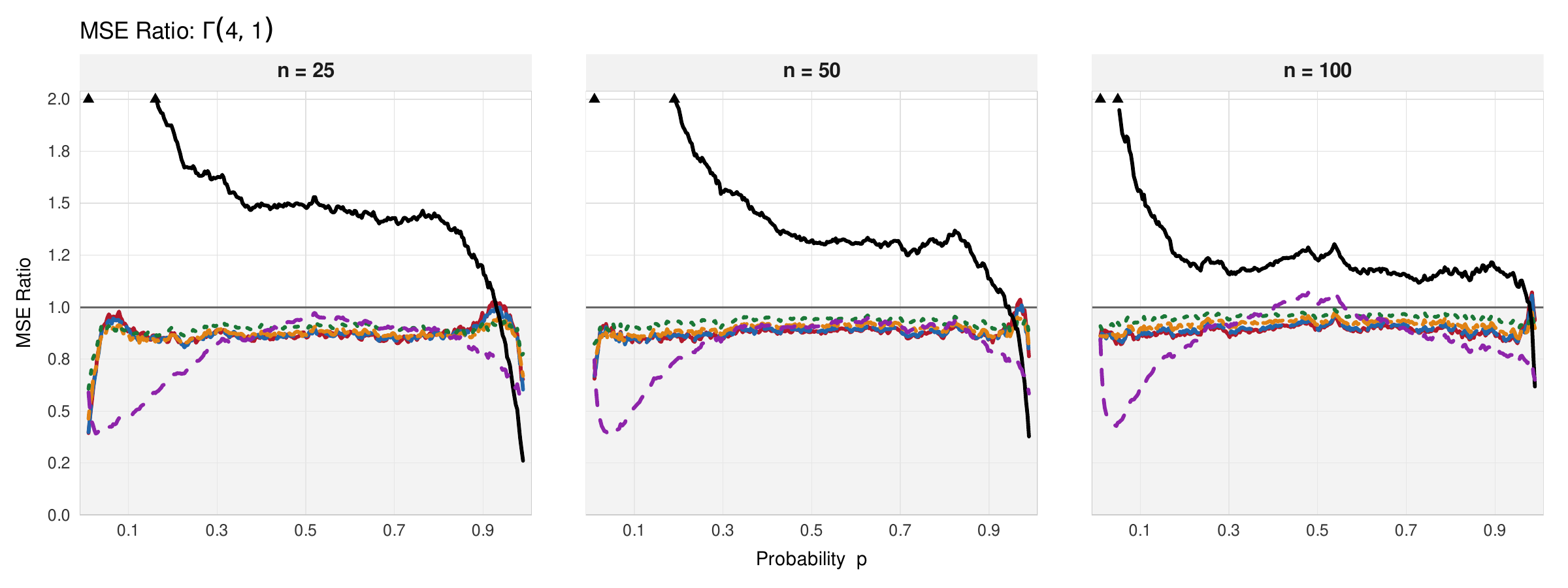}\par\medskip
\includegraphics[width=0.95\textwidth,height=0.28\textheight,keepaspectratio]{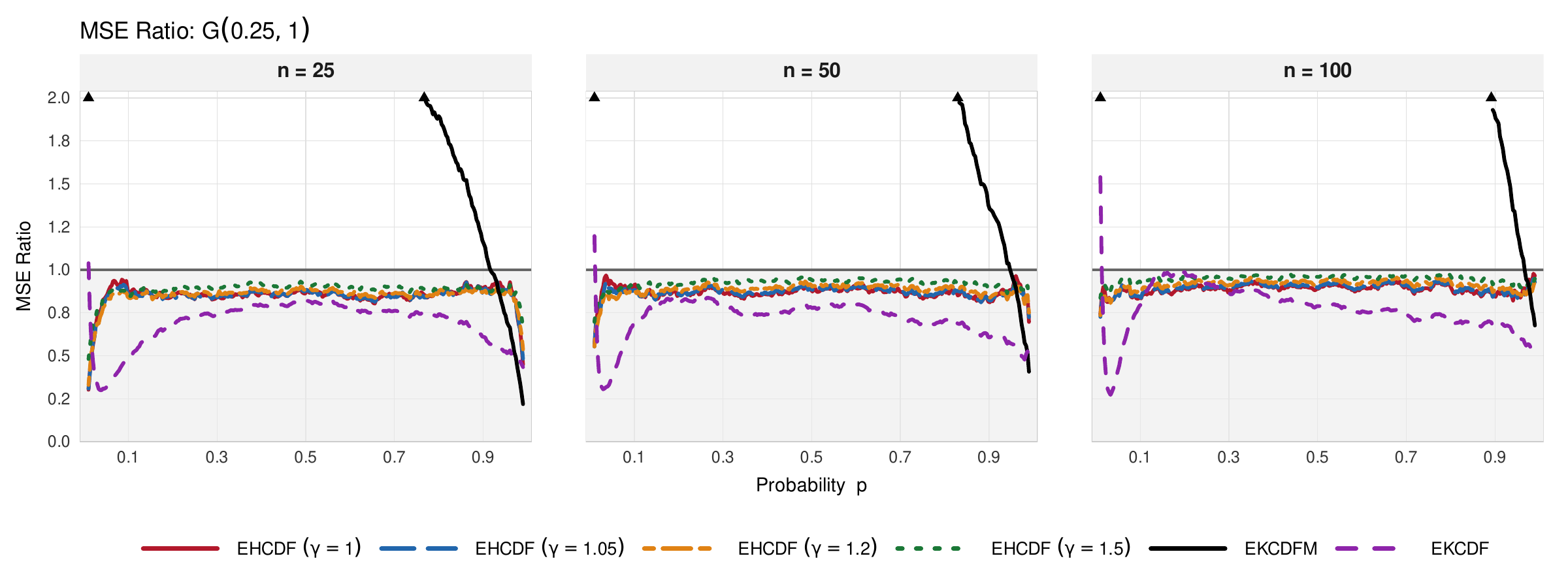}\par\medskip

\caption{
MSE ratio for $\Gamma(3,1)$, $\Gamma(4,1)$ and $\mathcal G(0.25,1)$ (top to bottom). Different lines correspond to different CDF estimators and the dotted horizontal line corresponds to 1.
}
\end{figure}

\begin{figure}[p]
\centering
\includegraphics[width=0.95\textwidth,height=0.28\textheight,keepaspectratio]{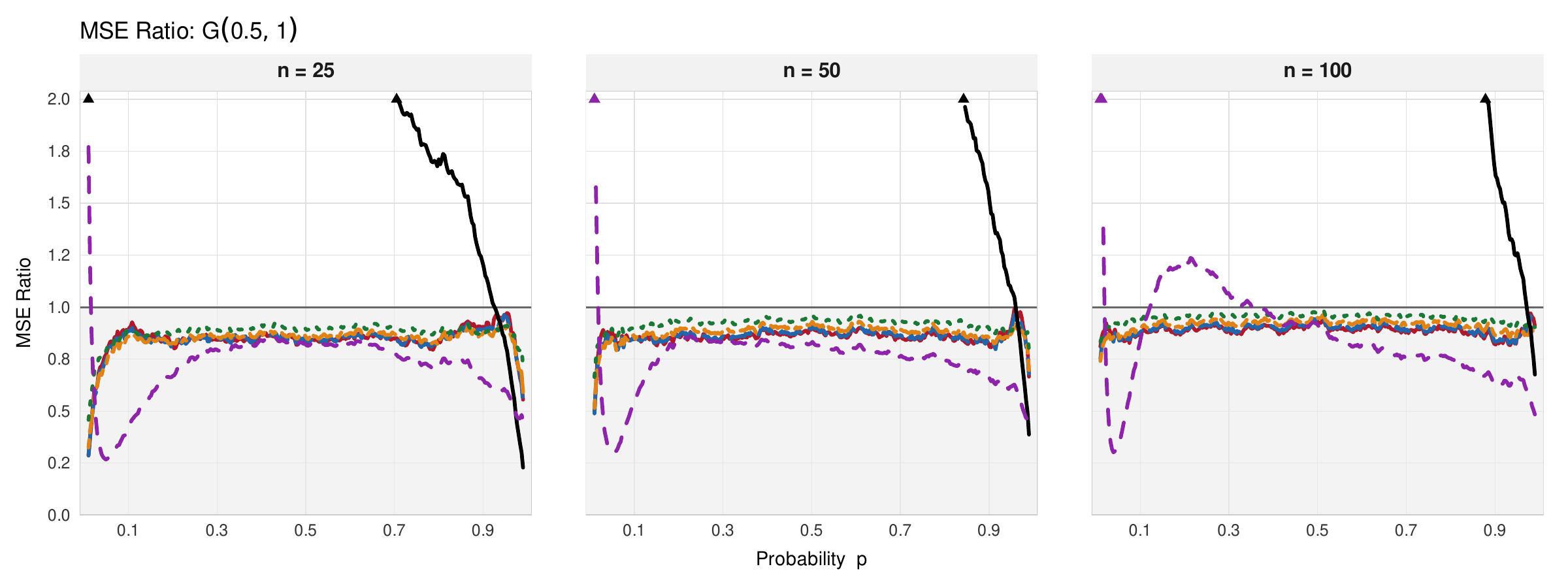}\par\medskip
\includegraphics[width=0.95\textwidth,height=0.28\textheight,keepaspectratio]{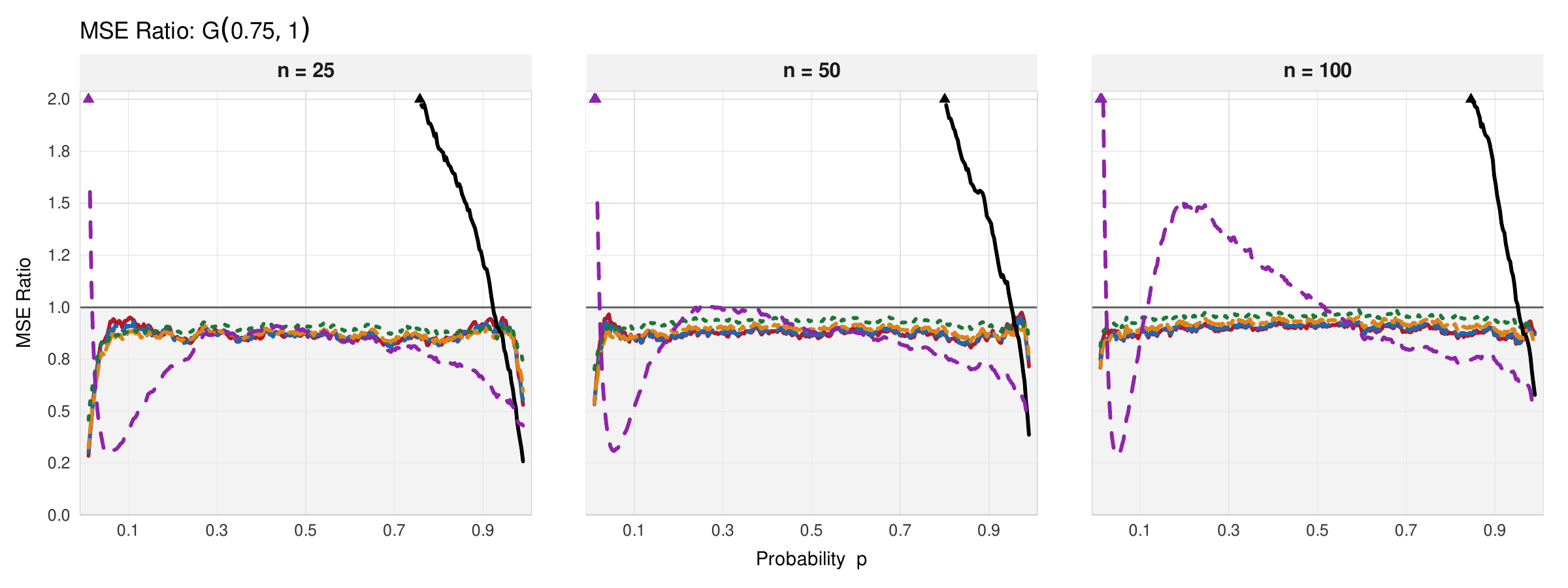}\par\medskip
\includegraphics[width=0.95\textwidth,height=0.28\textheight,keepaspectratio]{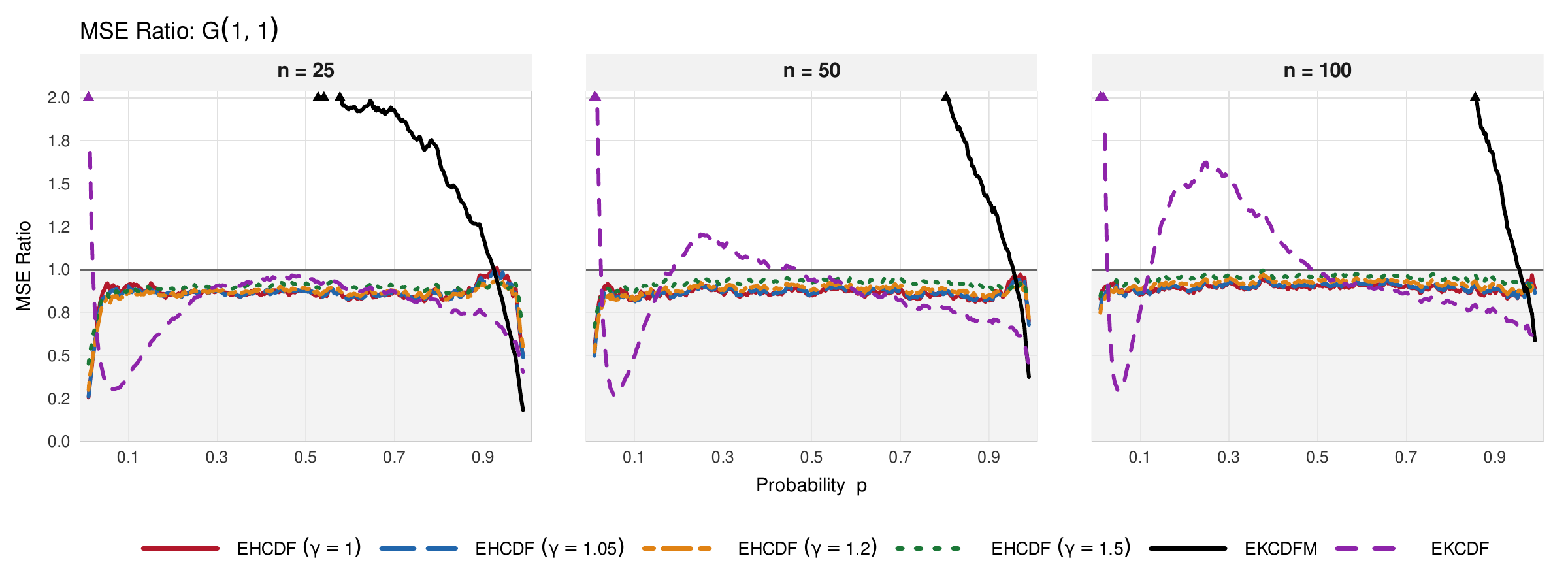}\par\medskip

\caption{
MSE ratio for $\mathcal G(0.5,1)$, $\mathcal G(0.75,1)$ and $\mathcal G(1,1)$ (top to bottom). Different lines correspond to different CDF estimators and the dotted horizontal line corresponds to 1.
}
\end{figure}

\begin{figure}[p]
\centering
\includegraphics[width=0.95\textwidth,height=0.28\textheight,keepaspectratio]{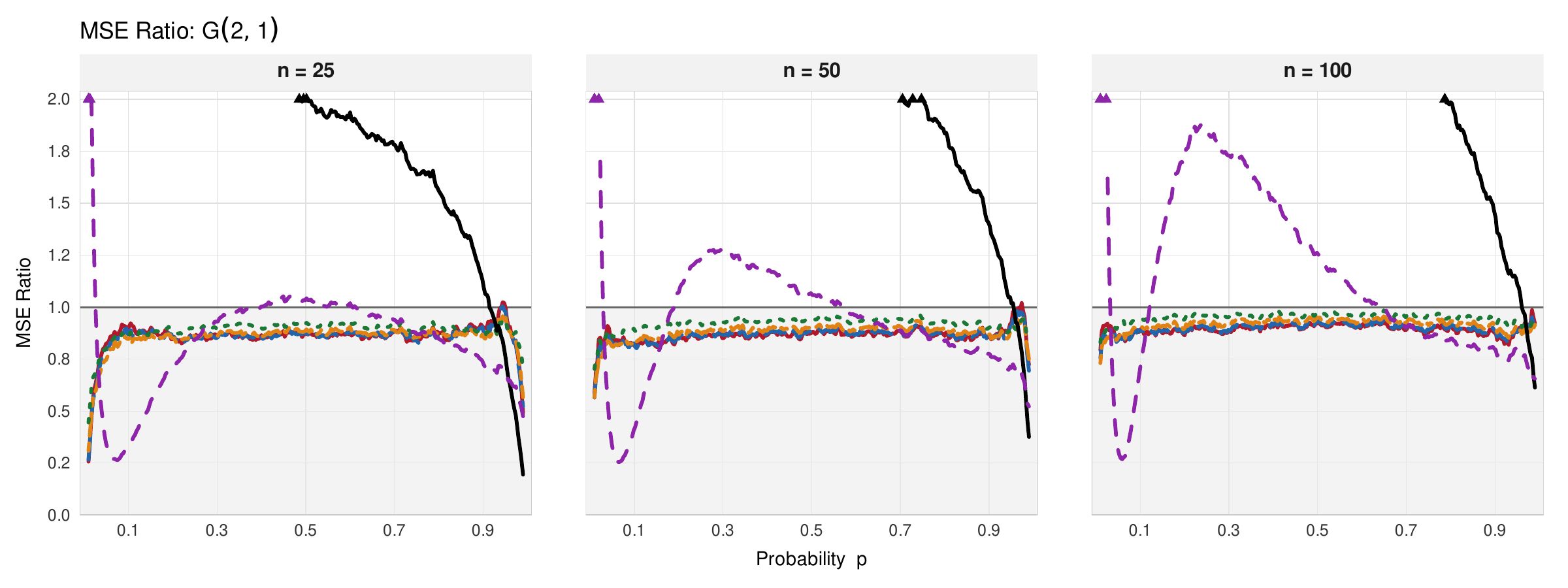}\par\medskip
\includegraphics[width=0.95\textwidth,height=0.28\textheight,keepaspectratio]{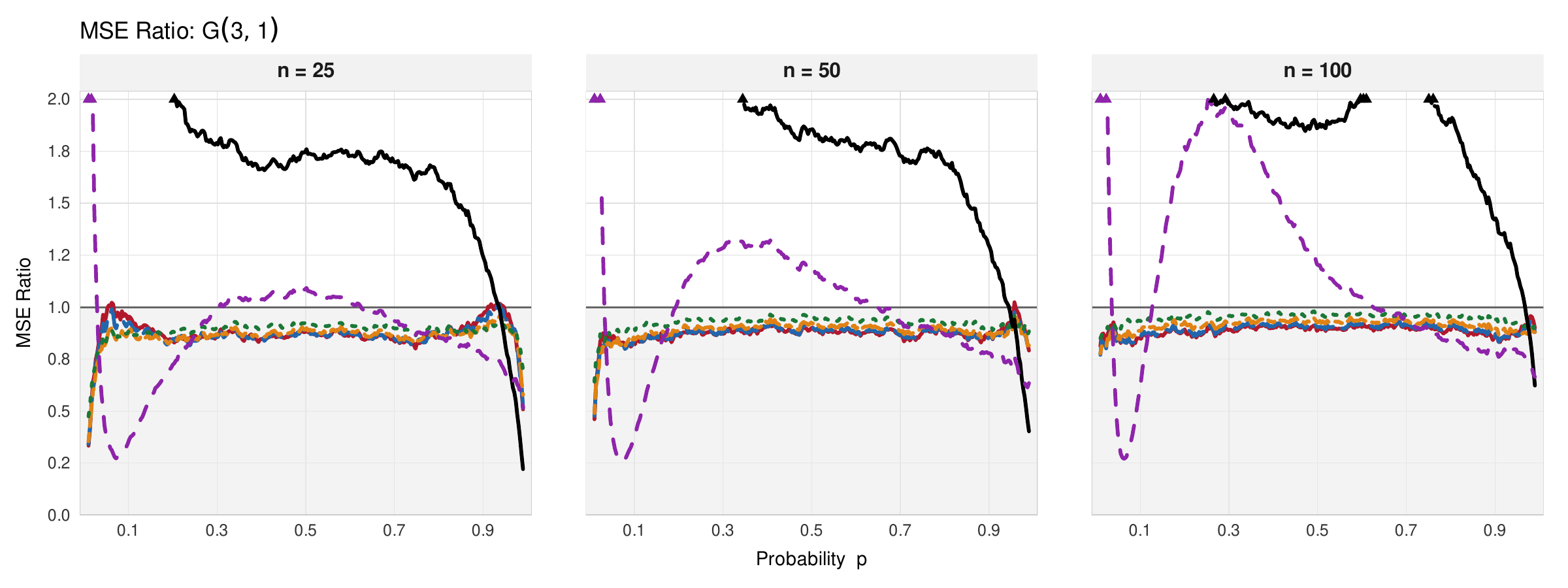}\par\medskip
\includegraphics[width=0.95\textwidth,height=0.28\textheight,keepaspectratio]{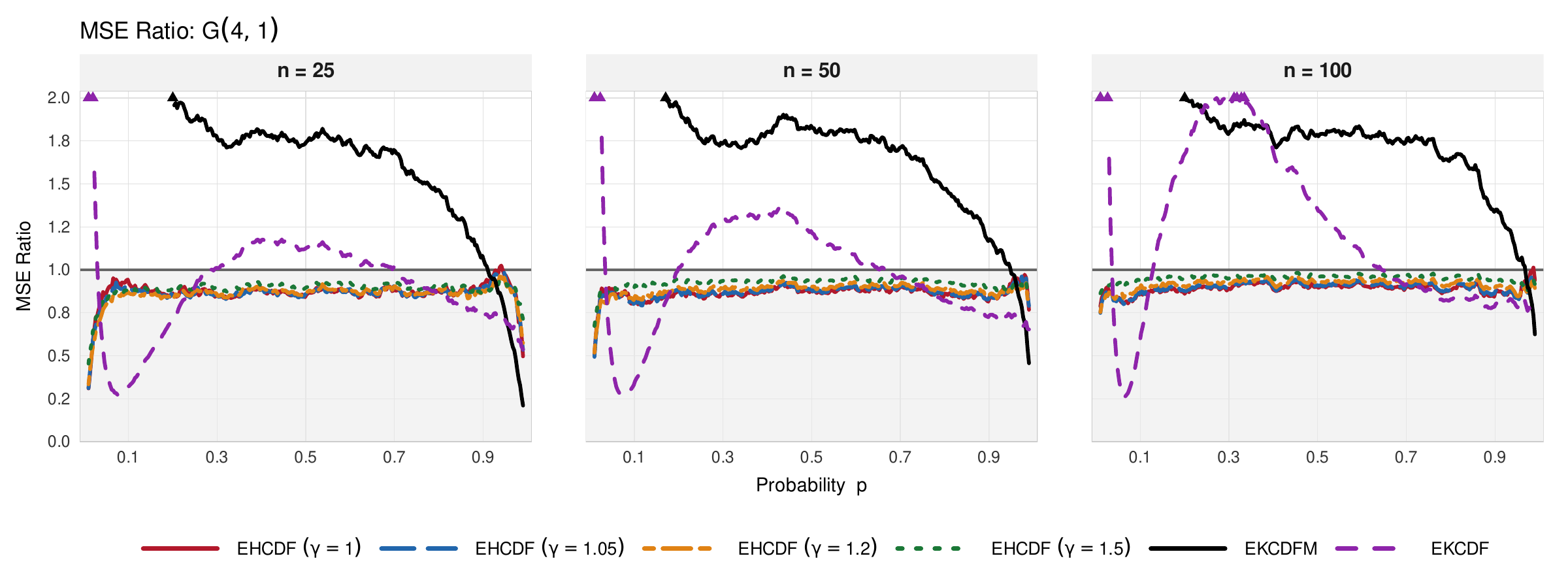}\par\medskip

\caption{
MSE ratio for $\mathcal G(2,1)$, $\mathcal G(3,1)$ and $\mathcal G(4,1)$ (top to bottom). Different lines correspond to different CDF estimators and the dotted horizontal line corresponds to 1.
}
\end{figure}

\begin{figure}[p]
\centering
\includegraphics[width=0.95\textwidth,height=0.28\textheight,keepaspectratio]{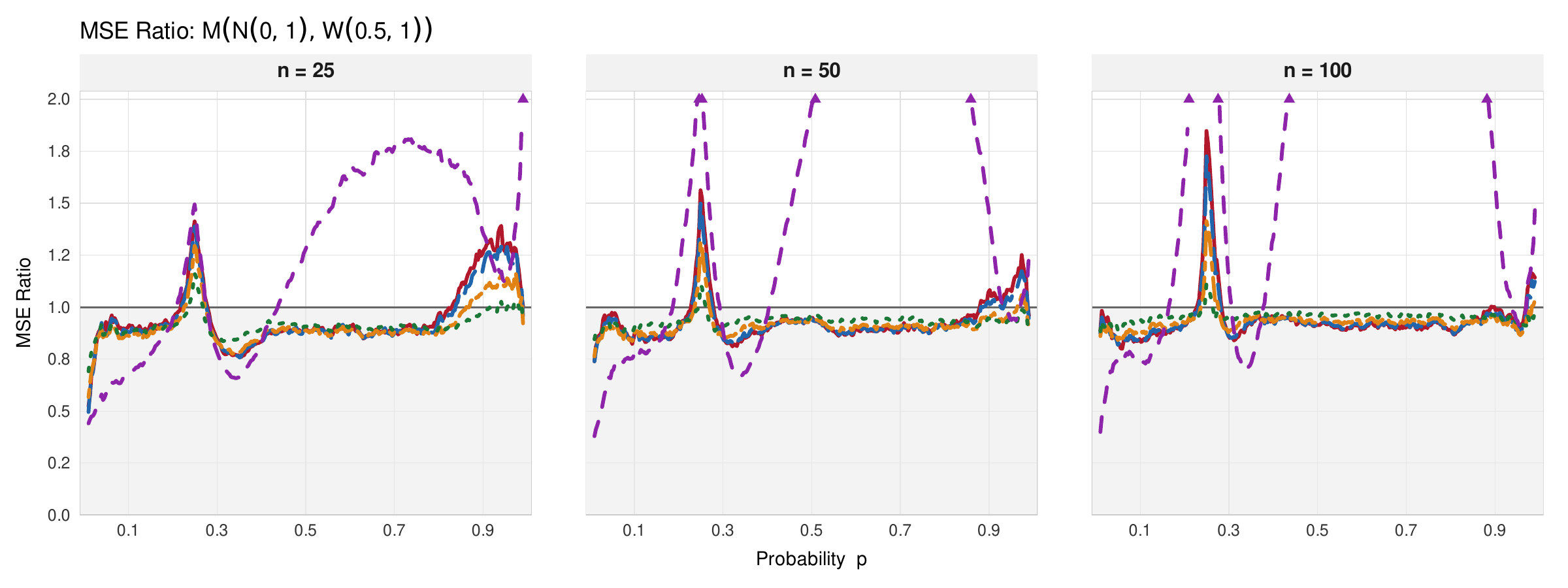}\par\medskip
\includegraphics[width=0.95\textwidth,height=0.28\textheight,keepaspectratio]{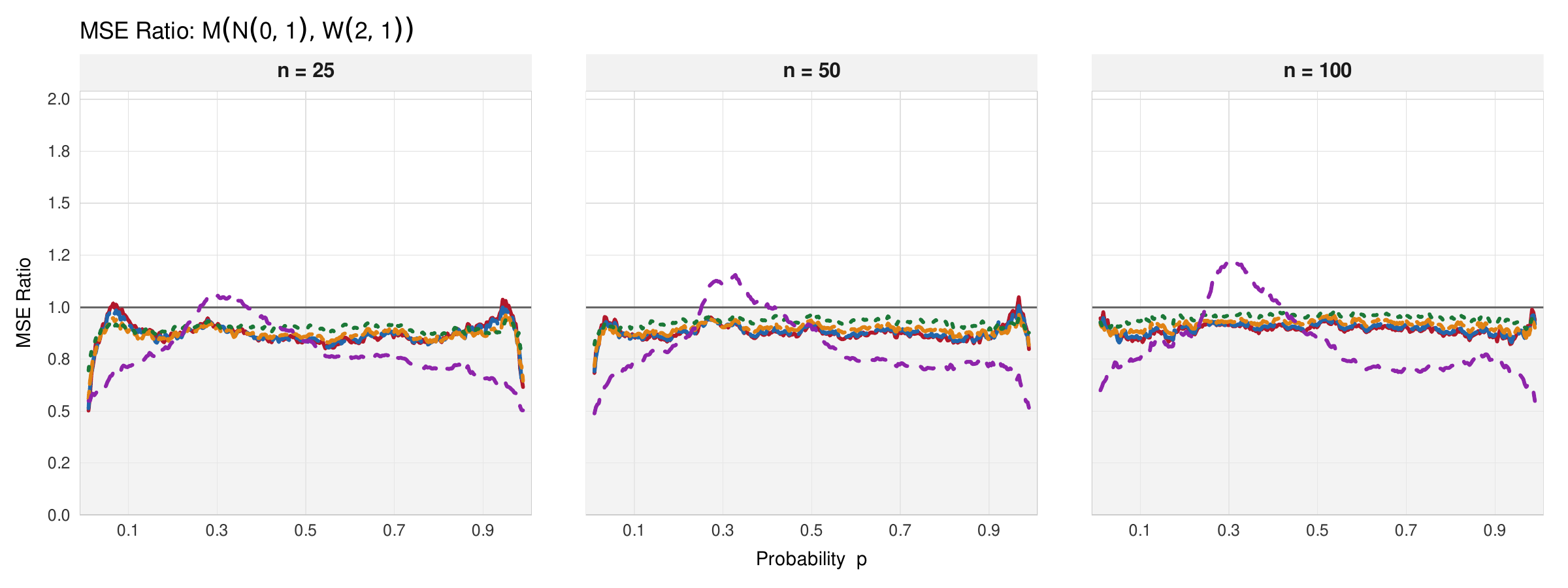}\par\medskip
\includegraphics[width=0.95\textwidth,height=0.28\textheight,keepaspectratio]{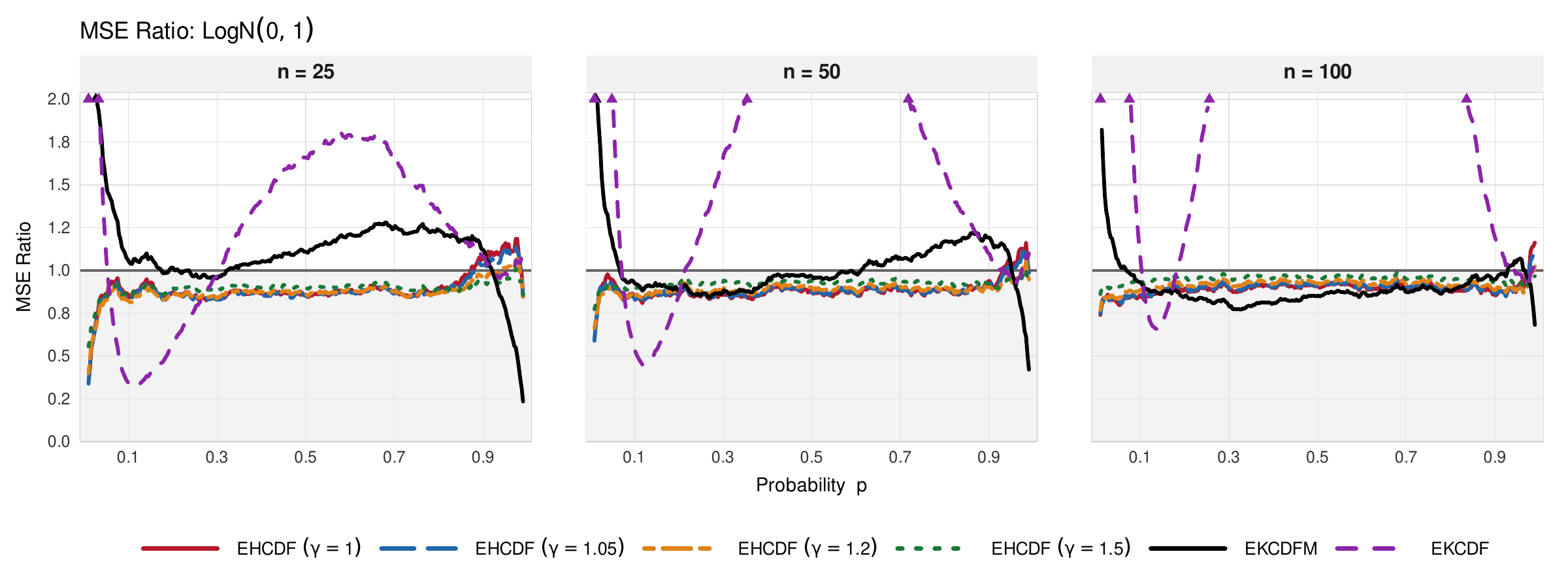}\par\medskip

\caption{
MSE ratio for $\mathcal M(\mathcal N(0,1),\mathcal W(1,2))$, $\mathcal M(\mathcal N(0,1),\mathcal W(1,0.5))$ and $\mathrm{LogN}(0,1)$ (top to bottom). Different lines correspond to different CDF estimators and the dotted horizontal line corresponds to 1.
}
\end{figure}

\begin{figure}[p]
\centering
\includegraphics[width=0.95\textwidth,height=0.28\textheight,keepaspectratio]{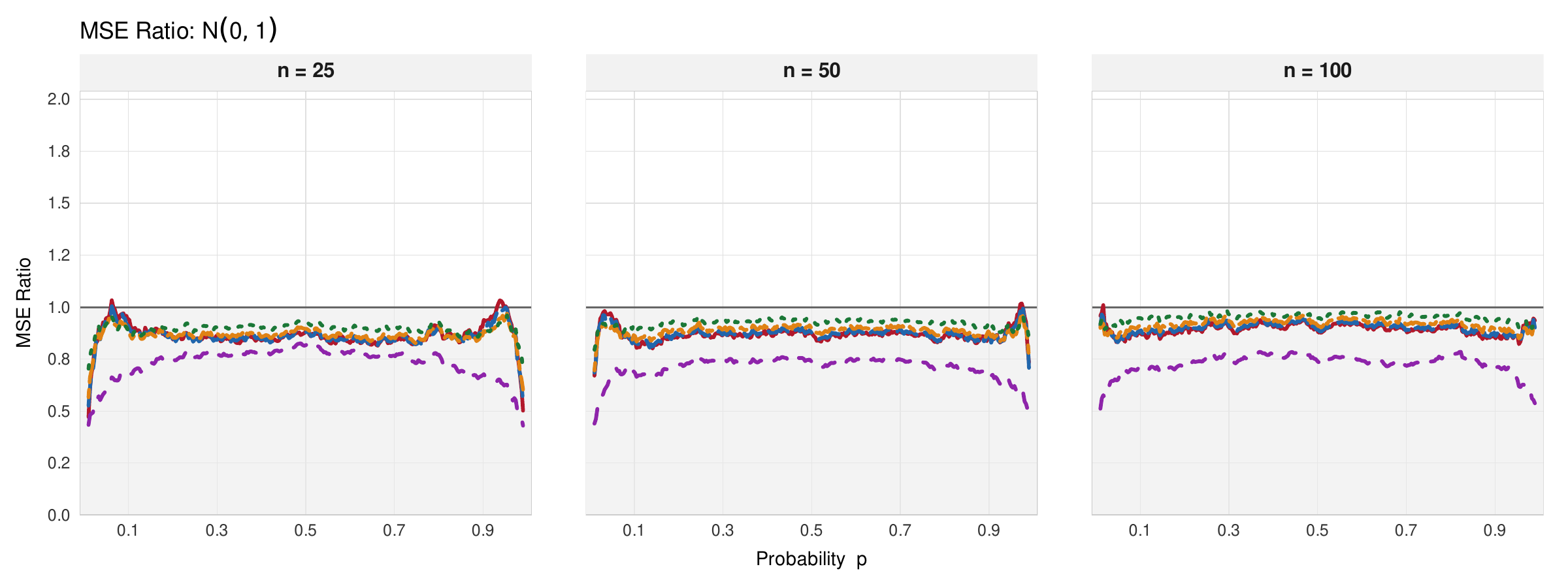}\par\medskip
\includegraphics[width=0.95\textwidth,height=0.28\textheight,keepaspectratio]{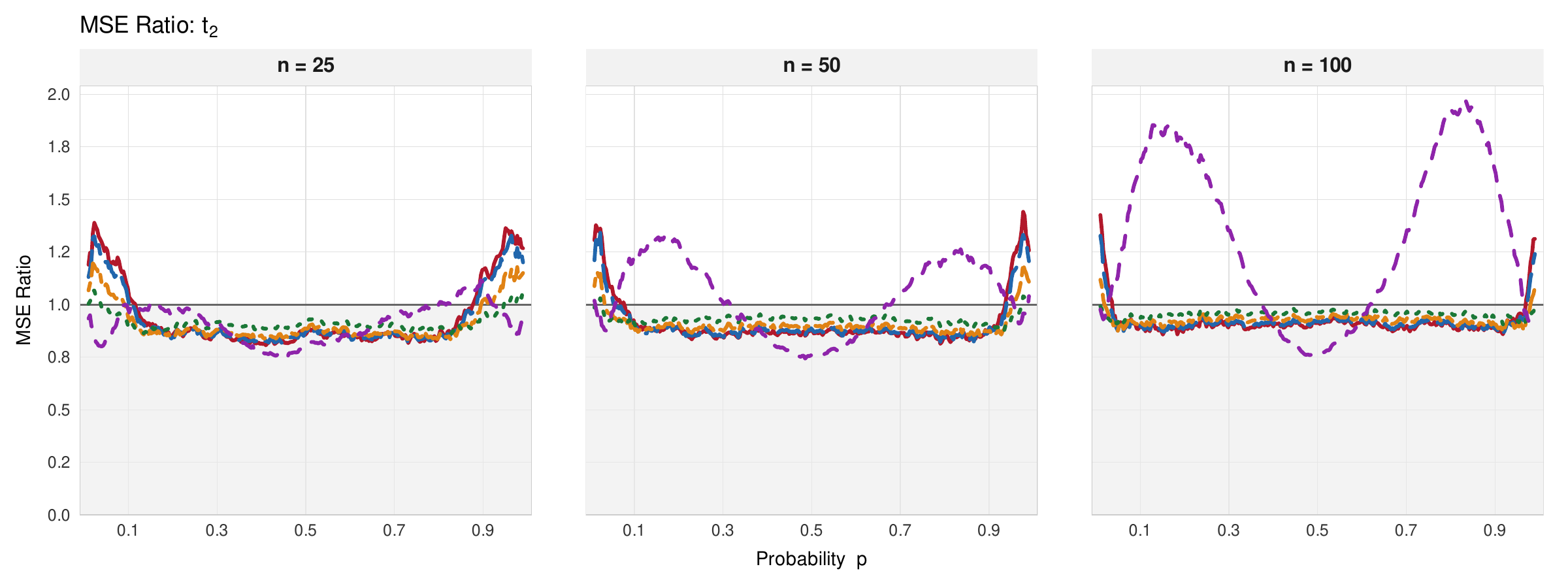}\par\medskip
\includegraphics[width=0.95\textwidth,height=0.28\textheight,keepaspectratio]{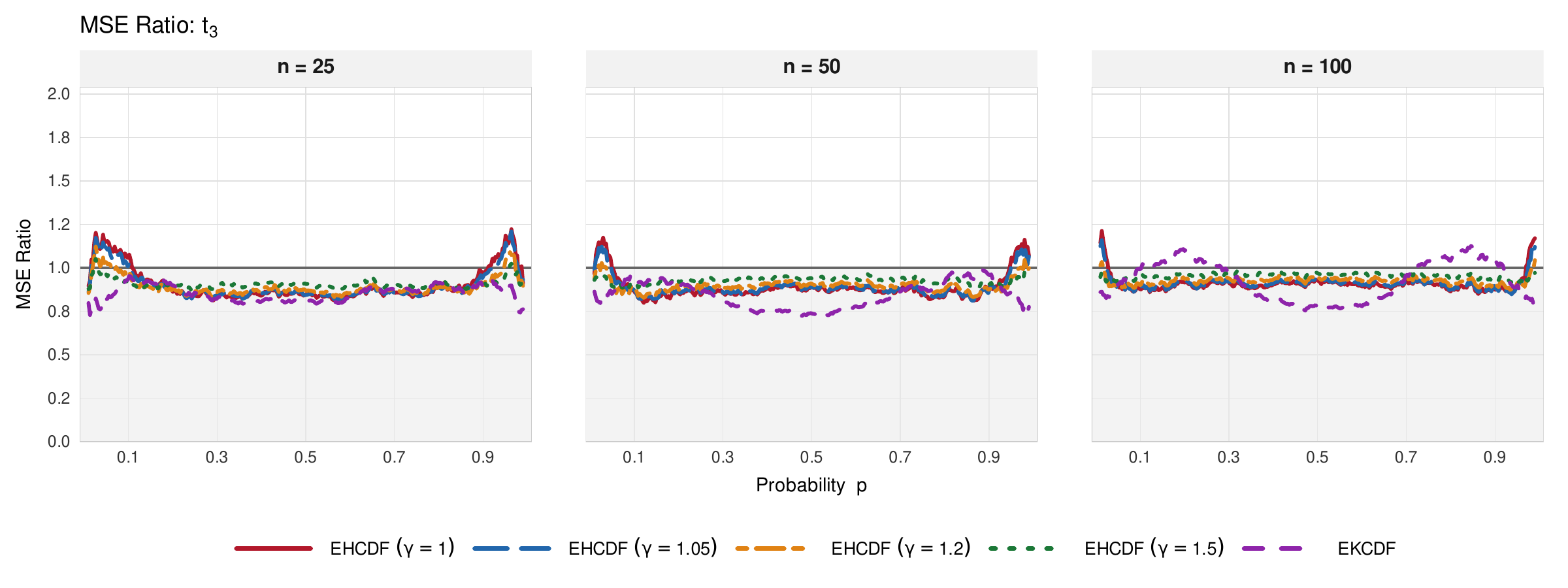}\par\medskip

\caption{
MSE ratio for $\mathcal N(0,1)$ and Student's $t$ distributions with 2 and 3 degrees of freedom (top to bottom). Different lines correspond to different CDF estimators and the dotted horizontal line corresponds to 1.
}
\end{figure}

\begin{figure}[p]
\centering
\includegraphics[width=0.95\textwidth,height=0.28\textheight,keepaspectratio]{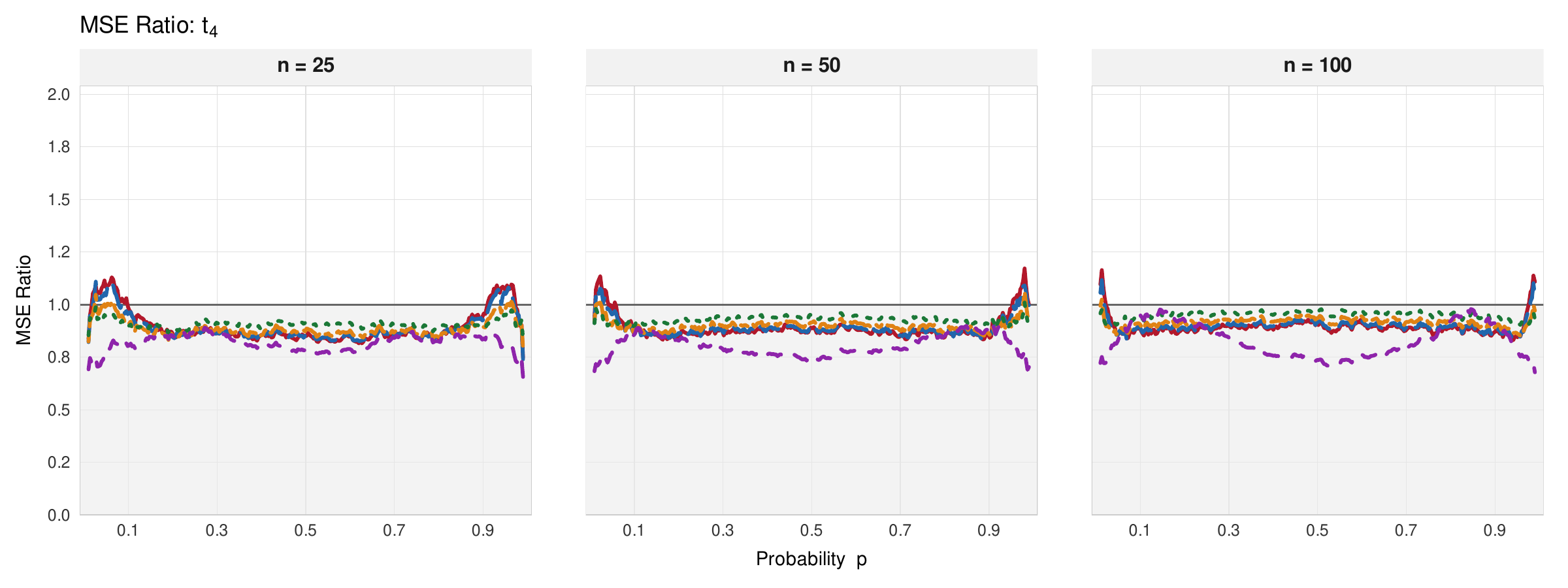}\par\medskip
\includegraphics[width=0.95\textwidth,height=0.28\textheight,keepaspectratio]{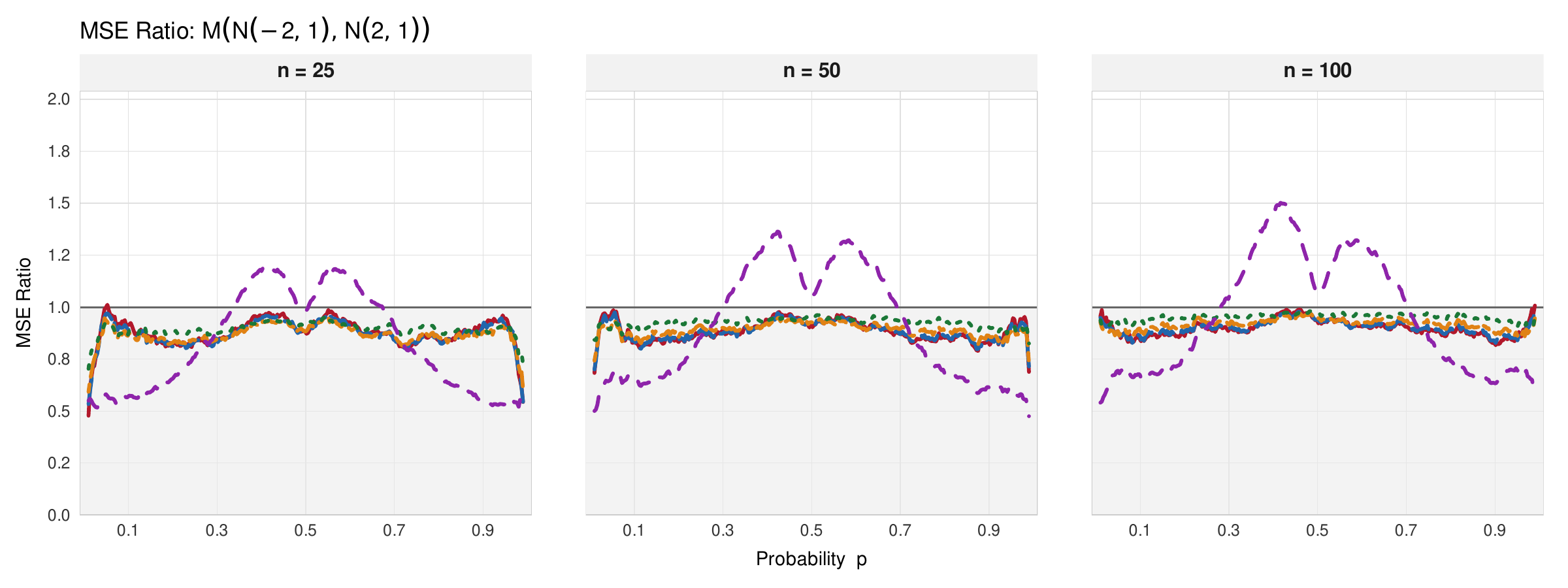}\par\medskip
\includegraphics[width=0.95\textwidth,height=0.28\textheight,keepaspectratio]{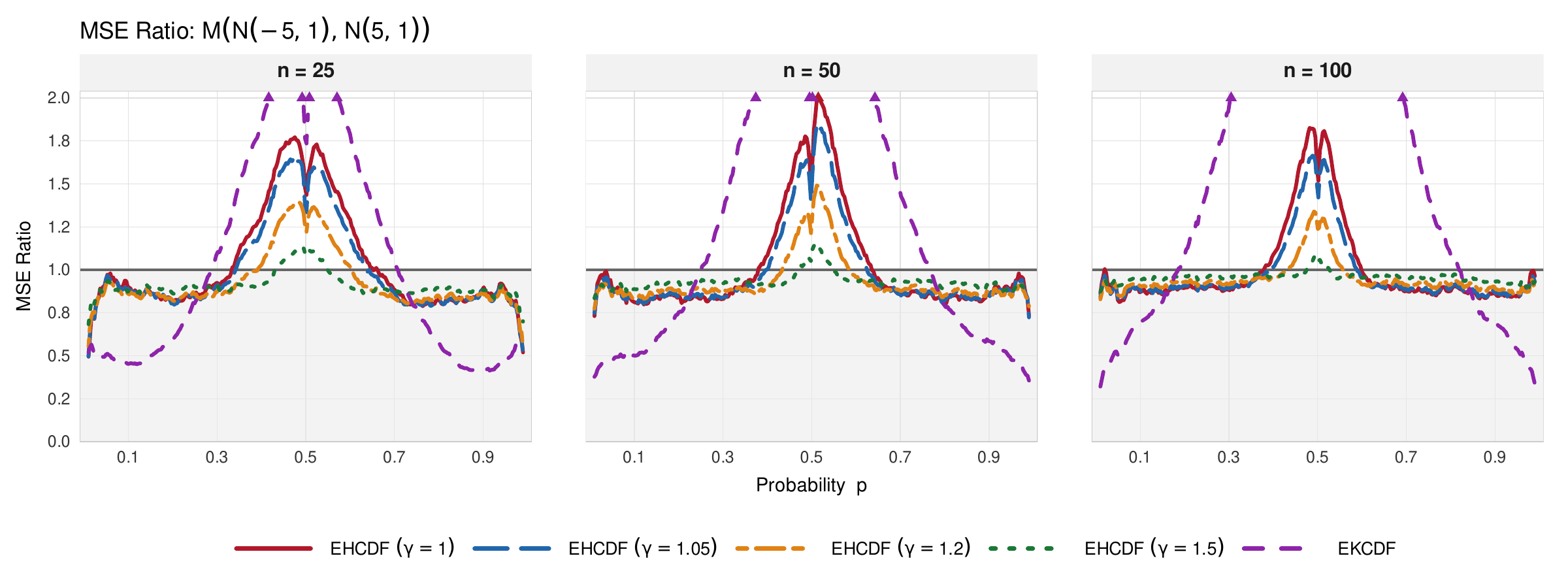}\par\medskip

\caption{
MSE ratio for Student's $t$ (4 d.f.) and bimodal mixtures $\mathcal M(\mathcal N(-2,1),\mathcal N(2,1))$, $\mathcal M(\mathcal N(-5,1),\mathcal N(5,1))$ (top to bottom). Different lines correspond to different CDF estimators and the dotted horizontal line corresponds to 1.
}
\label{finalplot}
\end{figure}

\begin{figure}[p]
\centering
\includegraphics[width=0.95\textwidth,height=0.28\textheight,keepaspectratio]{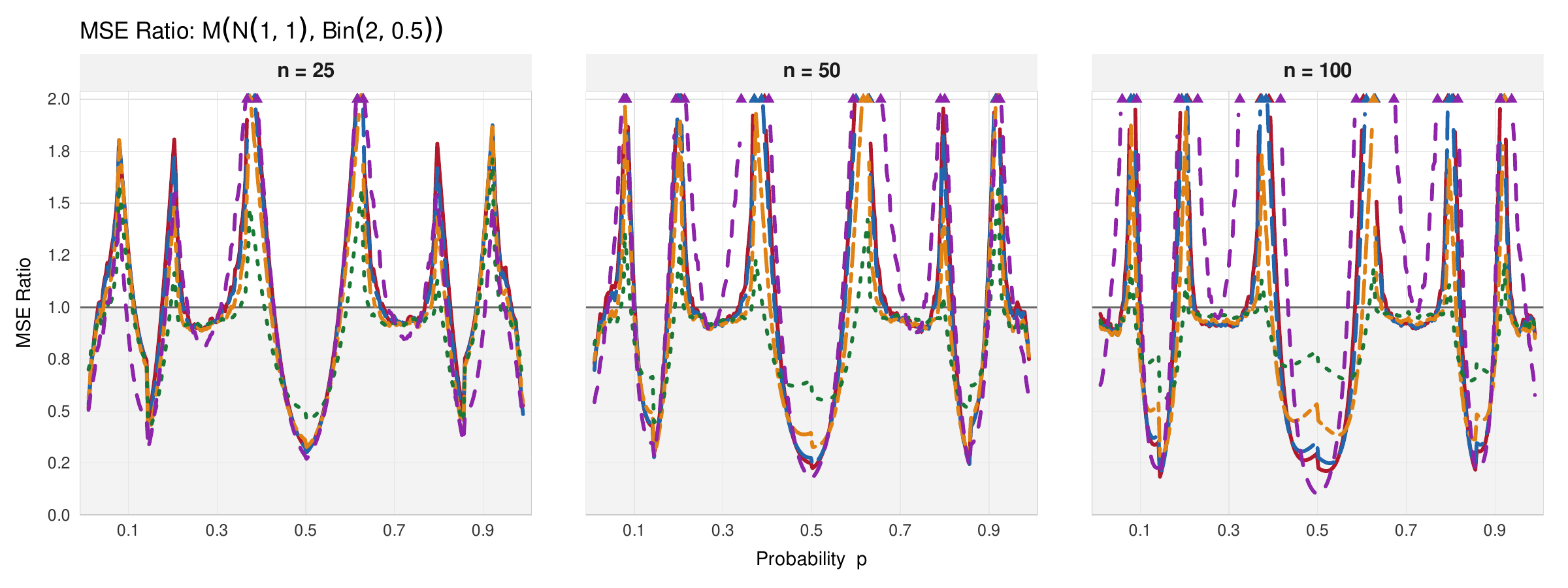}\par\medskip
\includegraphics[width=0.95\textwidth,height=0.28\textheight,keepaspectratio]{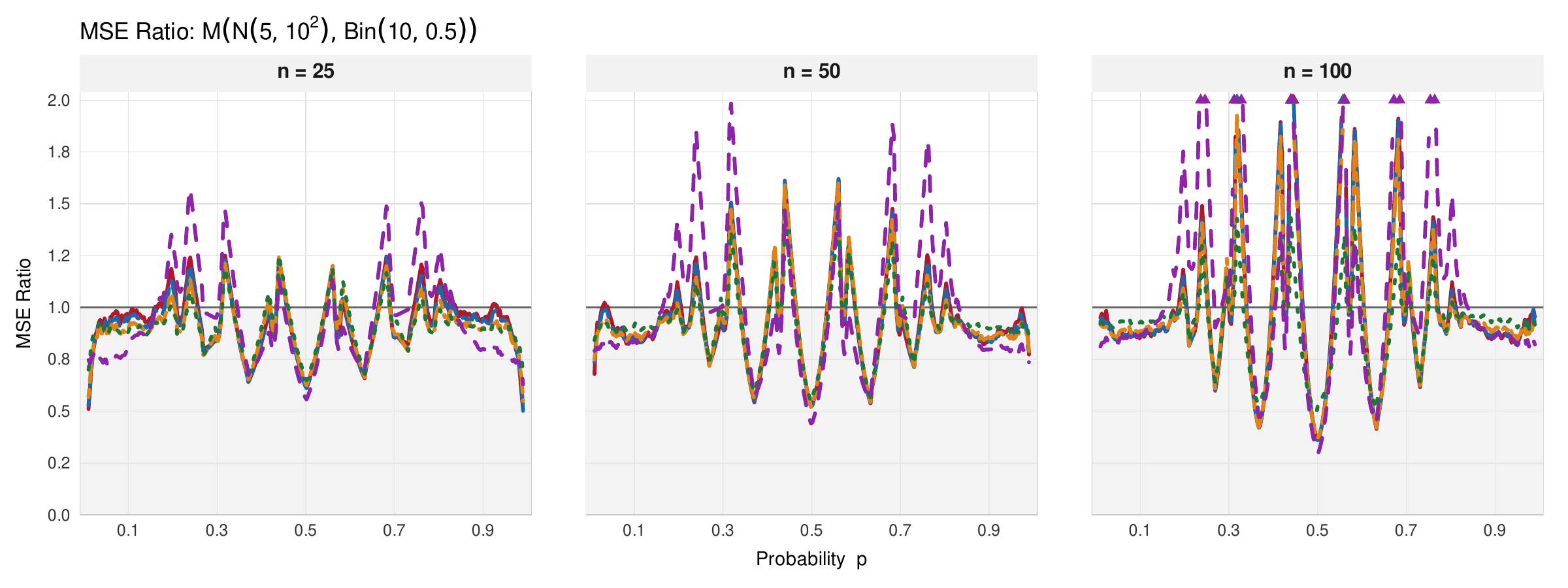}\par\medskip
\includegraphics[width=0.95\textwidth,height=0.28\textheight,keepaspectratio]{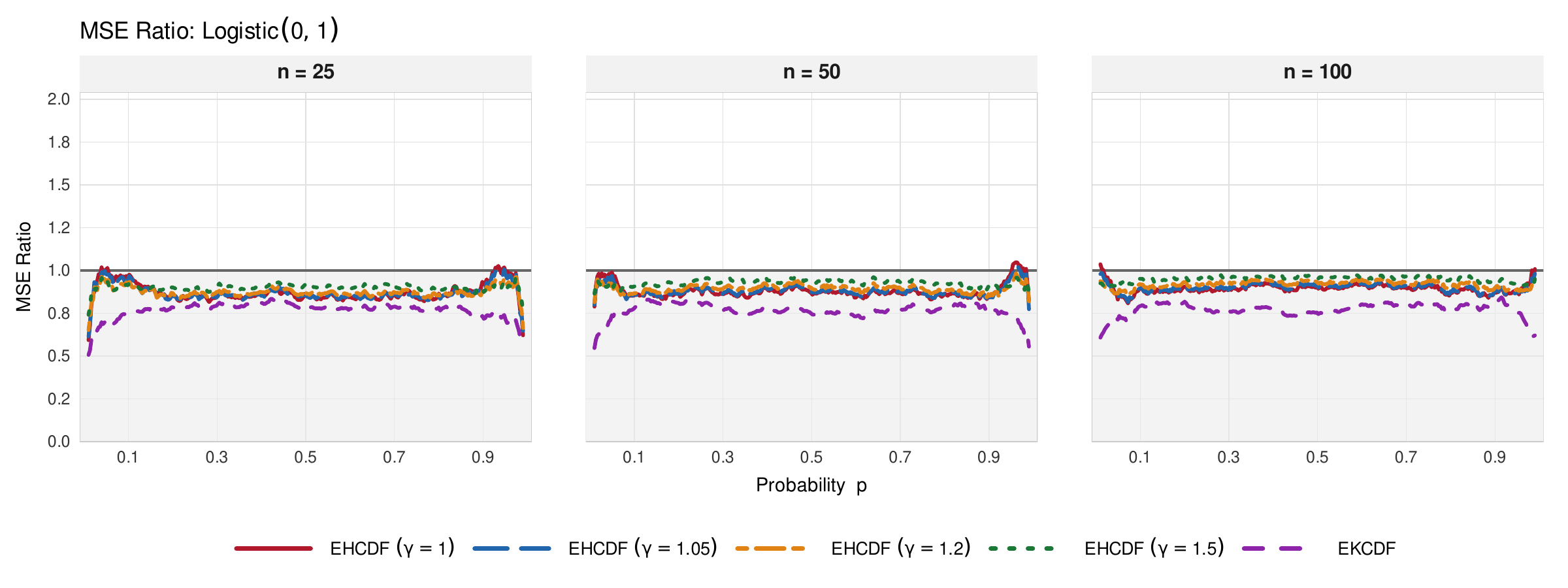}\par\medskip

\caption{
MSE ratio for $\mathcal M(\mathcal N(1,1),\mathrm{Bin}(2,0.5))$, $\mathcal M(\mathcal N(5,10^2),\mathrm{Bin}(10,0.5))$ and $\text{Logistic}(0,1)$ (top to bottom; last plot shows legend). Different lines correspond to different CDF estimators and the dotted horizontal line corresponds to 1.
}
\end{figure}

\FloatBarrier
\section{Additional results for the private EHCDF}\label{appendix:dp}
This section reports the results of the simulation study of Section~\ref{sec:dp_sim} that are not shown in the main text. The data are generated from the zero-inflated beta prime distributions described there, with $\alpha=1$, $\beta\in\{1.5,2,5\}$ and $p_0\in\{0,0.25\}$, and all results are based on $500$ Monte Carlo replications. The public intervals are $[0,b]$ with $b=cP_{99}$, and errors are measured by $\E\,W_1(\cdot,F^c)$. The default resolutions are $m^*=9,14,22$ and $k^*=9$ for $n=200$, and $m^*=25,40,63$ and $k^*=11$ for $n=1000$, for $\varepsilon=0.5,1,2$, respectively.

Figure~\ref{fig:dp_cross_n200} is the analogue of Figure~\ref{fig:dp_cross} for $n=200$. With the smaller sample, the EHCDF at $m^*$ is more accurate than the histogram at $k^*$ in almost all settings, and the histogram at its best resolution is more accurate than the best EHCDF only for the tightest bounds. Figures~\ref{fig:dp_sim_p25_b50_n200}--\ref{fig:dp_sim_p00_b50_n1000} are the analogues of Figure~\ref{fig:dp_sim} for the remaining models and sample sizes, and Tables~\ref{tab:dp_sim_hist_p00_b15}--\ref{tab:dp_sim_p00_b50} are the analogues of Tables~\ref{tab:dp_sim_hist} and~\ref{tab:dp_sim}. The conclusions are the same as in the main text: the discretisation error of the EHCDF is small for every public interval, so that a small number of atoms is sufficient, whereas the histogram requires more bins, and hence more privacy noise, as the public interval widens.

\begin{figure}[H]
    \centering
    \includegraphics[width=\textwidth,height=0.9\textheight,keepaspectratio]{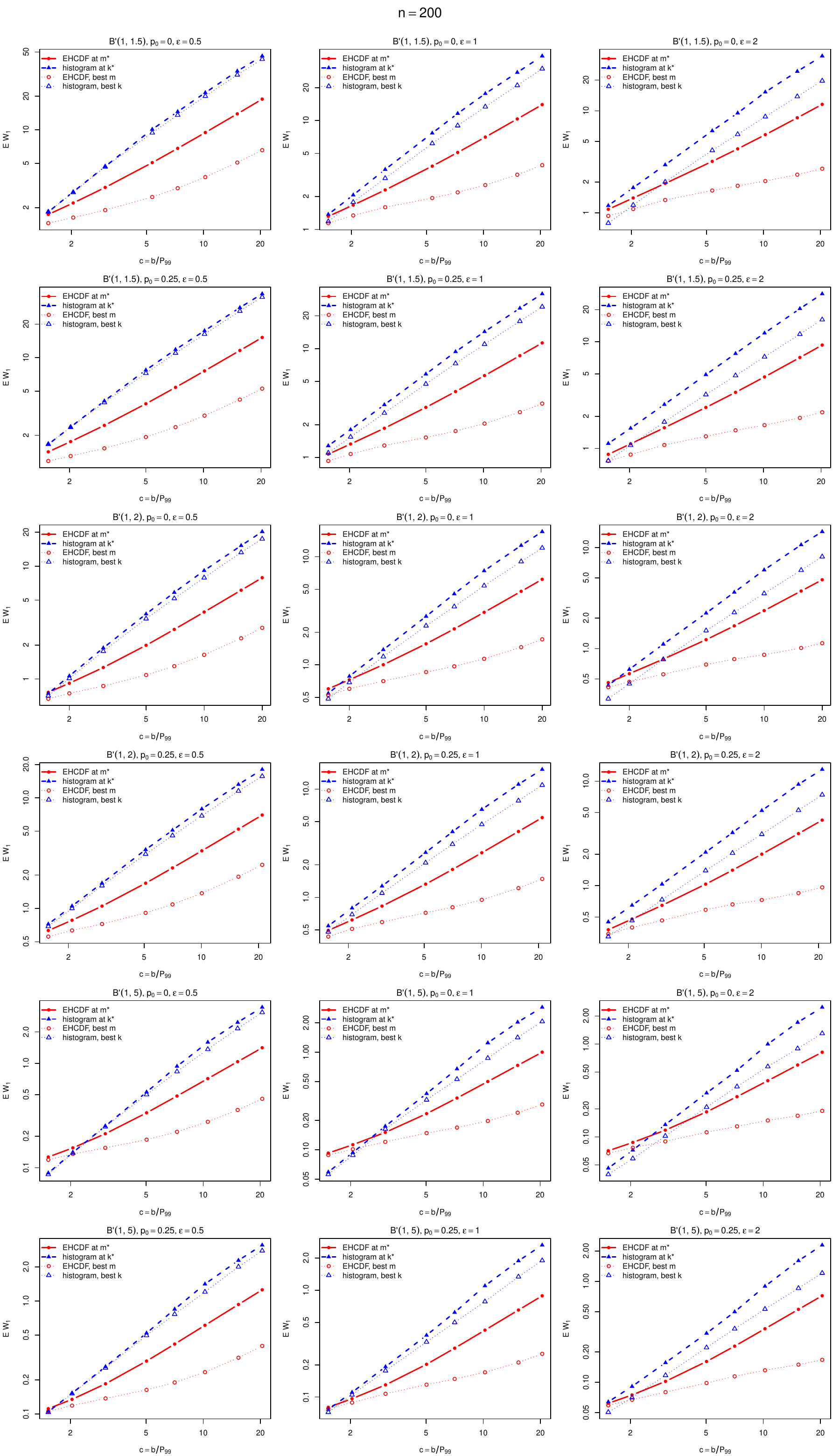}
    \caption{Mean error $\E W_1(\cdot,F^c)$ as a function of $c=b/P_{99}$ on a log-log scale, for the zero-inflated $\mathcal{B}'(1,\beta)$ with $\beta=1.5,2,5$ and $p_0=0,0.25$, $n=200$, and $\varepsilon=0.5,1,2$. Rows correspond to the pairs $(\beta,p_0)$, and columns to $\varepsilon$. Solid and dashed curves show the private EHCDF at $m^*$ and the perturbed histogram at $k^*$, respectively. Dotted curves show the two mechanisms at their best resolutions over the grid. Based on $500$ Monte Carlo replications.}
    \label{fig:dp_cross_n200}
\end{figure}

\begin{figure}[H]
    \centering
    \includegraphics[width=\textwidth]{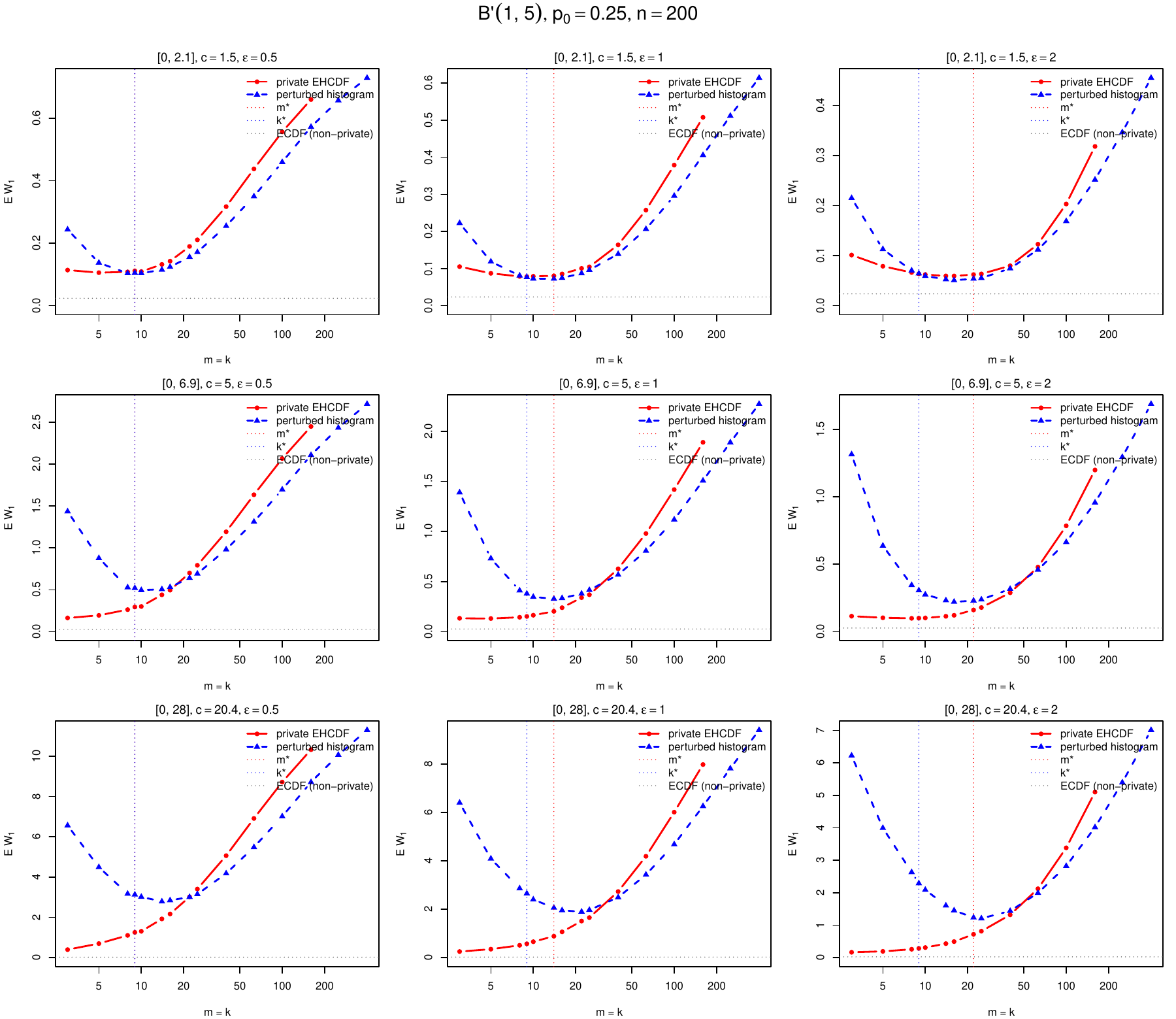}
    \caption{Mean error $\E\,W_1(\cdot,F^c)$ as a function of the common resolution $m=k$ on a logarithmic scale, for the zero-inflated $\mathcal{B}'(1,5)$ with $p_0=0.25$ and $n=200$. Rows correspond to $c=1.5,5,20$, and columns to $\varepsilon=0.5,1,2$. The solid and dashed curves represent the private EHCDF and the perturbed histogram, respectively. Vertical lines indicate $m^*$ and $k^*$, while the horizontal line indicates the error of the non-private ECDF. Based on $500$ Monte Carlo replications.}
    \label{fig:dp_sim_p25_b50_n200}
\end{figure}

\begin{figure}[H]
    \centering
    \includegraphics[width=\textwidth]{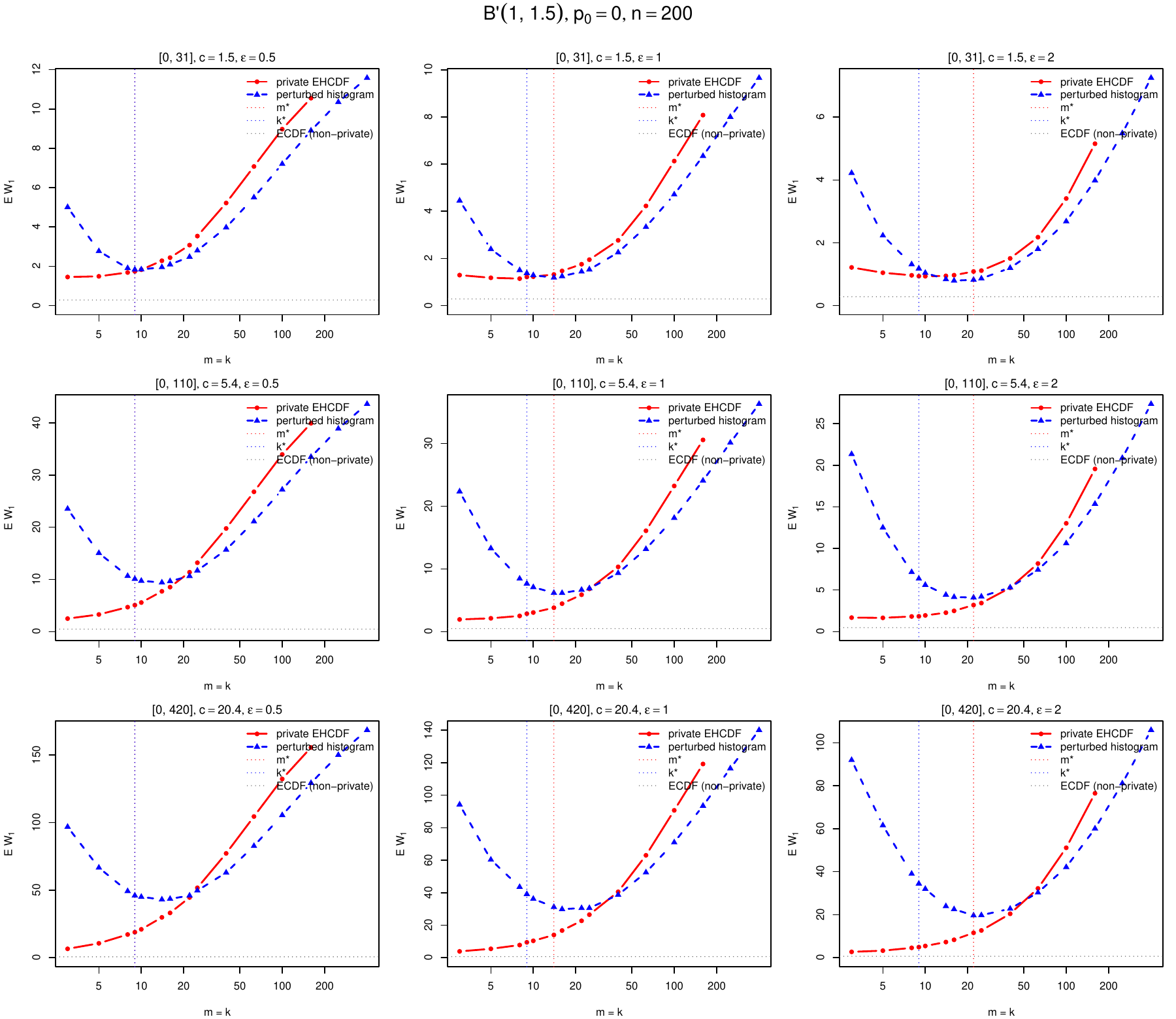}
    \caption{Mean error $\E\,W_1(\cdot,F^c)$ as a function of the common resolution $m=k$ on a logarithmic scale, for $\mathcal{B}'(1,1.5)$ without atom at zero ($p_0=0$) and $n=200$. Rows correspond to $c=1.5,5,20$, and columns to $\varepsilon=0.5,1,2$. The solid and dashed curves represent the private EHCDF and the perturbed histogram, respectively. Vertical lines indicate $m^*$ and $k^*$, while the horizontal line indicates the error of the non-private ECDF. Based on $500$ Monte Carlo replications.}
    \label{fig:dp_sim_p00_b15_n200}
\end{figure}

\begin{figure}[H]
    \centering
    \includegraphics[width=\textwidth]{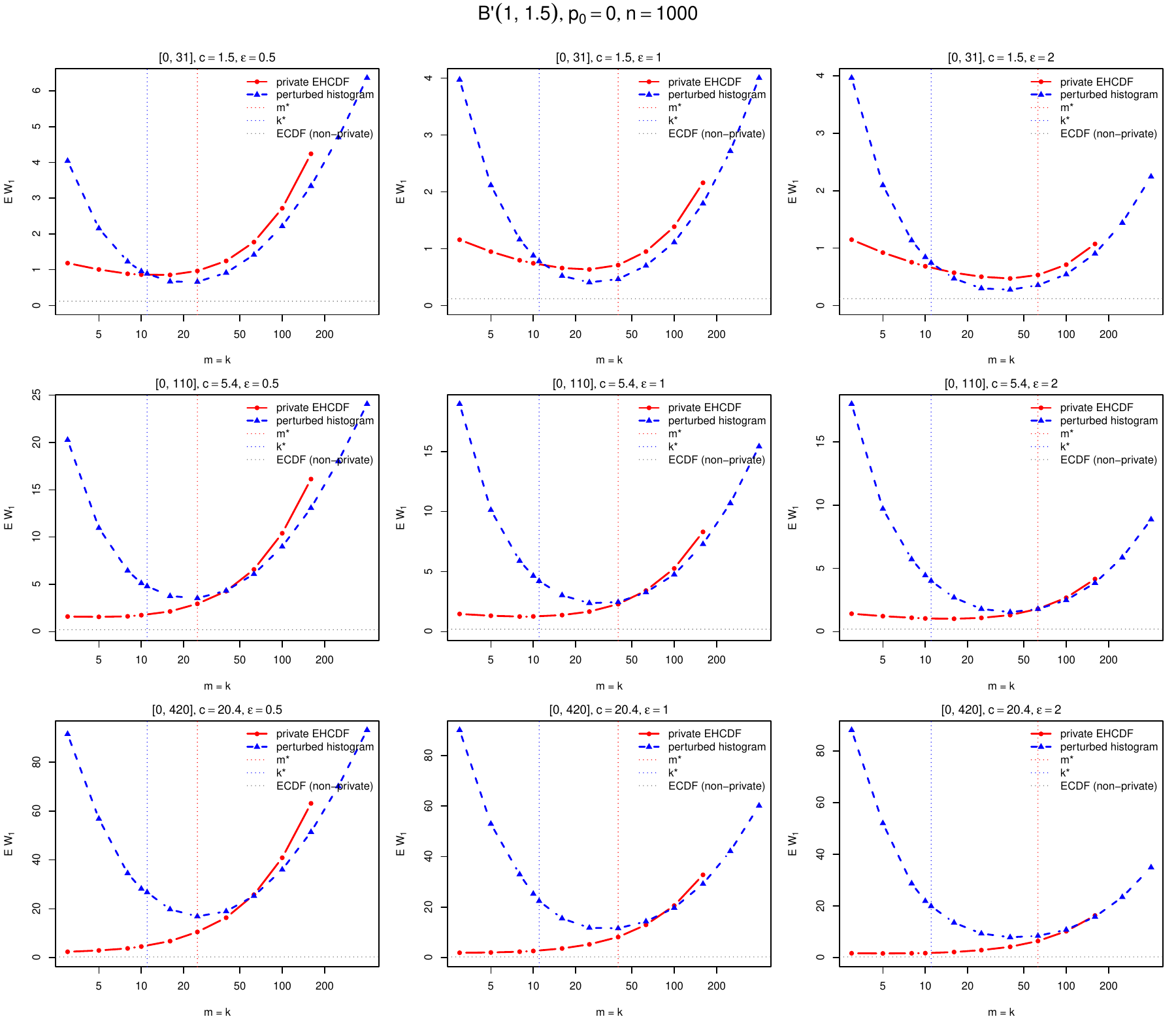}
    \caption{Mean error $\E\,W_1(\cdot,F^c)$ as a function of the common resolution $m=k$ on a logarithmic scale, for $\mathcal{B}'(1,1.5)$ without atom at zero ($p_0=0$) and $n=1000$. Rows correspond to $c=1.5,5,20$, and columns to $\varepsilon=0.5,1,2$. The solid and dashed curves represent the private EHCDF and the perturbed histogram, respectively. Vertical lines indicate $m^*$ and $k^*$, while the horizontal line indicates the error of the non-private ECDF. Based on $500$ Monte Carlo replications.}
    \label{fig:dp_sim_p00_b15_n1000}
\end{figure}

\begin{figure}[H]
    \centering
    \includegraphics[width=\textwidth]{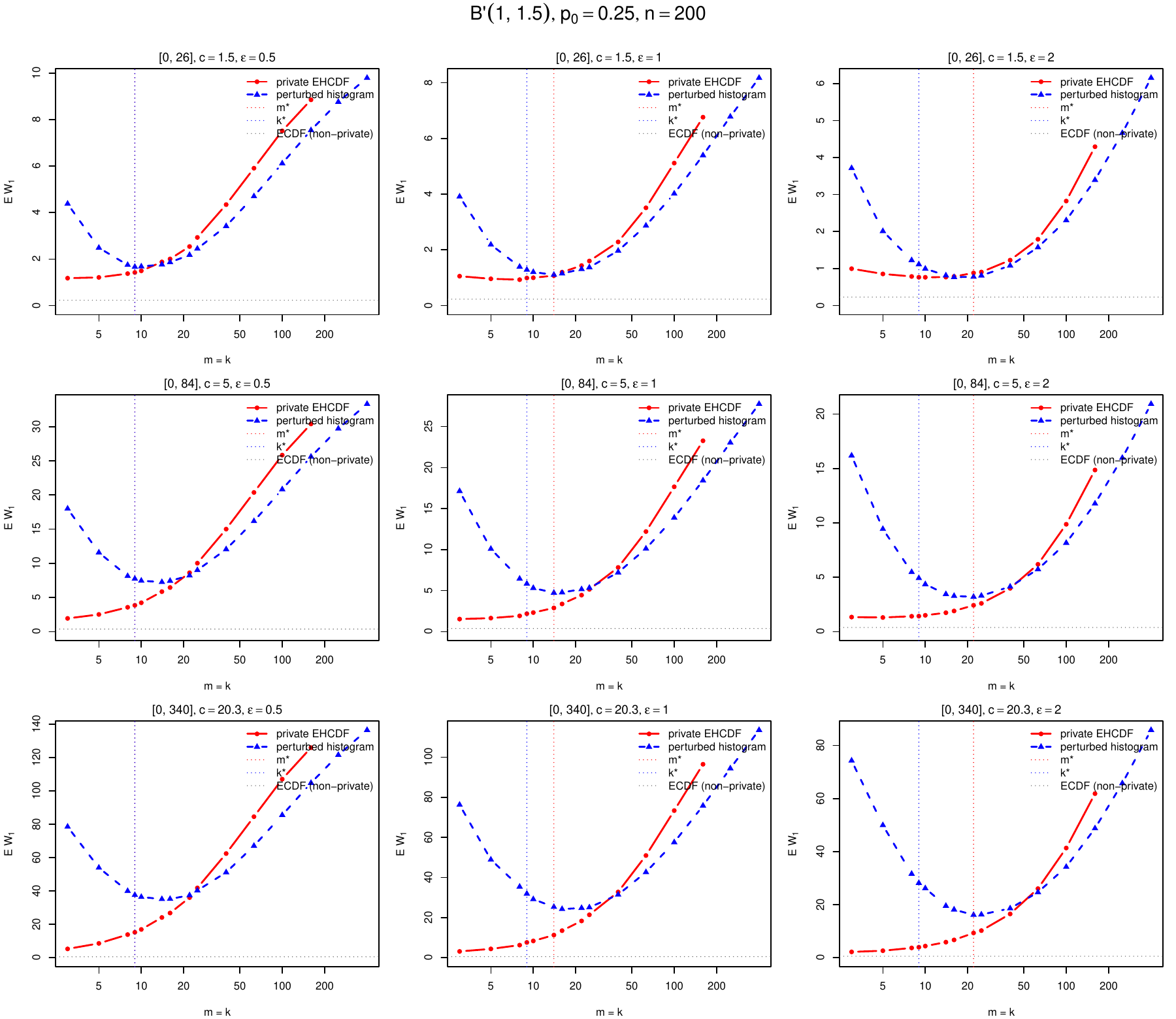}
    \caption{Mean error $\E\,W_1(\cdot,F^c)$ as a function of the common resolution $m=k$ on a logarithmic scale, for the zero-inflated $\mathcal{B}'(1,1.5)$ with $p_0=0.25$ and $n=200$. Rows correspond to $c=1.5,5,20$, and columns to $\varepsilon=0.5,1,2$. The solid and dashed curves represent the private EHCDF and the perturbed histogram, respectively. Vertical lines indicate $m^*$ and $k^*$, while the horizontal line indicates the error of the non-private ECDF. Based on $500$ Monte Carlo replications.}
    \label{fig:dp_sim_p25_b15_n200}
\end{figure}

\begin{figure}[H]
    \centering
    \includegraphics[width=\textwidth]{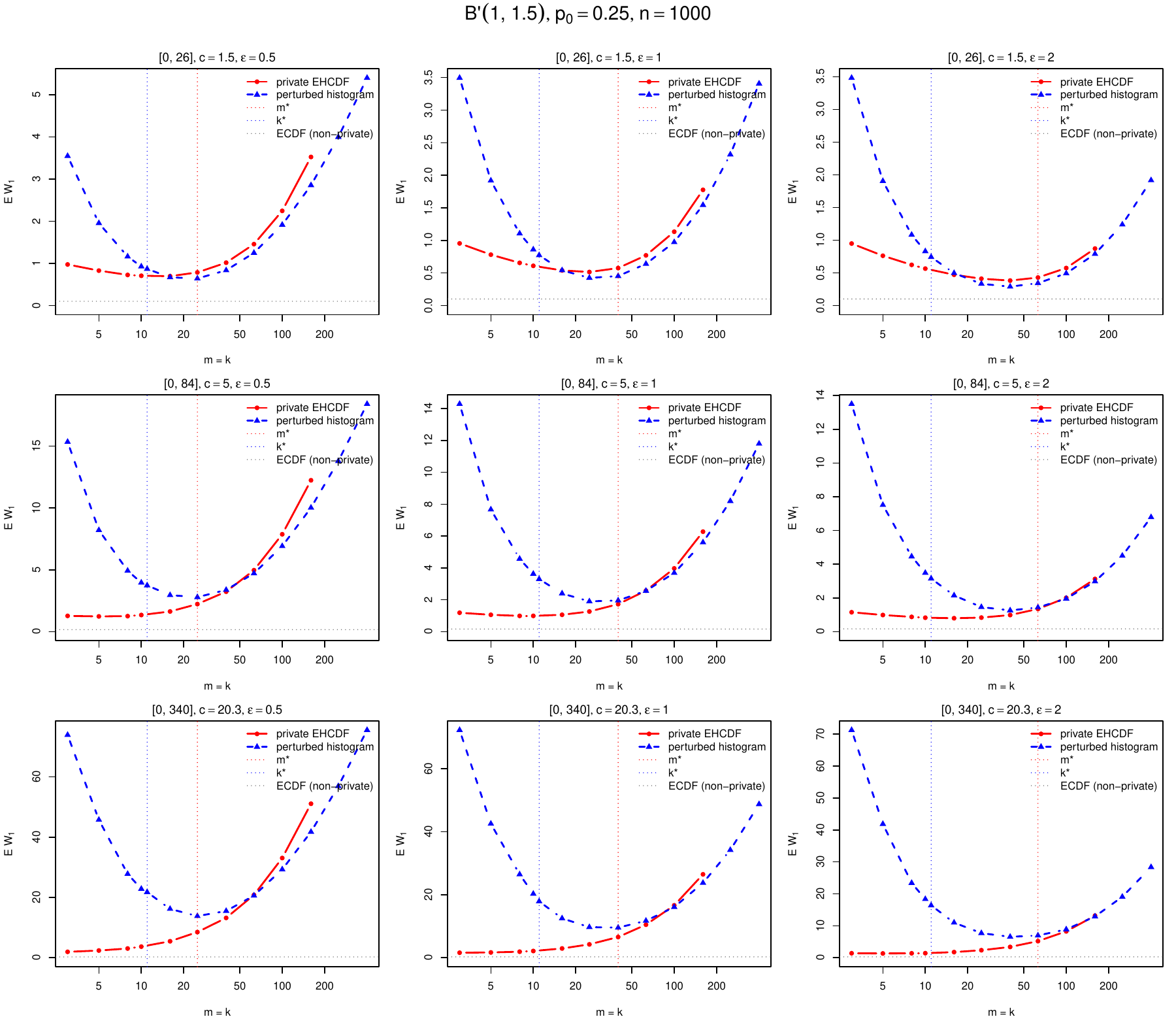}
    \caption{Mean error $\E\,W_1(\cdot,F^c)$ as a function of the common resolution $m=k$ on a logarithmic scale, for the zero-inflated $\mathcal{B}'(1,1.5)$ with $p_0=0.25$ and $n=1000$. Rows correspond to $c=1.5,5,20$, and columns to $\varepsilon=0.5,1,2$. The solid and dashed curves represent the private EHCDF and the perturbed histogram, respectively. Vertical lines indicate $m^*$ and $k^*$, while the horizontal line indicates the error of the non-private ECDF. Based on $500$ Monte Carlo replications.}
    \label{fig:dp_sim_p25_b15_n1000}
\end{figure}

\begin{figure}[H]
    \centering
    \includegraphics[width=\textwidth]{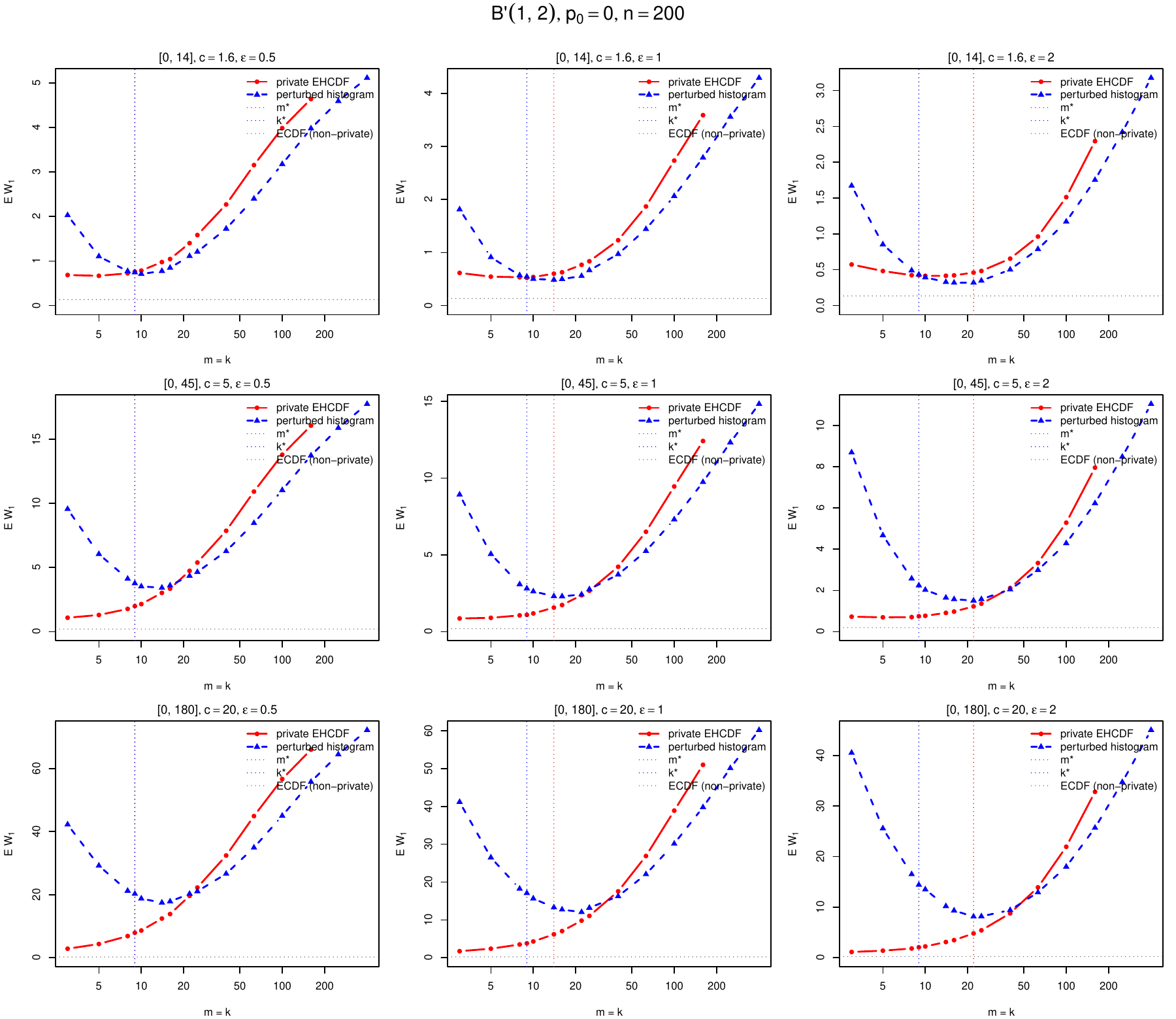}
    \caption{Mean error $\E\,W_1(\cdot,F^c)$ as a function of the common resolution $m=k$ on a logarithmic scale, for $\mathcal{B}'(1,2)$ without atom at zero ($p_0=0$) and $n=200$. Rows correspond to $c=1.5,5,20$, and columns to $\varepsilon=0.5,1,2$. The solid and dashed curves represent the private EHCDF and the perturbed histogram, respectively. Vertical lines indicate $m^*$ and $k^*$, while the horizontal line indicates the error of the non-private ECDF. Based on $500$ Monte Carlo replications.}
    \label{fig:dp_sim_p00_b20_n200}
\end{figure}

\begin{figure}[H]
    \centering
    \includegraphics[width=\textwidth]{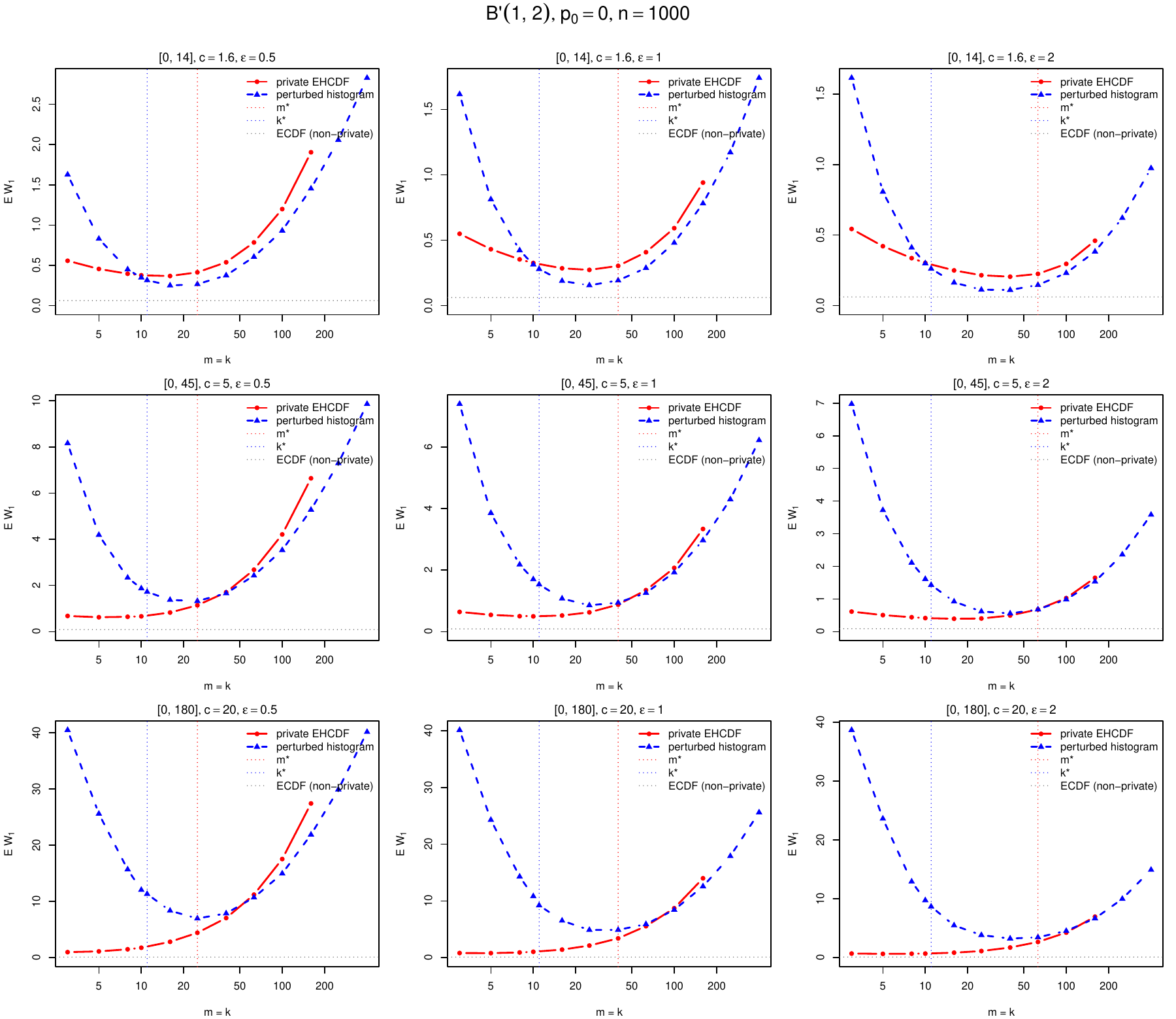}
    \caption{Mean error $\E\,W_1(\cdot,F^c)$ as a function of the common resolution $m=k$ on a logarithmic scale, for $\mathcal{B}'(1,2)$ without atom at zero ($p_0=0$) and $n=1000$. Rows correspond to $c=1.5,5,20$, and columns to $\varepsilon=0.5,1,2$. The solid and dashed curves represent the private EHCDF and the perturbed histogram, respectively. Vertical lines indicate $m^*$ and $k^*$, while the horizontal line indicates the error of the non-private ECDF. Based on $500$ Monte Carlo replications.}
    \label{fig:dp_sim_p00_b20_n1000}
\end{figure}

\begin{figure}[H]
    \centering
    \includegraphics[width=\textwidth]{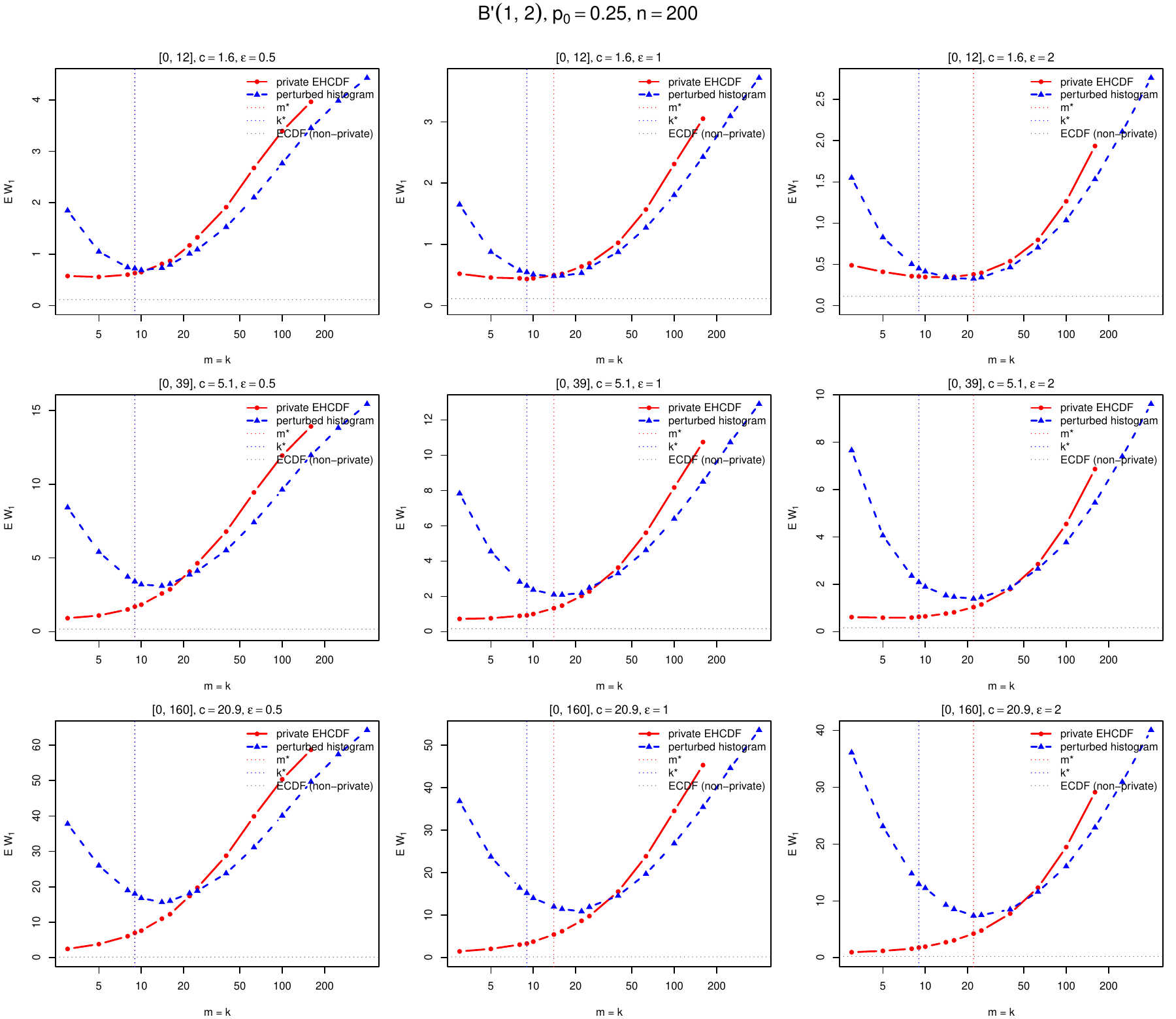}
    \caption{Mean error $\E\,W_1(\cdot,F^c)$ as a function of the common resolution $m=k$ on a logarithmic scale, for the zero-inflated $\mathcal{B}'(1,2)$ with $p_0=0.25$ and $n=200$. Rows correspond to $c=1.5,5,20$, and columns to $\varepsilon=0.5,1,2$. The solid and dashed curves represent the private EHCDF and the perturbed histogram, respectively. Vertical lines indicate $m^*$ and $k^*$, while the horizontal line indicates the error of the non-private ECDF. Based on $500$ Monte Carlo replications.}
    \label{fig:dp_sim_p25_b20_n200}
\end{figure}

\begin{figure}[H]
    \centering
    \includegraphics[width=\textwidth]{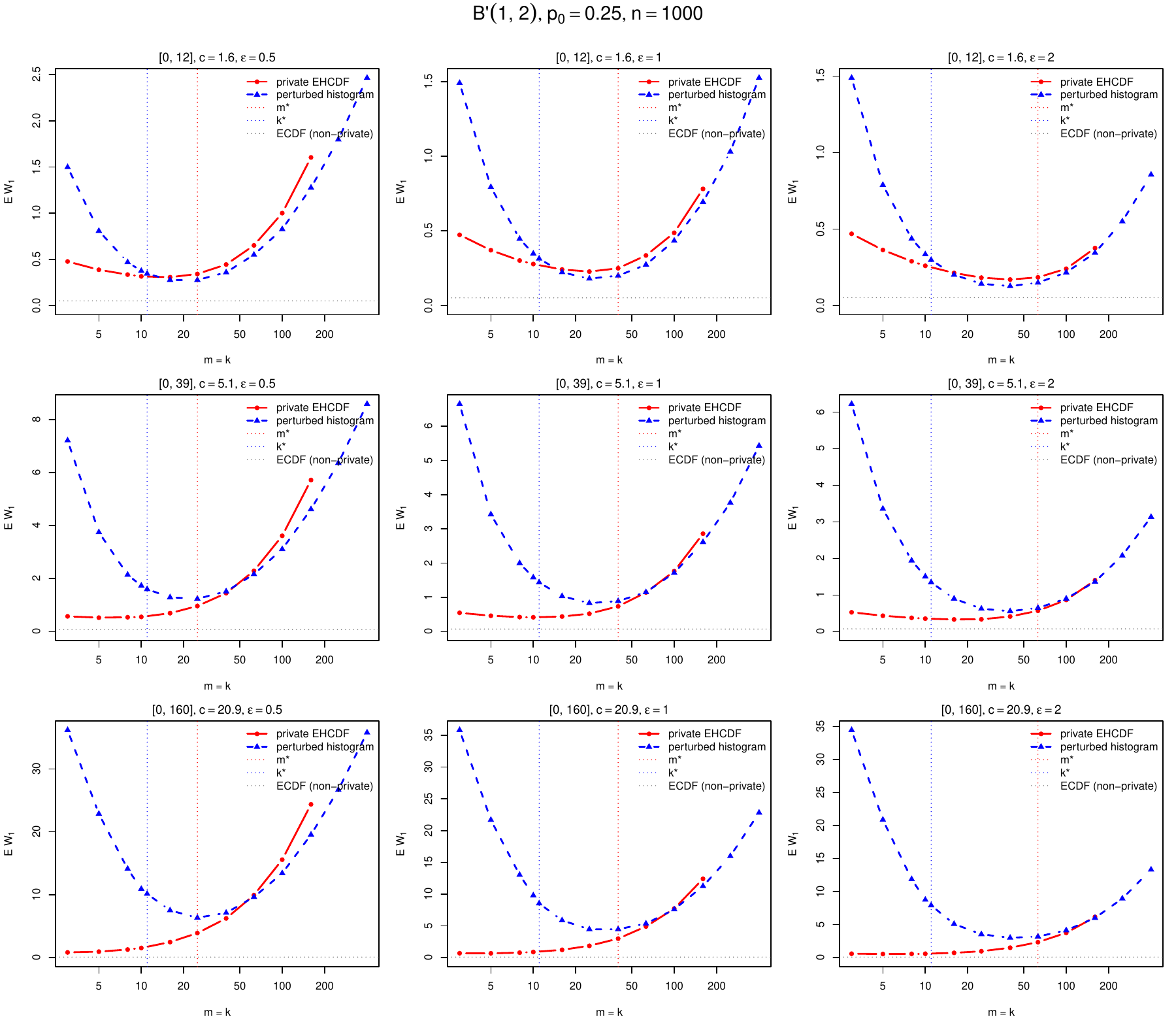}
    \caption{Mean error $\E\,W_1(\cdot,F^c)$ as a function of the common resolution $m=k$ on a logarithmic scale, for the zero-inflated $\mathcal{B}'(1,2)$ with $p_0=0.25$ and $n=1000$. Rows correspond to $c=1.5,5,20$, and columns to $\varepsilon=0.5,1,2$. The solid and dashed curves represent the private EHCDF and the perturbed histogram, respectively. Vertical lines indicate $m^*$ and $k^*$, while the horizontal line indicates the error of the non-private ECDF. Based on $500$ Monte Carlo replications.}
    \label{fig:dp_sim_p25_b20_n1000}
\end{figure}

\begin{figure}[H]
    \centering
    \includegraphics[width=\textwidth]{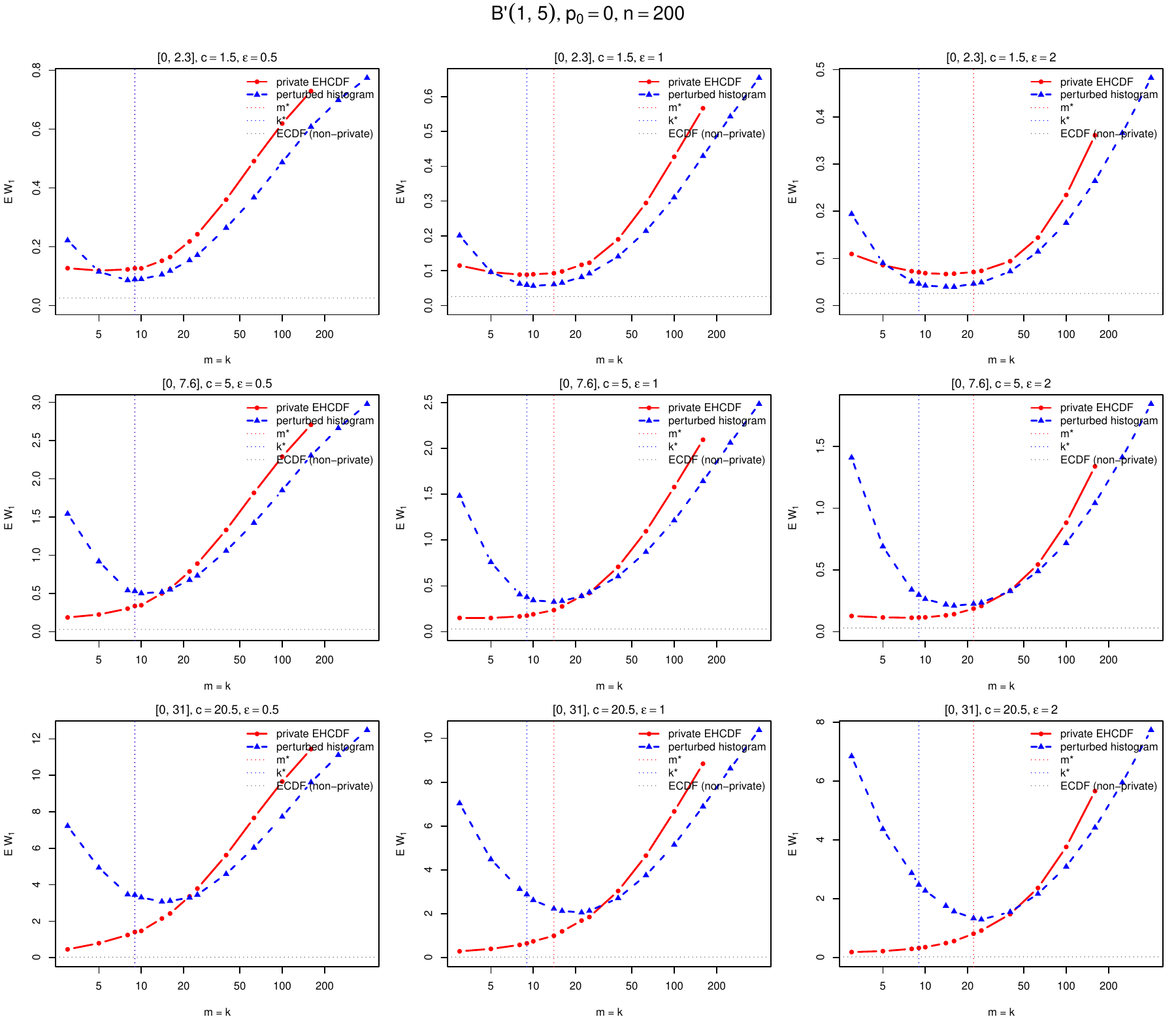}
    \caption{Mean error $\E\,W_1(\cdot,F^c)$ as a function of the common resolution $m=k$ on a logarithmic scale, for $\mathcal{B}'(1,5)$ without atom at zero ($p_0=0$) and $n=200$. Rows correspond to $c=1.5,5,20$, and columns to $\varepsilon=0.5,1,2$. The solid and dashed curves represent the private EHCDF and the perturbed histogram, respectively. Vertical lines indicate $m^*$ and $k^*$, while the horizontal line indicates the error of the non-private ECDF. Based on $500$ Monte Carlo replications.}
    \label{fig:dp_sim_p00_b50_n200}
\end{figure}

\begin{figure}[H]
    \centering
    \includegraphics[width=\textwidth]{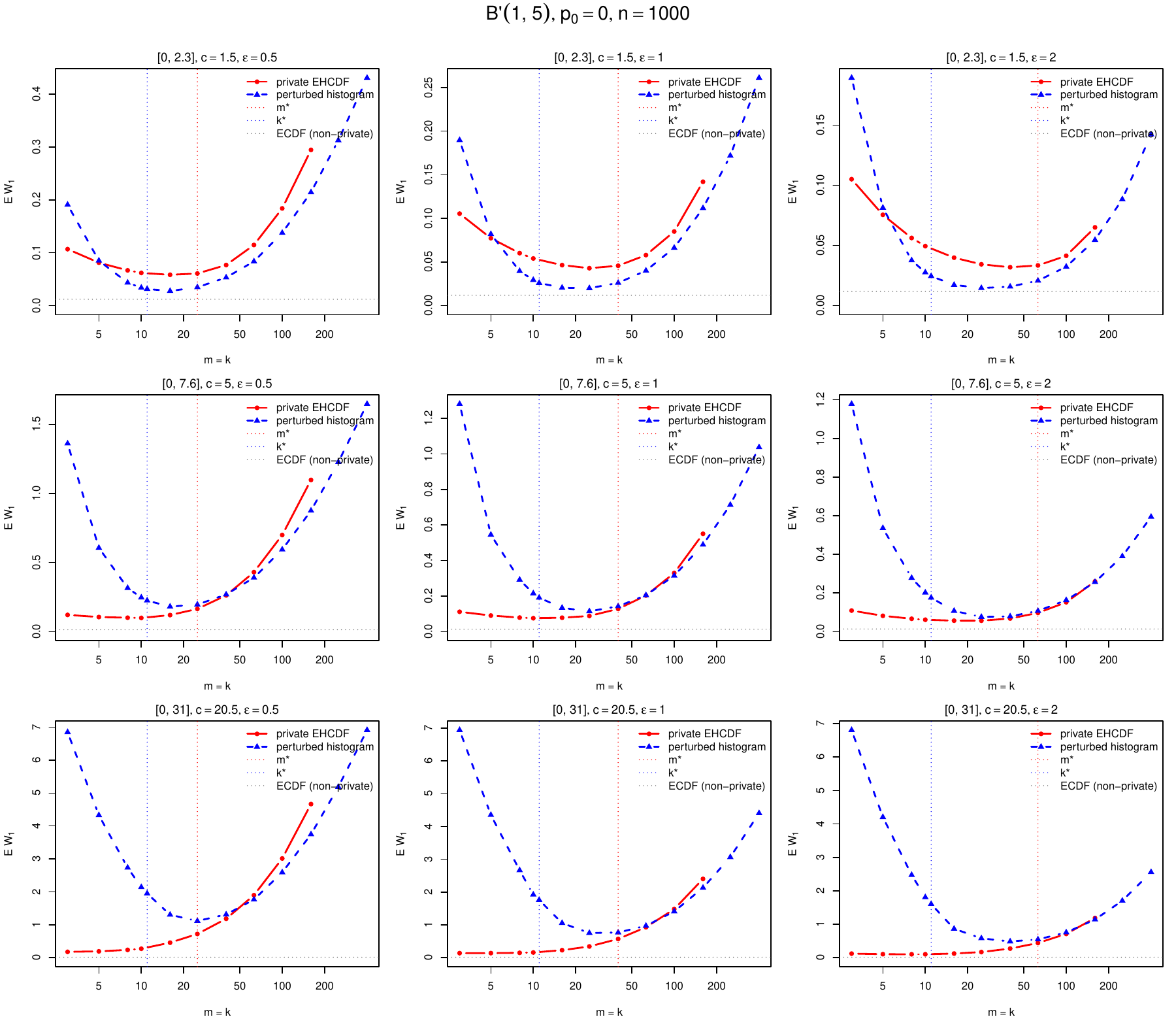}
    \caption{Mean error $\E\,W_1(\cdot,F^c)$ as a function of the common resolution $m=k$ on a logarithmic scale, for $\mathcal{B}'(1,5)$ without atom at zero ($p_0=0$) and $n=1000$. Rows correspond to $c=1.5,5,20$, and columns to $\varepsilon=0.5,1,2$. The solid and dashed curves represent the private EHCDF and the perturbed histogram, respectively. Vertical lines indicate $m^*$ and $k^*$, while the horizontal line indicates the error of the non-private ECDF. Based on $500$ Monte Carlo replications.}
    \label{fig:dp_sim_p00_b50_n1000}
\end{figure}

\FloatBarrier

\begin{table}[H]
    \centering\footnotesize
    \begin{adjustbox}{max width=0.95\textwidth}
        \begin{tabular}{c c c c c c cc cc}
\hline
 & & & & & & \multicolumn{2}{c}{default resolution} & \multicolumn{2}{c}{best resolution} \\
\cline{7-8}\cline{9-10}
$n$ & $[0,b]$ & $P(X>b)$ & $\varepsilon$ & $m^*$ & $k^*$ & EHCDF & Histogram & EHCDF ($m^\dagger$) & Histogram ($k^\dagger$) \\
\hline
200 & $[0,31]$ & 0.006 & 0.5 & 9 & 9 & 1.737 & 1.830 & 1.451 (3) & 1.830 (9) \\
200 & $[0,31]$ & 0.006 & 1.0 & 14 & 9 & 1.318 & 1.373 & 1.141 (8) & 1.183 (14) \\
200 & $[0,31]$ & 0.006 & 2.0 & 22 & 9 & 1.084 & 1.178 & 0.934 (10) & 0.793 (16) \\
200 & $[0,62]$ & 0.002 & 0.5 & 9 & 9 & 3.036 & 4.689 & 1.900 (3) & 4.650 (10) \\
200 & $[0,62]$ & 0.002 & 1.0 & 14 & 9 & 2.303 & 3.552 & 1.597 (5) & 2.950 (14) \\
200 & $[0,62]$ & 0.002 & 2.0 & 22 & 9 & 1.942 & 2.955 & 1.336 (9) & 2.008 (22) \\
200 & $[0,110]$ & 0.001 & 0.5 & 9 & 9 & 5.074 & 10.097 & 2.492 (3) & 9.403 (14) \\
200 & $[0,110]$ & 0.001 & 1.0 & 14 & 9 & 3.811 & 7.653 & 1.941 (3) & 6.155 (16) \\
200 & $[0,110]$ & 0.001 & 2.0 & 22 & 9 & 3.195 & 6.370 & 1.655 (5) & 4.094 (22) \\
200 & $[0,210]$ & 0.000 & 0.5 & 9 & 9 & 9.455 & 21.548 & 3.763 (3) & 20.045 (14) \\
200 & $[0,210]$ & 0.000 & 1.0 & 14 & 9 & 7.040 & 17.652 & 2.553 (3) & 13.374 (16) \\
200 & $[0,210]$ & 0.000 & 2.0 & 22 & 9 & 5.835 & 15.262 & 2.050 (3) & 8.755 (22) \\
200 & $[0,420]$ & 0.000 & 0.5 & 9 & 9 & 18.849 & 45.985 & 6.561 (3) & 43.066 (14) \\
200 & $[0,420]$ & 0.000 & 1.0 & 14 & 9 & 13.996 & 39.104 & 3.888 (3) & 29.861 (16) \\
200 & $[0,420]$ & 0.000 & 2.0 & 22 & 9 & 11.570 & 34.443 & 2.711 (3) & 19.663 (22) \\
\hline
1000 & $[0,31]$ & 0.006 & 0.5 & 25 & 11 & 0.967 & 0.890 & 0.856 (16) & 0.662 (25) \\
1000 & $[0,31]$ & 0.006 & 1.0 & 40 & 11 & 0.711 & 0.779 & 0.635 (25) & 0.408 (25) \\
1000 & $[0,31]$ & 0.006 & 2.0 & 63 & 11 & 0.532 & 0.744 & 0.471 (40) & 0.275 (40) \\
1000 & $[0,62]$ & 0.002 & 0.5 & 25 & 11 & 1.774 & 2.305 & 1.215 (8) & 1.709 (25) \\
1000 & $[0,62]$ & 0.002 & 1.0 & 40 & 11 & 1.352 & 2.022 & 0.973 (16) & 1.094 (25) \\
1000 & $[0,62]$ & 0.002 & 2.0 & 63 & 11 & 1.042 & 1.931 & 0.763 (25) & 0.708 (40) \\
1000 & $[0,110]$ & 0.001 & 0.5 & 25 & 11 & 2.939 & 4.799 & 1.559 (5) & 3.539 (25) \\
1000 & $[0,110]$ & 0.001 & 1.0 & 40 & 11 & 2.289 & 4.210 & 1.252 (8) & 2.373 (25) \\
1000 & $[0,110]$ & 0.001 & 2.0 & 63 & 11 & 1.793 & 4.007 & 1.015 (16) & 1.540 (40) \\
1000 & $[0,210]$ & 0.000 & 0.5 & 25 & 11 & 5.345 & 10.859 & 1.891 (3) & 7.616 (25) \\
1000 & $[0,210]$ & 0.000 & 1.0 & 40 & 11 & 4.166 & 9.216 & 1.591 (5) & 5.304 (25) \\
1000 & $[0,210]$ & 0.000 & 2.0 & 63 & 11 & 3.310 & 8.623 & 1.298 (10) & 3.499 (40) \\
1000 & $[0,420]$ & 0.000 & 0.5 & 25 & 11 & 10.537 & 26.762 & 2.417 (3) & 16.909 (25) \\
1000 & $[0,420]$ & 0.000 & 1.0 & 40 & 11 & 8.162 & 22.414 & 1.955 (3) & 11.620 (40) \\
1000 & $[0,420]$ & 0.000 & 2.0 & 63 & 11 & 6.445 & 19.917 & 1.647 (5) & 7.910 (40) \\
\hline
\end{tabular}

    \end{adjustbox}
    \caption{Mean error $\E\,W_1(\cdot,F^c)$ for $\mathcal{B}'(1,1.5)$ without atom at zero ($p_0=0$) and $n=200,1000$, on the public interval $[0,b]$, where $b=cP_{99}$ and $c\in\{1.5,3,5,10,20\}$. The table reports the private EHCDF at $m^*$ and the perturbed histogram at $k^*$, as well as both mechanisms at their best grid resolutions, given in parentheses. The quantity $P(X>b)$ is the clipped tail mass. Based on $500$ Monte Carlo replications.}
    \label{tab:dp_sim_hist_p00_b15}
\end{table}

\begin{table}[H]
    \centering\footnotesize
    \begin{adjustbox}{max width=0.95\textwidth}
        \begin{tabular}{c c c c r cc ccc}
\hline
$n$ & $[0,b]$ & $P(X>b)$ & $\varepsilon$ & $m^*$ & $\E W_1(\mathbb F_n,F^c)$ & $\E W_1(\mathbb F_{n,m},F^c)$ & $\E W_1(\widetilde{\mathbb F}_{n,m}^{\,\varepsilon},\mathbb F_{n,m})$ & $mB/(n\varepsilon)$ & $\E W_1(\widetilde{\mathbb F}_{n,m}^{\,\varepsilon},F^c)$ \\
\hline
200 & $[0,31]$ & 0.006 & 0.5 & 9 & 0.281 & 0.709 & 1.300 & 2.790 & 1.737 \\
200 & $[0,31]$ & 0.006 & 1.0 & 14 & 0.281 & 0.577 & 0.948 & 2.170 & 1.318 \\
200 & $[0,31]$ & 0.006 & 2.0 & 22 & 0.281 & 0.472 & 0.767 & 1.705 & 1.084 \\
200 & $[0,210]$ & 0.000 & 0.5 & 9 & 0.575 & 1.047 & 9.135 & 18.900 & 9.455 \\
200 & $[0,210]$ & 0.000 & 1.0 & 14 & 0.575 & 0.912 & 6.724 & 14.700 & 7.040 \\
200 & $[0,210]$ & 0.000 & 2.0 & 22 & 0.575 & 0.804 & 5.488 & 11.550 & 5.835 \\
\hline
1000 & $[0,31]$ & 0.006 & 0.5 & 25 & 0.120 & 0.382 & 0.691 & 1.550 & 0.967 \\
1000 & $[0,31]$ & 0.006 & 1.0 & 40 & 0.120 & 0.286 & 0.508 & 1.240 & 0.711 \\
1000 & $[0,31]$ & 0.006 & 2.0 & 63 & 0.120 & 0.217 & 0.380 & 0.977 & 0.532 \\
1000 & $[0,210]$ & 0.000 & 0.5 & 25 & 0.260 & 0.667 & 4.971 & 10.500 & 5.345 \\
1000 & $[0,210]$ & 0.000 & 1.0 & 40 & 0.260 & 0.558 & 3.794 & 8.400 & 4.166 \\
1000 & $[0,210]$ & 0.000 & 2.0 & 63 & 0.260 & 0.473 & 2.969 & 6.615 & 3.310 \\
\hline
\end{tabular}

    \end{adjustbox}
    \caption{Error decomposition for the private EHCDF with $m=m^*(\varepsilon)$ and $\mathcal{B}'(1,1.5)$ without atom at zero ($p_0=0$), for $n=200,1000$ and $c=1.5,10$. The table reports the estimation errors of the ECDF and the noiseless EHCDF, the privacy error, its bound $mB/(n\varepsilon)$ from Theorem~\ref{Theo: privacy}, and the total error.}
    \label{tab:dp_sim_p00_b15}
\end{table}

\begin{table}[H]
    \centering\footnotesize
    \begin{adjustbox}{max width=0.95\textwidth}
        \begin{tabular}{c c c c c c cc cc}
\hline
 & & & & & & \multicolumn{2}{c}{default resolution} & \multicolumn{2}{c}{best resolution} \\
\cline{7-8}\cline{9-10}
$n$ & $[0,b]$ & $P(X>b)$ & $\varepsilon$ & $m^*$ & $k^*$ & EHCDF & Histogram & EHCDF ($m^\dagger$) & Histogram ($k^\dagger$) \\
\hline
200 & $[0,26]$ & 0.005 & 0.5 & 9 & 9 & 1.423 & 1.663 & 1.177 (3) & 1.663 (9) \\
200 & $[0,26]$ & 0.005 & 1.0 & 14 & 9 & 1.074 & 1.282 & 0.929 (8) & 1.103 (14) \\
200 & $[0,26]$ & 0.005 & 2.0 & 22 & 9 & 0.880 & 1.112 & 0.763 (10) & 0.767 (16) \\
200 & $[0,51]$ & 0.002 & 0.5 & 9 & 9 & 2.461 & 4.082 & 1.530 (3) & 3.956 (10) \\
200 & $[0,51]$ & 0.002 & 1.0 & 14 & 9 & 1.856 & 3.052 & 1.289 (5) & 2.560 (14) \\
200 & $[0,51]$ & 0.002 & 2.0 & 22 & 9 & 1.563 & 2.582 & 1.077 (9) & 1.772 (22) \\
200 & $[0,84]$ & 0.001 & 0.5 & 9 & 9 & 3.850 & 7.737 & 1.932 (3) & 7.248 (14) \\
200 & $[0,84]$ & 0.001 & 1.0 & 14 & 9 & 2.886 & 5.836 & 1.528 (3) & 4.732 (14) \\
200 & $[0,84]$ & 0.001 & 2.0 & 22 & 9 & 2.421 & 4.913 & 1.299 (5) & 3.198 (22) \\
200 & $[0,170]$ & 0.000 & 0.5 & 9 & 9 & 7.606 & 17.472 & 3.011 (3) & 16.302 (14) \\
200 & $[0,170]$ & 0.000 & 1.0 & 14 & 9 & 5.652 & 14.310 & 2.049 (3) & 10.984 (16) \\
200 & $[0,170]$ & 0.000 & 2.0 & 22 & 9 & 4.681 & 12.058 & 1.651 (3) & 7.218 (22) \\
200 & $[0,340]$ & 0.000 & 0.5 & 9 & 9 & 15.211 & 37.563 & 5.271 (3) & 35.112 (14) \\
200 & $[0,340]$ & 0.000 & 1.0 & 14 & 9 & 11.284 & 31.886 & 3.128 (3) & 24.237 (16) \\
200 & $[0,340]$ & 0.000 & 2.0 & 22 & 9 & 9.317 & 28.150 & 2.184 (3) & 16.119 (22) \\
\hline
1000 & $[0,26]$ & 0.005 & 0.5 & 25 & 11 & 0.785 & 0.868 & 0.696 (16) & 0.641 (25) \\
1000 & $[0,26]$ & 0.005 & 1.0 & 40 & 11 & 0.575 & 0.772 & 0.515 (25) & 0.425 (25) \\
1000 & $[0,26]$ & 0.005 & 2.0 & 63 & 11 & 0.428 & 0.742 & 0.382 (40) & 0.289 (40) \\
1000 & $[0,51]$ & 0.002 & 0.5 & 25 & 11 & 1.427 & 2.029 & 0.981 (8) & 1.515 (25) \\
1000 & $[0,51]$ & 0.002 & 1.0 & 40 & 11 & 1.082 & 1.797 & 0.785 (16) & 1.009 (25) \\
1000 & $[0,51]$ & 0.002 & 2.0 & 63 & 11 & 0.831 & 1.721 & 0.613 (25) & 0.670 (40) \\
1000 & $[0,84]$ & 0.001 & 0.5 & 25 & 11 & 2.227 & 3.738 & 1.230 (5) & 2.786 (25) \\
1000 & $[0,84]$ & 0.001 & 1.0 & 40 & 11 & 1.723 & 3.299 & 0.987 (8) & 1.899 (25) \\
1000 & $[0,84]$ & 0.001 & 2.0 & 63 & 11 & 1.343 & 3.146 & 0.794 (16) & 1.258 (40) \\
1000 & $[0,170]$ & 0.000 & 0.5 & 25 & 11 & 4.291 & 8.903 & 1.526 (3) & 6.294 (25) \\
1000 & $[0,170]$ & 0.000 & 1.0 & 40 & 11 & 3.334 & 7.647 & 1.287 (5) & 4.421 (25) \\
1000 & $[0,170]$ & 0.000 & 2.0 & 63 & 11 & 2.643 & 7.100 & 1.043 (10) & 2.955 (40) \\
1000 & $[0,340]$ & 0.000 & 0.5 & 25 & 11 & 8.489 & 21.755 & 1.942 (3) & 13.800 (25) \\
1000 & $[0,340]$ & 0.000 & 1.0 & 40 & 11 & 6.564 & 17.897 & 1.580 (3) & 9.543 (40) \\
1000 & $[0,340]$ & 0.000 & 2.0 & 63 & 11 & 5.176 & 16.382 & 1.330 (5) & 6.533 (40) \\
\hline
\end{tabular}

    \end{adjustbox}
    \caption{Mean error $\E\,W_1(\cdot,F^c)$ for the zero-inflated $\mathcal{B}'(1,1.5)$ with $p_0=0.25$ and $n=200,1000$, on the public interval $[0,b]$, where $b=cP_{99}$ and $c\in\{1.5,3,5,10,20\}$. The table reports the private EHCDF at $m^*$ and the perturbed histogram at $k^*$, as well as both mechanisms at their best grid resolutions, given in parentheses. The quantity $P(X>b)$ is the clipped tail mass. Based on $500$ Monte Carlo replications.}
    \label{tab:dp_sim_hist_p25_b15}
\end{table}

\begin{table}[H]
    \centering\footnotesize
    \begin{adjustbox}{max width=0.95\textwidth}
        \begin{tabular}{c c c c r cc ccc}
\hline
$n$ & $[0,b]$ & $P(X>b)$ & $\varepsilon$ & $m^*$ & $\E W_1(\mathbb F_n,F^c)$ & $\E W_1(\mathbb F_{n,m},F^c)$ & $\E W_1(\widetilde{\mathbb F}_{n,m}^{\,\varepsilon},\mathbb F_{n,m})$ & $mB/(n\varepsilon)$ & $\E W_1(\widetilde{\mathbb F}_{n,m}^{\,\varepsilon},F^c)$ \\
\hline
200 & $[0,26]$ & 0.005 & 0.5 & 9 & 0.229 & 0.586 & 1.057 & 2.340 & 1.423 \\
200 & $[0,26]$ & 0.005 & 1.0 & 14 & 0.229 & 0.477 & 0.765 & 1.820 & 1.074 \\
200 & $[0,26]$ & 0.005 & 2.0 & 22 & 0.229 & 0.389 & 0.616 & 1.430 & 0.880 \\
200 & $[0,170]$ & 0.000 & 0.5 & 9 & 0.466 & 0.857 & 7.339 & 15.300 & 7.606 \\
200 & $[0,170]$ & 0.000 & 1.0 & 14 & 0.466 & 0.746 & 5.380 & 11.900 & 5.652 \\
200 & $[0,170]$ & 0.000 & 2.0 & 22 & 0.466 & 0.655 & 4.397 & 9.350 & 4.681 \\
\hline
1000 & $[0,26]$ & 0.005 & 0.5 & 25 & 0.101 & 0.318 & 0.554 & 1.300 & 0.785 \\
1000 & $[0,26]$ & 0.005 & 1.0 & 40 & 0.101 & 0.238 & 0.404 & 1.040 & 0.575 \\
1000 & $[0,26]$ & 0.005 & 2.0 & 63 & 0.101 & 0.181 & 0.301 & 0.819 & 0.428 \\
1000 & $[0,170]$ & 0.000 & 0.5 & 25 & 0.215 & 0.549 & 3.980 & 8.500 & 4.291 \\
1000 & $[0,170]$ & 0.000 & 1.0 & 40 & 0.215 & 0.460 & 3.027 & 6.800 & 3.334 \\
1000 & $[0,170]$ & 0.000 & 2.0 & 63 & 0.215 & 0.390 & 2.361 & 5.355 & 2.643 \\
\hline
\end{tabular}

    \end{adjustbox}
    \caption{Error decomposition for the private EHCDF with $m=m^*(\varepsilon)$ and the zero-inflated $\mathcal{B}'(1,1.5)$ with $p_0=0.25$, for $n=200,1000$ and $c=1.5,10$. The table reports the estimation errors of the ECDF and the noiseless EHCDF, the privacy error, its bound $mB/(n\varepsilon)$ from Theorem~\ref{Theo: privacy}, and the total error.}
    \label{tab:dp_sim_p25_b15}
\end{table}

\begin{table}[H]
    \centering\footnotesize
    \begin{adjustbox}{max width=0.95\textwidth}
        \begin{tabular}{c c c c c c cc cc}
\hline
 & & & & & & \multicolumn{2}{c}{default resolution} & \multicolumn{2}{c}{best resolution} \\
\cline{7-8}\cline{9-10}
$n$ & $[0,b]$ & $P(X>b)$ & $\varepsilon$ & $m^*$ & $k^*$ & EHCDF & Histogram & EHCDF ($m^\dagger$) & Histogram ($k^\dagger$) \\
\hline
200 & $[0,14]$ & 0.004 & 0.5 & 9 & 9 & 0.761 & 0.743 & 0.667 (5) & 0.709 (10) \\
200 & $[0,14]$ & 0.004 & 1.0 & 14 & 9 & 0.601 & 0.545 & 0.521 (9) & 0.486 (14) \\
200 & $[0,14]$ & 0.004 & 2.0 & 22 & 9 & 0.460 & 0.432 & 0.414 (10) & 0.318 (16) \\
200 & $[0,27]$ & 0.001 & 0.5 & 9 & 9 & 1.263 & 1.884 & 0.866 (3) & 1.761 (10) \\
200 & $[0,27]$ & 0.001 & 1.0 & 14 & 9 & 1.004 & 1.380 & 0.708 (5) & 1.191 (14) \\
200 & $[0,27]$ & 0.001 & 2.0 & 22 & 9 & 0.786 & 1.105 & 0.556 (8) & 0.780 (22) \\
200 & $[0,45]$ & 0.000 & 0.5 & 9 & 9 & 1.988 & 3.764 & 1.084 (3) & 3.423 (14) \\
200 & $[0,45]$ & 0.000 & 1.0 & 14 & 9 & 1.566 & 2.808 & 0.857 (3) & 2.296 (16) \\
200 & $[0,45]$ & 0.000 & 2.0 & 22 & 9 & 1.224 & 2.242 & 0.694 (5) & 1.506 (22) \\
200 & $[0,90]$ & 0.000 & 0.5 & 9 & 9 & 3.921 & 9.140 & 1.637 (3) & 7.899 (14) \\
200 & $[0,90]$ & 0.000 & 1.0 & 14 & 9 & 3.060 & 7.403 & 1.138 (3) & 5.394 (22) \\
200 & $[0,90]$ & 0.000 & 2.0 & 22 & 9 & 2.377 & 6.010 & 0.868 (3) & 3.512 (22) \\
200 & $[0,180]$ & 0.000 & 0.5 & 9 & 9 & 7.903 & 20.239 & 2.840 (3) & 17.441 (14) \\
200 & $[0,180]$ & 0.000 & 1.0 & 14 & 9 & 6.179 & 17.055 & 1.722 (3) & 12.049 (22) \\
200 & $[0,180]$ & 0.000 & 2.0 & 22 & 9 & 4.801 & 14.399 & 1.129 (3) & 8.125 (22) \\
\hline
1000 & $[0,14]$ & 0.004 & 0.5 & 25 & 11 & 0.413 & 0.314 & 0.367 (16) & 0.249 (16) \\
1000 & $[0,14]$ & 0.004 & 1.0 & 40 & 11 & 0.304 & 0.279 & 0.273 (25) & 0.155 (25) \\
1000 & $[0,14]$ & 0.004 & 2.0 & 63 & 11 & 0.225 & 0.261 & 0.205 (40) & 0.111 (40) \\
1000 & $[0,27]$ & 0.001 & 0.5 & 25 & 11 & 0.725 & 0.845 & 0.510 (10) & 0.669 (16) \\
1000 & $[0,27]$ & 0.001 & 1.0 & 40 & 11 & 0.549 & 0.748 & 0.396 (16) & 0.410 (25) \\
1000 & $[0,27]$ & 0.001 & 2.0 & 63 & 11 & 0.417 & 0.697 & 0.303 (25) & 0.273 (40) \\
1000 & $[0,45]$ & 0.000 & 0.5 & 25 & 11 & 1.141 & 1.720 & 0.623 (5) & 1.322 (25) \\
1000 & $[0,45]$ & 0.000 & 1.0 & 40 & 11 & 0.877 & 1.532 & 0.496 (10) & 0.855 (25) \\
1000 & $[0,45]$ & 0.000 & 2.0 & 63 & 11 & 0.679 & 1.429 & 0.394 (16) & 0.557 (40) \\
1000 & $[0,90]$ & 0.000 & 0.5 & 25 & 11 & 2.204 & 4.293 & 0.792 (3) & 3.103 (25) \\
1000 & $[0,90]$ & 0.000 & 1.0 & 40 & 11 & 1.696 & 3.710 & 0.639 (5) & 2.127 (25) \\
1000 & $[0,90]$ & 0.000 & 2.0 & 63 & 11 & 1.335 & 3.473 & 0.512 (10) & 1.389 (40) \\
1000 & $[0,180]$ & 0.000 & 0.5 & 25 & 11 & 4.432 & 11.329 & 0.985 (3) & 6.979 (25) \\
1000 & $[0,180]$ & 0.000 & 1.0 & 40 & 11 & 3.406 & 9.219 & 0.806 (5) & 4.883 (25) \\
1000 & $[0,180]$ & 0.000 & 2.0 & 63 & 11 & 2.672 & 8.675 & 0.652 (5) & 3.247 (40) \\
\hline
\end{tabular}

    \end{adjustbox}
    \caption{Mean error $\E\,W_1(\cdot,F^c)$ for $\mathcal{B}'(1,2)$ without atom at zero ($p_0=0$) and $n=200,1000$, on the public interval $[0,b]$, where $b=cP_{99}$ and $c\in\{1.5,3,5,10,20\}$. The table reports the private EHCDF at $m^*$ and the perturbed histogram at $k^*$, as well as both mechanisms at their best grid resolutions, given in parentheses. The quantity $P(X>b)$ is the clipped tail mass. Based on $500$ Monte Carlo replications.}
    \label{tab:dp_sim_hist_p00_b20}
\end{table}

\begin{table}[H]
    \centering\footnotesize
    \begin{adjustbox}{max width=0.95\textwidth}
        \begin{tabular}{c c c c r cc ccc}
\hline
$n$ & $[0,b]$ & $P(X>b)$ & $\varepsilon$ & $m^*$ & $\E W_1(\mathbb F_n,F^c)$ & $\E W_1(\mathbb F_{n,m},F^c)$ & $\E W_1(\widetilde{\mathbb F}_{n,m}^{\,\varepsilon},\mathbb F_{n,m})$ & $mB/(n\varepsilon)$ & $\E W_1(\widetilde{\mathbb F}_{n,m}^{\,\varepsilon},F^c)$ \\
\hline
200 & $[0,14]$ & 0.004 & 0.5 & 9 & 0.133 & 0.316 & 0.594 & 1.260 & 0.761 \\
200 & $[0,14]$ & 0.004 & 1.0 & 14 & 0.133 & 0.254 & 0.450 & 0.980 & 0.601 \\
200 & $[0,14]$ & 0.004 & 2.0 & 22 & 0.133 & 0.207 & 0.335 & 0.770 & 0.460 \\
200 & $[0,90]$ & 0.000 & 0.5 & 9 & 0.213 & 0.398 & 3.858 & 8.100 & 3.921 \\
200 & $[0,90]$ & 0.000 & 1.0 & 14 & 0.213 & 0.336 & 2.984 & 6.300 & 3.060 \\
200 & $[0,90]$ & 0.000 & 2.0 & 22 & 0.213 & 0.290 & 2.291 & 4.950 & 2.377 \\
\hline
1000 & $[0,14]$ & 0.004 & 0.5 & 25 & 0.061 & 0.164 & 0.307 & 0.700 & 0.413 \\
1000 & $[0,14]$ & 0.004 & 1.0 & 40 & 0.061 & 0.124 & 0.224 & 0.560 & 0.304 \\
1000 & $[0,14]$ & 0.004 & 2.0 & 63 & 0.061 & 0.097 & 0.165 & 0.441 & 0.225 \\
1000 & $[0,90]$ & 0.000 & 0.5 & 25 & 0.105 & 0.233 & 2.111 & 4.500 & 2.204 \\
1000 & $[0,90]$ & 0.000 & 1.0 & 40 & 0.105 & 0.192 & 1.598 & 3.600 & 1.696 \\
1000 & $[0,90]$ & 0.000 & 2.0 & 63 & 0.105 & 0.162 & 1.240 & 2.835 & 1.335 \\
\hline
\end{tabular}

    \end{adjustbox}
    \caption{Error decomposition for the private EHCDF with $m=m^*(\varepsilon)$ and $\mathcal{B}'(1,2)$ without atom at zero ($p_0=0$), for $n=200,1000$ and $c=1.5,10$. The table reports the estimation errors of the ECDF and the noiseless EHCDF, the privacy error, its bound $mB/(n\varepsilon)$ from Theorem~\ref{Theo: privacy}, and the total error.}
    \label{tab:dp_sim_p00_b20}
\end{table}

\begin{table}[H]
    \centering\footnotesize
    \begin{adjustbox}{max width=0.95\textwidth}
        \begin{tabular}{c c c c c c cc cc}
\hline
 & & & & & & \multicolumn{2}{c}{default resolution} & \multicolumn{2}{c}{best resolution} \\
\cline{7-8}\cline{9-10}
$n$ & $[0,b]$ & $P(X>b)$ & $\varepsilon$ & $m^*$ & $k^*$ & EHCDF & Histogram & EHCDF ($m^\dagger$) & Histogram ($k^\dagger$) \\
\hline
200 & $[0,12]$ & 0.004 & 0.5 & 9 & 9 & 0.631 & 0.721 & 0.556 (5) & 0.688 (10) \\
200 & $[0,12]$ & 0.004 & 1.0 & 14 & 9 & 0.495 & 0.545 & 0.434 (9) & 0.480 (14) \\
200 & $[0,12]$ & 0.004 & 2.0 & 22 & 9 & 0.378 & 0.449 & 0.345 (14) & 0.325 (22) \\
200 & $[0,23]$ & 0.001 & 0.5 & 9 & 9 & 1.049 & 1.691 & 0.723 (3) & 1.607 (10) \\
200 & $[0,23]$ & 0.001 & 1.0 & 14 & 9 & 0.830 & 1.271 & 0.592 (5) & 1.096 (14) \\
200 & $[0,23]$ & 0.001 & 2.0 & 22 & 9 & 0.646 & 1.032 & 0.464 (8) & 0.733 (22) \\
200 & $[0,39]$ & 0.000 & 0.5 & 9 & 9 & 1.691 & 3.399 & 0.911 (3) & 3.098 (14) \\
200 & $[0,39]$ & 0.000 & 1.0 & 14 & 9 & 1.325 & 2.596 & 0.720 (3) & 2.088 (16) \\
200 & $[0,39]$ & 0.000 & 2.0 & 22 & 9 & 1.032 & 2.085 & 0.587 (5) & 1.393 (22) \\
200 & $[0,77]$ & 0.000 & 0.5 & 9 & 9 & 3.318 & 7.925 & 1.371 (3) & 6.868 (14) \\
200 & $[0,77]$ & 0.000 & 1.0 & 14 & 9 & 2.582 & 6.492 & 0.949 (3) & 4.725 (22) \\
200 & $[0,77]$ & 0.000 & 2.0 & 22 & 9 & 1.999 & 5.225 & 0.728 (3) & 3.102 (22) \\
200 & $[0,160]$ & 0.000 & 0.5 & 9 & 9 & 6.989 & 18.046 & 2.477 (3) & 15.659 (14) \\
200 & $[0,160]$ & 0.000 & 1.0 & 14 & 9 & 5.456 & 15.200 & 1.480 (3) & 10.893 (22) \\
200 & $[0,160]$ & 0.000 & 2.0 & 22 & 9 & 4.232 & 12.941 & 0.962 (3) & 7.402 (22) \\
\hline
1000 & $[0,12]$ & 0.004 & 0.5 & 25 & 11 & 0.342 & 0.344 & 0.306 (16) & 0.276 (25) \\
1000 & $[0,12]$ & 0.004 & 1.0 & 40 & 11 & 0.250 & 0.314 & 0.227 (25) & 0.180 (25) \\
1000 & $[0,12]$ & 0.004 & 2.0 & 63 & 11 & 0.184 & 0.298 & 0.170 (40) & 0.127 (40) \\
1000 & $[0,23]$ & 0.001 & 0.5 & 25 & 11 & 0.600 & 0.813 & 0.424 (10) & 0.643 (25) \\
1000 & $[0,23]$ & 0.001 & 1.0 & 40 & 11 & 0.450 & 0.730 & 0.332 (16) & 0.419 (25) \\
1000 & $[0,23]$ & 0.001 & 2.0 & 63 & 11 & 0.340 & 0.686 & 0.253 (25) & 0.286 (40) \\
1000 & $[0,39]$ & 0.000 & 0.5 & 25 & 11 & 0.966 & 1.603 & 0.527 (5) & 1.239 (25) \\
1000 & $[0,39]$ & 0.000 & 1.0 & 40 & 11 & 0.738 & 1.438 & 0.419 (10) & 0.833 (25) \\
1000 & $[0,39]$ & 0.000 & 2.0 & 63 & 11 & 0.568 & 1.346 & 0.333 (16) & 0.557 (40) \\
1000 & $[0,77]$ & 0.000 & 0.5 & 25 & 11 & 1.859 & 3.800 & 0.666 (3) & 2.756 (25) \\
1000 & $[0,77]$ & 0.000 & 1.0 & 40 & 11 & 1.424 & 3.310 & 0.541 (5) & 1.920 (25) \\
1000 & $[0,77]$ & 0.000 & 2.0 & 63 & 11 & 1.118 & 3.068 & 0.432 (10) & 1.282 (40) \\
1000 & $[0,160]$ & 0.000 & 0.5 & 25 & 11 & 3.911 & 10.141 & 0.837 (3) & 6.346 (25) \\
1000 & $[0,160]$ & 0.000 & 1.0 & 40 & 11 & 2.999 & 8.542 & 0.689 (5) & 4.471 (25) \\
1000 & $[0,160]$ & 0.000 & 2.0 & 63 & 11 & 2.348 & 7.920 & 0.557 (5) & 3.011 (40) \\
\hline
\end{tabular}

    \end{adjustbox}
    \caption{Mean error $\E\,W_1(\cdot,F^c)$ for the zero-inflated $\mathcal{B}'(1,2)$ with $p_0=0.25$ and $n=200,1000$, on the public interval $[0,b]$, where $b=cP_{99}$ and $c\in\{1.5,3,5,10,20\}$. The table reports the private EHCDF at $m^*$ and the perturbed histogram at $k^*$, as well as both mechanisms at their best grid resolutions, given in parentheses. The quantity $P(X>b)$ is the clipped tail mass. Based on $500$ Monte Carlo replications.}
    \label{tab:dp_sim_hist_p25_b20}
\end{table}

\begin{table}[H]
    \centering\footnotesize
    \begin{adjustbox}{max width=0.95\textwidth}
        \begin{tabular}{c c c c r cc ccc}
\hline
$n$ & $[0,b]$ & $P(X>b)$ & $\varepsilon$ & $m^*$ & $\E W_1(\mathbb F_n,F^c)$ & $\E W_1(\mathbb F_{n,m},F^c)$ & $\E W_1(\widetilde{\mathbb F}_{n,m}^{\,\varepsilon},\mathbb F_{n,m})$ & $mB/(n\varepsilon)$ & $\E W_1(\widetilde{\mathbb F}_{n,m}^{\,\varepsilon},F^c)$ \\
\hline
200 & $[0,12]$ & 0.004 & 0.5 & 9 & 0.113 & 0.271 & 0.483 & 1.080 & 0.631 \\
200 & $[0,12]$ & 0.004 & 1.0 & 14 & 0.113 & 0.218 & 0.364 & 0.840 & 0.495 \\
200 & $[0,12]$ & 0.004 & 2.0 & 22 & 0.113 & 0.177 & 0.271 & 0.660 & 0.378 \\
200 & $[0,77]$ & 0.000 & 0.5 & 9 & 0.183 & 0.343 & 3.263 & 6.930 & 3.318 \\
200 & $[0,77]$ & 0.000 & 1.0 & 14 & 0.183 & 0.289 & 2.520 & 5.390 & 2.582 \\
200 & $[0,77]$ & 0.000 & 2.0 & 22 & 0.183 & 0.249 & 1.928 & 4.235 & 1.999 \\
\hline
1000 & $[0,12]$ & 0.004 & 0.5 & 25 & 0.051 & 0.142 & 0.247 & 0.600 & 0.342 \\
1000 & $[0,12]$ & 0.004 & 1.0 & 40 & 0.051 & 0.107 & 0.178 & 0.480 & 0.250 \\
1000 & $[0,12]$ & 0.004 & 2.0 & 63 & 0.051 & 0.083 & 0.131 & 0.378 & 0.184 \\
1000 & $[0,77]$ & 0.000 & 0.5 & 25 & 0.090 & 0.203 & 1.773 & 3.850 & 1.859 \\
1000 & $[0,77]$ & 0.000 & 1.0 & 40 & 0.090 & 0.167 & 1.334 & 3.080 & 1.424 \\
1000 & $[0,77]$ & 0.000 & 2.0 & 63 & 0.090 & 0.141 & 1.032 & 2.425 & 1.118 \\
\hline
\end{tabular}

    \end{adjustbox}
    \caption{Error decomposition for the private EHCDF with $m=m^*(\varepsilon)$ and the zero-inflated $\mathcal{B}'(1,2)$ with $p_0=0.25$, for $n=200,1000$ and $c=1.5,10$. The table reports the estimation errors of the ECDF and the noiseless EHCDF, the privacy error, its bound $mB/(n\varepsilon)$ from Theorem~\ref{Theo: privacy}, and the total error.}
    \label{tab:dp_sim_p25_b20}
\end{table}

\begin{table}[H]
    \centering\footnotesize
    \begin{adjustbox}{max width=0.95\textwidth}
        \begin{tabular}{c c c c c c cc cc}
\hline
 & & & & & & \multicolumn{2}{c}{default resolution} & \multicolumn{2}{c}{best resolution} \\
\cline{7-8}\cline{9-10}
$n$ & $[0,b]$ & $P(X>b)$ & $\varepsilon$ & $m^*$ & $k^*$ & EHCDF & Histogram & EHCDF ($m^\dagger$) & Histogram ($k^\dagger$) \\
\hline
200 & $[0,2.3]$ & 0.003 & 0.5 & 9 & 9 & 0.126 & 0.089 & 0.119 (5) & 0.086 (8) \\
200 & $[0,2.3]$ & 0.003 & 1.0 & 14 & 9 & 0.093 & 0.059 & 0.089 (9) & 0.056 (10) \\
200 & $[0,2.3]$ & 0.003 & 2.0 & 22 & 9 & 0.071 & 0.046 & 0.067 (14) & 0.040 (16) \\
200 & $[0,4.6]$ & 0.000 & 0.5 & 9 & 9 & 0.211 & 0.250 & 0.155 (3) & 0.246 (10) \\
200 & $[0,4.6]$ & 0.000 & 1.0 & 14 & 9 & 0.151 & 0.174 & 0.121 (5) & 0.163 (10) \\
200 & $[0,4.6]$ & 0.000 & 2.0 & 22 & 9 & 0.118 & 0.136 & 0.090 (10) & 0.102 (16) \\
200 & $[0,7.6]$ & 0.000 & 0.5 & 9 & 9 & 0.335 & 0.530 & 0.186 (3) & 0.501 (10) \\
200 & $[0,7.6]$ & 0.000 & 1.0 & 14 & 9 & 0.235 & 0.377 & 0.149 (5) & 0.325 (14) \\
200 & $[0,7.6]$ & 0.000 & 2.0 & 22 & 9 & 0.186 & 0.297 & 0.112 (8) & 0.209 (16) \\
200 & $[0,16]$ & 0.000 & 0.5 & 9 & 9 & 0.713 & 1.590 & 0.275 (3) & 1.361 (14) \\
200 & $[0,16]$ & 0.000 & 1.0 & 14 & 9 & 0.501 & 1.246 & 0.198 (3) & 0.869 (16) \\
200 & $[0,16]$ & 0.000 & 2.0 & 22 & 9 & 0.403 & 0.996 & 0.150 (3) & 0.571 (22) \\
200 & $[0,31]$ & 0.000 & 0.5 & 9 & 9 & 1.409 & 3.441 & 0.457 (3) & 3.074 (14) \\
200 & $[0,31]$ & 0.000 & 1.0 & 14 & 9 & 1.000 & 2.882 & 0.291 (3) & 2.061 (22) \\
200 & $[0,31]$ & 0.000 & 2.0 & 22 & 9 & 0.815 & 2.479 & 0.191 (3) & 1.298 (25) \\
\hline
1000 & $[0,2.3]$ & 0.003 & 0.5 & 25 & 11 & 0.061 & 0.031 & 0.058 (16) & 0.027 (16) \\
1000 & $[0,2.3]$ & 0.003 & 1.0 & 40 & 11 & 0.046 & 0.026 & 0.043 (25) & 0.020 (25) \\
1000 & $[0,2.3]$ & 0.003 & 2.0 & 63 & 11 & 0.033 & 0.024 & 0.032 (40) & 0.015 (25) \\
1000 & $[0,4.6]$ & 0.000 & 0.5 & 25 & 11 & 0.104 & 0.100 & 0.080 (10) & 0.085 (16) \\
1000 & $[0,4.6]$ & 0.000 & 1.0 & 40 & 11 & 0.080 & 0.082 & 0.061 (16) & 0.055 (25) \\
1000 & $[0,4.6]$ & 0.000 & 2.0 & 63 & 11 & 0.059 & 0.075 & 0.045 (25) & 0.035 (25) \\
1000 & $[0,7.6]$ & 0.000 & 0.5 & 25 & 11 & 0.164 & 0.225 & 0.099 (10) & 0.181 (16) \\
1000 & $[0,7.6]$ & 0.000 & 1.0 & 40 & 11 & 0.128 & 0.191 & 0.075 (10) & 0.114 (25) \\
1000 & $[0,7.6]$ & 0.000 & 2.0 & 63 & 11 & 0.097 & 0.175 & 0.056 (25) & 0.076 (25) \\
1000 & $[0,16]$ & 0.000 & 0.5 & 25 & 11 & 0.356 & 0.682 & 0.137 (5) & 0.499 (25) \\
1000 & $[0,16]$ & 0.000 & 1.0 & 40 & 11 & 0.280 & 0.576 & 0.105 (10) & 0.319 (25) \\
1000 & $[0,16]$ & 0.000 & 2.0 & 63 & 11 & 0.213 & 0.536 & 0.076 (10) & 0.208 (40) \\
1000 & $[0,31]$ & 0.000 & 0.5 & 25 & 11 & 0.720 & 1.950 & 0.177 (3) & 1.115 (25) \\
1000 & $[0,31]$ & 0.000 & 1.0 & 40 & 11 & 0.571 & 1.760 & 0.139 (3) & 0.753 (25) \\
1000 & $[0,31]$ & 0.000 & 2.0 & 63 & 11 & 0.439 & 1.605 & 0.102 (10) & 0.485 (40) \\
\hline
\end{tabular}

    \end{adjustbox}
    \caption{Mean error $\E\,W_1(\cdot,F^c)$ for $\mathcal{B}'(1,5)$ without atom at zero ($p_0=0$) and $n=200,1000$, on the public interval $[0,b]$, where $b=cP_{99}$ and $c\in\{1.5,3,5,10,20\}$. The table reports the private EHCDF at $m^*$ and the perturbed histogram at $k^*$, as well as both mechanisms at their best grid resolutions, given in parentheses. The quantity $P(X>b)$ is the clipped tail mass. Based on $500$ Monte Carlo replications.}
    \label{tab:dp_sim_hist_p00_b50}
\end{table}

\begin{table}[H]
    \centering\footnotesize
    \begin{adjustbox}{max width=0.95\textwidth}
        \begin{tabular}{c c c c r cc ccc}
\hline
$n$ & $[0,b]$ & $P(X>b)$ & $\varepsilon$ & $m^*$ & $\E W_1(\mathbb F_n,F^c)$ & $\E W_1(\mathbb F_{n,m},F^c)$ & $\E W_1(\widetilde{\mathbb F}_{n,m}^{\,\varepsilon},\mathbb F_{n,m})$ & $mB/(n\varepsilon)$ & $\E W_1(\widetilde{\mathbb F}_{n,m}^{\,\varepsilon},F^c)$ \\
\hline
200 & $[0,2.3]$ & 0.003 & 0.5 & 9 & 0.026 & 0.054 & 0.105 & 0.207 & 0.126 \\
200 & $[0,2.3]$ & 0.003 & 1.0 & 14 & 0.026 & 0.042 & 0.075 & 0.161 & 0.093 \\
200 & $[0,2.3]$ & 0.003 & 2.0 & 22 & 0.026 & 0.034 & 0.055 & 0.127 & 0.071 \\
200 & $[0,16]$ & 0.000 & 0.5 & 9 & 0.029 & 0.056 & 0.720 & 1.440 & 0.713 \\
200 & $[0,16]$ & 0.000 & 1.0 & 14 & 0.029 & 0.045 & 0.505 & 1.120 & 0.501 \\
200 & $[0,16]$ & 0.000 & 2.0 & 22 & 0.029 & 0.037 & 0.404 & 0.880 & 0.403 \\
\hline
1000 & $[0,2.3]$ & 0.003 & 0.5 & 25 & 0.012 & 0.025 & 0.048 & 0.115 & 0.061 \\
1000 & $[0,2.3]$ & 0.003 & 1.0 & 40 & 0.012 & 0.019 & 0.036 & 0.092 & 0.046 \\
1000 & $[0,2.3]$ & 0.003 & 2.0 & 63 & 0.012 & 0.016 & 0.026 & 0.072 & 0.033 \\
1000 & $[0,16]$ & 0.000 & 0.5 & 25 & 0.014 & 0.028 & 0.356 & 0.800 & 0.356 \\
1000 & $[0,16]$ & 0.000 & 1.0 & 40 & 0.014 & 0.022 & 0.278 & 0.640 & 0.280 \\
1000 & $[0,16]$ & 0.000 & 2.0 & 63 & 0.014 & 0.018 & 0.211 & 0.504 & 0.213 \\
\hline
\end{tabular}

    \end{adjustbox}
    \caption{Error decomposition for the private EHCDF with $m=m^*(\varepsilon)$ and $\mathcal{B}'(1,5)$ without atom at zero ($p_0=0$), for $n=200,1000$ and $c=1.5,10$. The table reports the estimation errors of the ECDF and the noiseless EHCDF, the privacy error, its bound $mB/(n\varepsilon)$ from Theorem~\ref{Theo: privacy}, and the total error.}
    \label{tab:dp_sim_p00_b50}
\end{table}

\clearpage
\end{document}